\documentclass[11pt]{article} 
\usepackage{graphicx,subcaption}
\graphicspath{{figure_bounded/}}
\usepackage{hyperref}
\usepackage[dvipsnames]{xcolor}
\usepackage{amssymb,amsthm,amsmath}
\usepackage{float,longtable}
\usepackage[numbers]{natbib}
\usepackage{mathtools}

\usepackage[left=1in,top=1in,right=1in,bottom=1in]{geometry}
\newcommand{\blind}{1}

\DeclarePairedDelimiter{\ceil}{\lceil}{\rceil}

\DeclareMathOperator{\argmin}{\arg\min}

\DeclareMathOperator{\Dirac}{Dirac}
\DeclareMathOperator{\expit}{expit}
\DeclareMathOperator{\logit}{logit}

\newcommand{\As}[2]{\textbf{B#1}(#2)}

\newcommand{\Cdot}{\raisebox{-0.3ex}{\scalebox{1.2}{$\cdot$}}}

\newcommand{\CV}{\mathrm{CV}}
\newcommand{\bbG}{\mathbb{G}}
\newcommand{\bbH}{\mathbb{H}}
\newcommand{\bbK}{\mathbb{K}}
\newcommand{\bG}{G}
\newcommand{\bbP}{\mathbb{P}}
\newcommand{\bq}{\bar{q}}
\newcommand{\bQ}{\bar{Q}}
\newcommand{\bbQ}{\mathbb{Q}}
\newcommand{\bbbQ}{\bar{\bbQ}}
\newcommand{\floor}[1]{\lfloor #1 \rfloor}
\newcommand{\hnk}{h_{n,\kappa_{n}}}
\newcommand{\IF}{\mathrm{IF}}
\newcommand{\KL}{\mathrm{KL}}
\newcommand{\Rem}{\mathrm{Rem}}
\renewcommand{\th}{\tilde{h}}
\newcommand{\Var}{\mathrm{Var}}
\newcommand{\xC}{\mathcal{C}}
\newcommand{\xG}{\mathcal{G}}
\newcommand{\xH}{\mathcal{H}}
\newcommand{\xM}{\mathcal{M}}
\newcommand{\xO}{\mathcal{O}}
\newcommand{\xQ}{\mathcal{Q}}
\newcommand{\xR}{\mathbb{R}}
\newcommand{\xS}{\mathcal{S}}
\newcommand{\xT}{\mathcal{T}}

\makeatletter
\renewcommand\subsubsection{%
  \@startsection{subsubsection}{3}{\z@}%
  {-3.25ex\@plus -1ex \@minus -.2ex}%
  {-1.5ex \@plus .2ex}%
  {\normalfont\normalsize\bfseries}%
}
\makeatother

\newtheorem{theorem}{Theorem}
\newtheorem{lemma}[theorem]{Lemma}
\newtheorem{corollary}[theorem]{Corollary}

\allowdisplaybreaks

\title{Collaborative targeted minimum loss  inference from continuously indexed nuisance parameter estimators
}

\author{Cheng Ju$^1$, Antoine Chambaz$^2$, Mark J. van der Laan$^1$\\[1em] 
  $^1$ Division of Biostatistics, UC Berkeley\\
  $^2$ MAP5 (UMR 8145), Universit\'e Paris Descartes}

\begin{document}
\def\spacingset#1{\renewcommand{\baselinestretch}%
  {#1}\small\normalsize} \spacingset{1}%
\if1\blind {%
  \title{\textbf{Collaborative targeted 
      inference  from  continuously  indexed nuisance  parameter  estimators}}
  \author{Cheng Ju
    \hspace{.2cm}\\
    Division of Biostatistics, UC Berkeley\\
    and \\
    Antoine Chambaz\\
    MAP5 (UMR 8145), Universit\'e Paris Descartes\\
    and \\
    Mark J. van der Laan\\
    Division of Biostatistics, UC Berkeley }
  \maketitle
} \fi \if0\blind { \bigskip \bigskip \bigskip
  \begin{center}
    \LARGE{\textbf{Collaborative 
        inference from continuously indexed nuisance parameter estimators}}
  \end{center}
  \medskip } \fi

\bigskip
\begin{abstract}
  Suppose that we wish to infer the  value of a statistical parameter at a law
  from which we sample independent  observations.  Suppose that this parameter
  is   smooth   and   that    we   can   define   two   variation-independent,
  infinite-dimensional  features   of  the  law,   its  so  called   $Q$-  and
  $G$-components (comp.), such that if we estimate them consistently at a fast
  enough product of rates, then we can build a confidence interval (CI) with a
  given  asymptotic level  based on  a plain  targeted minimum  loss estimator
  (TMLE).  The  estimators of the  $Q$- and  $G$-comp.\ would typically  be by
  products of machine learning algorithms.  We focus on the case
  that 
  the  machine  learning algorithm  for  the  $G$-comp.\  is fine-tuned  by  a
  real-valued parameter $h$. 
  Then, a  plain TMLE with an  $h$ chosen by cross-validation  would typically
  not lend itself  to the construction of  a CI, because the  selection of $h$
  would  trade-off its  empirical bias  with something  akin to  the empirical
  variance of the estimator of the $G$-comp.\  as opposed to that of the TMLE.
  A  collaborative TMLE  (C-TMLE)  might, however,  succeed  in achieving  the
  relevant trade-off.  We prove that this is the case indeed.

  We construct  a C-TMLE and  show that, under high-level  empirical processes
  conditions, and if  there exists an oracle $h$ that  makes a bulky remainder
  term  asymptotically Gaussian,  then the  C-TMLE is  asymptotically Gaussian
  hence amenable to building a CI provided that its asymptotic variance can be
  estimated too.  The construction hinges  on guaranteeing that an additional,
  well chosen estimating equation is solved  on top of the estimating equation
  that a plain  TMLE solves. The optimal $h$ is  chosen by cross-validating an
  empirical criterion  that guarantees the wished  trade-off between empirical
  bias and variance.

  We illustrate the construction and main  result with the inference of the so
  called average treatment effect, where the $Q$-comp.\ consists in a marginal
  law and a conditional expectation, and  the $G$-comp.\ is a propensity score
  (a  conditional probability).   We also  conduct a  multi-faceted simulation
  study to investigate the empirical properties of the collaborative TMLE when
  the $G$-comp.\  is estimated by  the LASSO.  Here, $h$  is the bound  on the
  $\ell^{1}$-norm  of the  candidate coefficients.   The variety  of scenarios
  shed light  on small and  moderate sample properties,  in the face  of low-,
  moderate- or  high-dimensional baseline covariates, and  possibly positivity
  violation.
\end{abstract}

\noindent%
\vfill

\newpage

\section{Introduction}
\label{sec:intro}

We wish to infer  the value of a statistical parameter at a  law from which we
sample independent  observations.  The parameter  is a smooth function  of the
data distribution.   We assume that  we can define  two variation-independent,
infinite-dimensional   features  of   the  law,   its  so   called  $Q$-   and
$G$-components, such  that if we estimate  them consistently at a  fast enough
joint  rate,  then we  can  build  a confidence  interval  (CI)  with a  given
asymptotic  level   based  on   a  plain   targeted  minimum   loss  estimator
(TMLE)~\citep{van2006targeted,van2011targeted}.    Typically,   the   parameter
depends  on the  law only  through  its $Q$-component,  whereas its  canonical
gradient depends  on the law  through both  its $Q$- and  $G$-components.  The
estimators of  the $Q$- and $G$-components  would typically be by  products of
machine learning algorithms.  We focus on the case that
the  machine  learning  algorithm  for  the  $G$-component  is  fine-tuned  by  a
real-valued parameter $h$.  Is it possible to construct an estimator that will
lend itself to the construction of a CI, by fine-tuning data-adaptively and in
a targeted fashion both the algorithm for the estimation of the $G$-component and
the resulting estimator of the parameter of interest?

\subsubsection*{Literature overview.}
\label{subsec:literature}

The  general problem  that we  address is  often encountered  in observational
studies of the  effect of an exposure,  for instance when one  wishes to infer
the average  effect of a two-level  exposure. It is then  necessary to account
for  the fact  that the  level  of exposure  is  not fully  randomized in  the
observed population.   A pivotal object  of interest  in such studies,  the so
called  exposure mechanism  (that is,  the conditional  law of  exposure given
baseline covariates) is  an example of what we generally  call a $G$-component
of the law of the experiment.

A  wide range  of estimators  of the  average effect  of a  two-level exposure
require the  estimation of  the propensity score:  Horvitz-Thompson estimators
\citep{horvitz1952generalization};   estimators  based   on  propensity   score
matching      \citep{rosenbaum1983central,ho2007matching,ho2006matchit}      or
stratification  \citep{cochran1968effectiveness,   rosenbaum1984reducing};  any
estimator relying on the efficient  influence curve, among which double-robust
inverse   probability  of   exposure   weighted  estimators   \citep{Robins00a,
  Robins00b, Robins00c} or estimators built based on the targeted minimum loss
estimation (TMLE) methodology~\citep{van2006targeted,van2011targeted}.

Common methods  for the  estimation of the  propensity score  are multivariate
logistic   regression~\citep{kurth2006results},  high-dimensional   propensity
score  adjustment~\citep{schneeweiss2009high,franklin2015regularized},  and  a
variety                  of                  machine                  learning
algorithms~\citep{lee2010improving,gruber2015ensemble,
  ju2016propensity}. Except in the so called \textit{collaborative} variant of
TMLE that we will discuss shortly,  the estimators of the propensity score can
be  derived at  a preliminary  step, regardless  essentially of  why they  are
needed  and how  they are  used at  the subsequent  step. This  is problematic
because  optimality at  the preliminary  step has  little if  any relation  to
optimality at the subsequent step.  For instance, the optimal estimator of the
propensity score at the preliminary step might take values very close to zero,
therefore disqualifying it  as a viable estimator at the  subsequent step, not
to mention an optimal one.  In a less dramatic scenario, using an instrumental
variable (which only influences exposure but  not the outcome) to estimate the
propensity score could concomitantly yield a better estimator thereof and only
increase  the   variance  of  the   resulting  estimator  of  the   effect  of
exposure~\citep{van2010collaborative, van2011targeted}.

This prompted the development of  the so called \textit{collaborative} version
of         the         targeted        minimum         loss         estimation
methodology~\citep{van2010collaborative,    van2011targeted},     where    the
estimation of  the $G$-component is not  separated from that of  the parameter of
main interest  anymore. More concretely, collaborative  TMLE (C-TMLE) consists
in building a sequence of estimators of the $G$-component and in selecting one of
them by  optimizing a criterion that  targets the parameter of  main interest.
For  instance,  in the  above  less  dramatic  scenario, covariates  that  are
strongly  predictive of  exposure but  not of  the outcome  would be  removed,
resulting  in  less   bias  for  the  estimator  of  the   parameter  of  main
interest. 

The C-TMLE methodology  has been adapted to a wide  range of fields, including
genomics~\citep{gruber2010application,  wang2011finding},   survival  analysis
\citep{stitelman2010collaborative}, and clinical studies\citep{ju2016scalable}.
Because  the   derivation  of  C-TMLE  estimators   is  often  computationally
demanding, scalable  versions have also  been developed~\citep{ju2016scalable}.

In~\citep{schnitzer2017collaborative}, the authors  propose a C-TMLE algorithm
that uses regression shrinkage of the exposure model for the estimation of the
propensity score.  It  sequentially reduces the parameter  that determines the
amount of  penalty placed on the  size of the coefficient  values, and selects
the   appropriate  parameter   by  cross-validation.    The  methodology   for
continuously  fine-tuned, collaborative  targeted learning  that we  develop in
this           article           encompasses           the           algorithm
of~\citep{schnitzer2017collaborative}. Its statistical analysis sheds light on
why,  and  under  which  assumptions,   it  would  provide  valid  statistical
inference.

The present  study builds  upon~\citep{chapterTLBII}.  The methodology  is also
studied in~\citep{ju2017adaptive,ju2017collaborative}, the latter an example of
real-life application.

At this point in the introduction, we wish to formalize what is the problem at
stake.  What follows  recasts the  introductory paragraph  in the  theoretical
framework that we adopt in the article.

\subsubsection*{Setting the scene.}
\label{subsec:scene}

Let $O_{1}, \ldots,  O_{n}$ be $n$ independent  draws from a law  $P_{0}$ on a
set $\xO$.  We  view $P_{0}$ as an  element of the statistical  model $\xM$, a
collection of  plausible laws for  $O_{1}, \ldots,  O_{n}$.  The more  we know
about $P_{0}$, the smaller  is $\xM$.  Our primary goal is  to infer the value
of    parameter    $\Psi   :    \xM    \to    \xR$   at    $P_{0}$,    namely,
$\psi_{0} \equiv \Psi(P_{0})$.  Our statistical  analysis is asymptotic in the
number of observations.

We  consider  the  case  that  $\Psi$  is  pathwise  differentiable  at  every
$P\in     \xM$    with     respect    to     (w.r.t.)     a     tangent    set
$\xS_{P} \subset  L_{0}^{2} (P)$: there  exists $D^{*} (P) \in  L_{0}^{2} (P)$
such   that,   for   every   $s   \in  S_{P}$,   there   exists   a   submodel
$\{P_{t}  :  t  \in  \xR,  |t| <  c\}  \subset  \xM$  satisfying  \textit{(i)}
$P_{t}|_{t=0}  = P$,  \textit{(ii)}  $P_{t} \ll  P$ for  all  $t \in  ]-c,c[$,
\textit{(iii)}
\begin{equation*}
  \left.\frac{d}{dt} \log \frac{dP_{t}}{dP} (O)\right|_{t=0} = s(O)
\end{equation*}
(the submodel's score function equals  $s$), and \textit{(iv)} the real valued
mapping $t \mapsto  \Psi(P_{t})$ is differentiable at $t=0$  with a derivative
equal to $P D^{*}(P) s$, where $Pf$ is a shorthand notation for $E_{P} (f(O))$
(any  measurable $f$).   It is  assumed  moreover that  every $P  \in \xM$  is
associated with  two possibly  infinite-dimensional features  $Q \in  \xQ$ and
$\bG  \in   \xG$  such   that  \textit{(i)}  $Q$   and  $\bG$   are  unrelated
(\textit{i.e.},  variation  independent:  knowing  anything  about  $Q$  tells
nothing about  $\bG$ and vice  versa), \textit{(ii)} $\Psi(P)$ depends  on $P$
only through $Q$,  \textit{(iii)} $ D^{*}(P)$ depends on $P$  only through $Q$
and $\bG$, and \textit{(iv)} $\bG$ is a  mapping from $\xO$ to $\xR$.  At this
early stage, we can introduce the pivotal
\begin{equation*}
  \label{eq:remainder:intro}
  \Rem_{20} (Q, \bG) \equiv \Psi(P) - \Psi(P_{0}) + P_{0} D^{*}(P)
\end{equation*}
for every $P\in\xM$.   The notation is justified \textit{(i)}  because we wish
to  think  of  the  right-hand-side   expression  as  a  remainder  term,  and
\textit{(ii)} by  the fact that  $\Psi(P)$ and  $D^{*}(P)$ depend on  $P$ only
through $Q$  and $\bG$.  We  consider the case  that parameter $\Psi$  is such
that, for some pseudo-distances $d_{\xQ}$ and $d_{\xG}$ on $\xQ$ and $\xG$,
\begin{equation}
  \label{eq:remainder:bound:intro}
  |\Rem_{20}  (Q,   \bG)|  \lesssim  d_{\xQ}(Q,  Q_{0})   \times  d_{\xG}(\bG,
  \bG_{0}), 
\end{equation}
where $a  \lesssim b$ stand for  ``there exists a universal  positive constant
$c>0$   such   that    $a   \leq   bc$''.   A    remainder   term   satisfying
\eqref{eq:remainder:bound:intro} is said double-robust. 

Let $\hat{Q}$ be an algorithm for the estimation of $Q_{0}$, the $Q$-component of
the true  law $P_{0}$.  Likewise,  let $\hat{\bG}_{h}$  ($h \in \xH$,  an open
interval of $\xR_{+}^{*}$ of which the  closure contains 0) be an $h$-specific
algorithm  for  the estimation  of  $\bG_{0}$,  the $\bG$-component  of  $P_{0}$.
Formally, we view $\hat{Q}$ and each $\hat{\bG}_{h}$ as mappings from
\begin{equation*}
  \bigcup_{N \geq 1} \left\{N^{-1}\sum_{i=1}^{N}  \Dirac(o_{i}) : o_{1}, \ldots, o_{N}
    \in \xO\right\} 
\end{equation*}
to  $\xQ$ and  $\xG$,  respectively,  that can  ``learn''  from the  empirical
measure $P_{n}$ some estimators $\hat{Q}(P_{n})$ and $\hat{\bG}_{h}(P_{n})$ of
$Q_{0}$ and $\bG_{0}$.  Set $Q_{n}^{0} \equiv \hat{Q}(P_{n})$ (the superscript
0 stands  for ``initial''),  $\bG_{n,h} \equiv \hat{\bG}_{h}(P_{n})$,  and let
$P_{n,h}^{0}  \in \xM$  be any  element of  the model  of which  the $Q$-  and
$\bG$-components  equal  $Q_{n}^{0}$  and   $\bG_{n,h}$.   Derived  by  the  mere
substitution of $P_{n}^{0}$ for  $P_{0}$ in $\Psi(P_{0})$, $\Psi(P_{n,h}^{0})$
is a natural  estimator of $\Psi(P_{0})$.  It is  \textit{not} targeted toward
the inference of $\Psi(P_{0})$ in the sense that none of the known features of
$P_{n,h}^{0}$ was derived  specifically for the sake  of ultimately estimating
$\Psi(P_{0})$.

It  is  well  documented  in  the  TMLE literature  that  one  way  to  target
$\Psi(P_{n,h}^{0})$  toward $\Psi(P_{0})$  is to  build $P_{n,h}^{*}  \in \xM$
from $P_{n,h}^{0}$ in such a way that
\begin{equation*}
  P_{n} D^{*} (P_{n,h}^{*}) = o_{P} (1/\sqrt{n})
\end{equation*}
and to infer $\Psi(P_{0})$ with  $\Psi(P_{n,h}^{*})$. This can be achieved, in
such a way that $\bG_{n,h}$ is not modified, by ``fluctuating'' $P_{n,h}^{0}$,
a procedure that we will develop in details in the specific example studied in
the article.  Then, by~\eqref{eq:remainder:intro}, the estimator satisfies the
asymptotic expansion:
\begin{equation}
  \label{eq:asym:exp:intro}
  \Psi(P_{n,h}^{*})  -  \Psi(P_{0})  =  (P_{n}  -  P_{0})  D^{*}  (P_{n,h}^{*})  +
  \Rem_{20} (Q_{n,h}^{*}, \bG_{n,h}) + o_{P} (1/\sqrt{n}).  
\end{equation}

By convention, we agree that small values of $h$ correspond with less bias for
$\bG_{n,h}$ as an estimator of $\bG_0$.  Moreover, we assume that there exists
$h_{n}      \in      \xH$,      $h_{n}       =      o(1)$,      such      that
$d_{\xG}(\bG_{n,h_{n}},    \bG_{0})   =    o_{P}   (\rho_{1,n})$    for   some
$\rho_{1,n}   =  o(1)$,   \textit{i.e.},  that   $\bG_{n,h_{n}}$  consistently
estimates $\bG_{0}$ at  rate $\rho_{1,n}$.  If $Q_{n,h_{n}}^{*}$  is also such
that   $d_{\xQ}(Q_{n,h_{n}}^{*},   Q_{0})   =  o_{P}(\rho_{2,n})$   for   some
$\rho_{2,n}  = o(1)$,  and if  $\rho_{1,n} \rho_{2,n}  = o(1/\sqrt{n})$,  then
\eqref{eq:remainder:bound:intro} and \eqref{eq:asym:exp:intro} yield
\begin{equation*}
  \Psi(P_{n,h_{n}}^{*})   -    \Psi(P_{0})   =    (P_{n}   -    P_{0})   D^{*}
  (P_{n,h_{n}}^{*})  + o_{P} (1/\sqrt{n})
\end{equation*}
which may in turn imply the asymptotic linear expansion
\begin{equation}
  \label{eq:asymp:exp:intro}
  \Psi(P_{n,h_{n}}^{*}) -  \Psi(P_{0}) = (P_{n}  - P_{0}) \IF  + o_{P}
  (1/\sqrt{n}), 
\end{equation}
with influence function $\IF \equiv D^{*} (P_{0})$, depending in particular on
how data-adaptive are algorithms $\hat{Q}$  and $\hat{\bG}_{h}$ ($h \in \xH$).
By  the  central  limit theorem,  \eqref{eq:asymp:exp:intro}  guarantees  that
$\sqrt{n} (\Psi(P_{n,h_{n}}^{*}) - \Psi(P_{0}))$ is asymptotically Gaussian.

\textit{We    focus     on    a    more    challenging     situation,    where
  $\rho_{1,n} \rho_{2,n}$ is not  necessarily $o(1/\sqrt{n})$.}  We anticipate
that our  analysis is also  very relevant at  small and moderate  sample sizes
when $\rho_{1,n} \rho_{2,n} = o(1/\sqrt{n})$.

In   order   to   derive   an   asymptotic   linear   expansion   similar   to
\eqref{eq:asymp:exp:intro} from  \eqref{eq:asym:exp:intro} in  this situation,
we     would    have     to    derive     an    asymptotic     expansion    of
$\Rem_{20} (Q_{n,h_{n}}^{*},  \bG_{n,h_{n}})$. Unfortunately, we  have reasons
to believe that this is not  possible without targeting (their presentation in
an example is deferred to Section~\ref{subsec:select:uncoop}).

Now,  observe that  the estimators  $\Psi(P_{n,h}^{*})$ $(h  \in \xH)$  do not
cooperate in  the sense that,  although $Q_{n,h}^{*}$ and  $Q_{n,h'}^{*}$ (for
any  two $h,  h' \in  \xH$,  $h \neq  h'$)  share the  same initial  estimator
$Q_{n}^{0}$, the construction of the latter does not capitalize on that of the
former.   In contrast,  we propose  to  build collaboratively  a continuum  of
estimators  of  the form  $\Psi(P_{n,h}^{*})$  ($h  \in  \xH$) and  to  select
data-adaptively one  among them  that will  be asymptotically  Gaussian, under
conditions often encountered in empirical process theory.

\subsubsection*{Organization of the article.}
\label{subsec:org}

In  Section~\ref{sec:general},  we  lay   out  a  high-level  presentation  of
collaborative    TMLE,     and    state     a    high-level     result.     In
Sections~\ref{sec:ctmle_con},  \ref{sec:software},  \ref{sec:experiments}  and
\ref{sec:transfer},     we     consider     a    specific     example.      In
Section~\ref{sec:ctmle_con}, we particularize the theoretical construction and
analysis.    In   Section~\ref{sec:software},   we  describe   two   practical
instantiations of  the estimator developed in  Section~\ref{sec:ctmle_con}. In
Sections~\ref{sec:experiments}   and  \ref{sec:transfer},   we  carry   out  a
mutli-faceted  simulation study  of their  performances and  comment upon  its
results.   In  Section~\ref{sec:discuss},  we  summarize the  content  of  the
article. All the proofs are gathered in the appendix.

\section{High-level presentation and result} 
\label{sec:general}

We  now  state and  prove  a  general  result about  continuously  fine-tuned,
collaborative    targeted   minimum    loss   estimation,    a   version    of
\citep[Theorem~10.1  in][]{chapterTLBII}.    Its  high-level   assumptions  are
clarified in the particular example that we study in the next sections.

From now on, we slightly abuse  notation and denote $D^{*}(Q, \bG)$ instead of
$D^{*} (P)$,  where $Q$ and  $\bG$ are the  $Q$- and $\bG$-components  of $P$.
Let  $G_{\Cdot}   \equiv  \{\bG_{t}  :   t  \in   \xT\}  \subset  \xG$   be  a
(one-dimensional) subset of  $\xG$ (indexed by a real parameter  ranging in an
open subset  $\xT$ of $\xH$)  such that $t  \mapsto D^{*} (Q,  \bG_{t})(O)$ is
twice differentiable over  $\xT$ for all $Q \in  \xQ$ ($P_{0}$-almost surely).
We  characterize $\partial  D^{*}$ and  $\partial^{2} D^{*}$  by setting,  for
every $h \in \xT$ and $Q \in \xQ$,
\begin{eqnarray}
  \label{eq:dD:general}
  \partial_{h} D^{*} (Q, G_{\Cdot})(O) 
  &\equiv& \frac{d}{dt} D^{*} (Q, \bG_{t})(O)|_{t=h}, \\
  \notag 
  \partial_{h}^{2} D^{*} (Q, G_{\Cdot})(O) 
  &\equiv& \frac{d^{2}}{dt^{2}} D^{*} (Q, \bG_{t})(O)|_{t=h}.
\end{eqnarray}

Consider the following  inter-dependent assumptions. The first  one is indexed
by $(Q, h, c)\in \xQ \times \xH \times \xR_{+}^{*}$.
\begin{description}
\item[\textbf{A1}$(Q,h,c)$]    There     exists    an     open    neighborhood
  $\xT    \subset   \xH$    of   $h    \in    \xH$   for    which   the    set
  $\bG_{n,\Cdot} \equiv \{\hat{\bG}_{h}(P_{n}) \equiv  \bG_{n,h} : h \in \xT\}
  \subset \xG$ is  such that $t \mapsto D^{*} (Q,  \bG_{n,\Cdot})(O)$ is twice
  differentiable over $\xT$ ($P_{0}$-almost surely).  Moreover, $P_{0}$-almost
  surely,
  \begin{equation*}
    \sup_{h \in \xT} |\partial_{h}^{2} D^{*} (Q, \bG_{n,\Cdot}) (O)| \leq c.  
  \end{equation*}
\item[\textbf{A2}]   For   all  $h   \in   \xH$,   we   know  how   to   build
  $P_{n,h}^{*}   \in  \xM$,   with  $Q$-   and  $\bG$-components   denoted  by
  $Q_{n,h}^{*}$     and     $\bG_{n,h}$,     in    such     a     way     that
  $P_{n} D^{*}  (Q_{n,h}^{*}, \bG_{n,h}) = o_{P}  (1/\sqrt{n})$.  Moreover, we
  know how to choose $h_{n} \in \xH$ such that
  \begin{equation}
    \label{eq:hla2:one}
    P_{n} D^{*} (Q_{n,h_{n}}^{*}, \bG_{n,h_{n}}) 
    = o_{P} (1/\sqrt{n}).
  \end{equation}
  and,       for       some       deterministic       $c_{2}       >       0$,
  \textbf{A1}$(Q_{n,h_{n}}^{*}, h_{n}, c_{2})$ is met and
  \begin{equation}
    \label{eq:hla2:two}
    P_{n} \partial_{h_{n}} D^{*} (Q_{n,h_{n}}^{*}, \bG_{n,\Cdot}) 
    = o_{P}     (1/n^{1/4}). 
  \end{equation}
\item[\textbf{A3}]                It                 holds                that
  $d_{\xG}(\bG_{n,h_{n}},   \bG_{0})   =   o_{P}  (1)$,   and   there   exists
  $Q_{1} \in \xQ$ such that $d_{\xQ}(Q_{n,h_{n}}^{*}, Q_{1}) = o_{P} (1)$.  In
  addition,
  \begin{eqnarray}
    \label{eq:hla3:one}
    (P_{n} - P_{0}) \left(D^{*}(Q_{n,h_{n}}^{*}, \bG_{n,h_{n}}) - D^{*}(Q_{1},
    \bG_{0})\right) 
    &=& o_{P}(1/\sqrt{n}),\\
    \label{eq:hla3:three}
    \Rem_{20} (Q_{n,h_{n}}^{*}, \bG_{n,h_{n}}) - \Rem_{20} (Q_{1}, \bG_{n,h_{n}}) 
    &=& o_{P} (1/\sqrt{n}). 
  \end{eqnarray} 
\item[\textbf{A4}]   Let   $\Phi_{0}   :   \xG    \to   \xR$   be   given   by
  $\Phi_{0}   (\bG)   \equiv  P_{0}   D^{*}   (Q_{1},   \bG)$.   There   exist
  $\th_{n} \in \xH$ and $\Delta(P_{1}) \in L_{0}^{2} (P_{0})$ such that
  \begin{equation}
    \label{eq:hla4}
    \Phi_{0}(\bG_{n,\th_{n}})   -   \Phi_{0}(\bG_{0})   =  (P_{n}   -   P_{0})
    \Delta(P_{1}) + o_{P} (1/\sqrt{n}).
  \end{equation}
\item[\textbf{A5}]   It   holds   that   $(h_{n}  -   \th_{n})^{2}   =   o_{P}
  (1/\sqrt{n})$. Moreover, there exists a  deterministic $c_{5} > 0$ such that
  \textbf{A1}$(Q_{1}, h_{n}, c_{5})$ is met, and
  \begin{eqnarray}
    \label{eq:hla5:one}
    (P_{n} - P_{0}) \left(D^{*}(Q_{1}, \bG_{n,h_{n}}) - D^{*}(Q_{1},
    \bG_{n,\th_{n}})\right) 
    &=& o_{P}(1/\sqrt{n}),\\
    \label{eq:hla5:two}
    (h_{n} - \th_{n}) \times P_{0} \left(\partial_{h_{n}}    D^{*}(Q_{n,h_{n}}^{*},
    \bG_{n,\Cdot}) - \partial_{h_{n}} D^{*}(Q_{1},     \bG_{n,\Cdot})\right) 
    &=& o_{P}(1/\sqrt{n}),\\
    \label{eq:hla5:three}
    (P_{n}    -     P_{0})    \left(\partial_{h_{n}}    D^{*}(Q_{n,h_{n}}^{*},
    \bG_{n,\Cdot}) - \partial_{h_{n}} D^{*}(Q_{1},     \bG_{n,\Cdot})\right) 
    &=& o_{P}(1/\sqrt{n}).
  \end{eqnarray}
\end{description}

Now  that we  have introduced  our high-level  assumptions, we  can state  the
corresponding high-level  result that they  entail. The proof is  relegated to
the appendix.

\begin{theorem}[Asymptotics of the collaborative TMLE -- a high-level result] 
  \label{theo:high:level}
  Under assumptions \textbf{A2} to \textbf{A5}, it holds that
  \begin{equation}
    \label{eq:theo:high:level}
    \Psi(P_{n,h_{n}}^{*}) - \Psi(P_{0}) =  (P_{n} - P_{0}) \left(D^{*} (Q_{1},
      \bG_{0}) + \Delta(P_{1})\right) + o_{P} (1/\sqrt{n}).
  \end{equation}
\end{theorem}

\subsubsection*{Commenting on the high-level assumptions.}

Assumption  \textbf{A1}$(Q,h,c)$  concerns  both  $D^{*}$  (specifically,  how
$D^{*}(Q,  \bG)(O)$  depends  on  $\bG(O)$)  and  algorithms  $\hat{\bG}_{t}$,
$t \in \xH$  (specifically, how smooth is  $t \mapsto \hat{\bG}_{t}(P_{n})(O)$
around $h$).  In the particular example studied in the following sections, the
counterpart  \textbf{C1}  of  \textbf{A1}$(Q,h,c)$  concerns  only  algorithms
$\hat{\bG}_{t}$, $t \in \xH$.

In the example, we show how  $P_{n,h_{n}}^{*}$ can be built collaboratively in
such a way that \textbf{A2} is met, under a series of nested assumptions about
the  smoothness  of  data-dependent,  real-valued functions  over  $\xH$,  the
construction of which notably involve algorithms $\hat{\bG}_{t}$, $t \in \xH$.
To understand why achieving \eqref{eq:hla2:two}  is relevant, observe that the
following                  oracle                  version                  of
$P_{n} \partial_{h_{n}} D^{*} (Q_{n,h_{n}}^{*}, \bG_{n,\Cdot})$,
\begin{equation*}
  \lim\limits_{\substack{t\to    0\\t\neq    0}}   \frac{1}{t}   P_{0}
  \left(D^{*}   (Q_{n,h_{n}}^{*},   \bG_{n,h_{n}+t})   -   D^{*}   (Q_{n,h_{n}}^{*},
    \bG_{n,h_{n}})\right), 
\end{equation*}
can be rewritten as 
\begin{equation*}
  \lim\limits_{\substack{t\to    0\\t\neq     0}}    \frac{1}{t}    \left(\Rem_{20}
    (Q_{n,h_{n}}^{*},    \bG_{n,h_{n}+t})   -    \Rem_{20}   (Q_{n,h_{n}}^{*},
    \bG_{n,h_{n}})\right) 
\end{equation*}
in  view of  \eqref{eq:remainder:intro}.  Thus,  achieving \eqref{eq:hla2:one}
relates         to         finding          critical         points         of
$h \mapsto \Rem_{20} (Q_{n,h_{n}}^{*}, \bG_{n,h})$.  

Assumption \textbf{A3}  formalizes the  convergence of $\bG_{n,h_{n}}$  to its
target $\bG_{0}$ w.r.t. $d_{\xG}$, and that of $Q_{n,h_{n}}^{*}$ to some limit
$Q_{1} \in \xQ$  w.r.t. $d_{\xQ}$.  It does not require  that $Q_{1}$ be equal
to the target $Q_{0}$ of  $Q_{n,h_{n}}^{*}$, but \textbf{A4} may be impossible
to meet when $Q_{1} \neq Q_{0}$ (see below).  Condition~\eqref{eq:hla3:one} in
\textbf{A3}   is   met   for    instance   if   the   $L^{2}(P_{0})$-norm   of
$D^{*}(Q_{n,h_{n}}^{*}, \bG_{n,h_{n}})  - D^{*}(Q_{1}, \bG_{0})$ goes  to zero
in probability  and if the  difference falls  in a $P_{0}$-Donsker  class with
probability tending to one.  As  for \eqref{eq:hla3:three}, it typically holds
whenever  the product  of the  rates of  convergence of  $Q_{n,h_{n}}^{*}$ and
$\bG_{n,h_{n}}$  to their  limits is  $o_{P}(1/\sqrt{n})$. The  counterpart of
\textbf{A3} in the example studied in the following sections is \textbf{C2}.

With  \textbf{A4},  we  assume  the  existence of  an  oracle  $\th_{n}$  that
undersmoothes $\bG_{n,h}$  enough so  that $\Phi_{0} (\bG_{n,\th_{n}})$  is an
asymptotically  linear estimator  of $\Phi_{0}(\bG_{0})$,  where we  note that
$\Phi_{0}$ is pathwise differentiable in a  similar way as $\Psi$. We say that
$\th_{n}$ is an  oracle because the definition of  $\Phi_{0}$ involves $P_{0}$
and $Q_{1}$. It happens that
\begin{lemma}
  \label{lem:A4}
  Under    \textbf{A2}    and    \textbf{A3},     if    $Q_{1}    =    Q_{0}$,
  $d_{\xG}(\bG_{n,h_{n}},  \bG_{0}) \times  d_{\xQ}(Q_{n,h_{n}}^{*}, Q_{0})  =
  o_{P}  (1/\sqrt{n})$,   and  if   \eqref{eq:asymp:exp:intro}  is   met  with
  $\IF  = D^{*}(P_{0})$,  then \textbf{A4}  holds with  $h_{n} =  \th_{n}$ and
  $\Delta(P_{1}) = 0$.
\end{lemma}
It is  difficult to  assess whether or  not \textbf{A4} is  a tall  order when
$d_{\xG}(\bG_{n,h_{n}},  \bG_{0}) \times  d_{\xQ}(Q_{n,h_{n}}^{*}, Q_{0})$  is
not necessarily $o_{P} (1/\sqrt{n})$, or if $Q_{1} \neq Q_{0}$.

Finally, \textbf{A5} states  that the distance between  $\th_{n}$ and $h_{n}$,
introduced  in \textbf{A2},  is of  order  $o_{P} (1/n^{1/4})$  at most.   Its
conditions \eqref{eq:hla5:one} and \eqref{eq:hla5:three} are of similar nature
as  \eqref{eq:hla3:one}.   As   for  \eqref{eq:hla5:two},  the  Cauchy-Schwarz
inequality   reveals  that   it   is  met   if   the  $L^{2}(P_{0})$-norm   of
$\partial_{h_{n}}  D^{*}(Q_{n,h_{n}}^{*},  \bG_{n,\Cdot})  -  \partial_{h_{n}}
D^{*}(Q_{1}, \bG_{n,\Cdot})$ is $o_{P}(1/n^{1/4})$.

\section{Collaborative TMLE  for continuous tuning when  inferring the average
  treatment effect: presentation and analysis}
\label{sec:ctmle_con}

In this section,  we specialize the discussion to the  inference of a specific
statistical   parameter,    the   so   called   average    treatment   effect.
Section~\ref{subsec:prelim} introduces the parameter  and recalls what are the
corresponding   $D^{*}$   and    $\Rem_{20}$   from   Section~\ref{sec:intro}.
Section~\ref{subsec:continuum:TMLEs}   describes  the   \textit{un}cooperative
construction      of     a      continuum     of      uncooperative     TMLEs.
Section~\ref{subsec:select:uncoop}  argues why  the  selection of  one of  the
uncooperative  TMLEs  is unlikely  to  yield  a well  behaved  (\textit{i.e.},
asymptotically  Gaussian)  estimator   when  the  product  of   the  rates  of
convergence of the estimators of $Q_{0}$  and $\bG_{0}$ to their limits is not
fast        enough       (\textit{i.e.},        $o(1/\sqrt{n})$).        Then,
Sections~\ref{subsec:continuum:CTMLEs}  and \ref{subsec:select:colla}  present
the collaborative  construction of collaborative  TMLEs and how to  select one
among them  that is well  behaved, under assumptions  that are spelled  out in
Section~\ref{subsec:asymptot},          where          the          high-level
Theorem~\ref{theo:high:level} and its assumptions are specialized.

\subsection{Preliminary}
\label{subsec:prelim}

We  observe  $n$  independent  draws $O_{1}  \equiv  (W_{1},  A_{1},  Y_{1})$,
$\ldots$, $O_{n} \equiv  (W_{n}, A_{n}, Y_{n})$ from $P_{0}$, the  true law of
$O\equiv (W,A,Y)$.   It is  known that  $Y$ takes its  values in  $[0,1]$.  We
consider the statistical model $\xM$ that leaves unspecified the law $Q_{W,0}$
of $W$ and the conditional law of  $Y$ given $(A,W$), while we might know that
the conditional expectation $\bG_0$ of $A$ given $W$ belongs to a set $\xG$.

Introduce
\begin{equation*}
  \bQ_0(A,W)\equiv E_{P_{0}}(Y|A,W), \quad \bG_0(W)\equiv P_{0}(A=1|W).
\end{equation*}
The parameter of interest is the average treatment effect, 
\begin{equation*}
  \psi_0 \equiv E_{Q_{W,0}} \left(\bQ_0(1, W) - \bQ_{0} (0, W)\right).
\end{equation*}
We choose it because its study provides  a wealth of information and paves the
way for the analysis of a variety of other parameters often encountered in the
statistical literature.

More generally,  every $P\in\xM$ gives  rise to $Q_{W}$,  $\bQ(A,W)$, $\bG(W)$
and $Q\equiv  (Q_{W}, \bQ)$, which  are respectively  the marginal law  of $W$
under $P$,  the conditional expectation  of $Y$  given $(A,W)$ under  $P$, the
conditional  probability  that $A=1$  given  $W$  under  $P$, and  the  couple
consisting of  $Q_{W}$ and  $\bQ$.  For  each of  them, the  average treatment
effect is $\Psi(P)$, where $\Psi : \xM \to [0,1]$ is given by
\begin{equation*}
  \Psi(P) \equiv E_{Q_{W}} \left(\bQ(1, W) - \bQ (0, W)\right). 
\end{equation*}

For notational conciseness, we let $\ell\bG$ be given by
\begin{equation}
  \label{eq:ellbG}
  \ell\bG(A,W) \equiv A \bG(W) + (1-A) (1-\bG(W))
\end{equation}
for  every  $\bG  \in  \xG$.   Note that  $\ell\bG(A,W)$  is  the  conditional
likelihood of $A$ given $W$ when $A$ given $W$ is drawn from the Bernoulli law
with  parameter $\bG(W)$,  hence the  ``$\ell$'' in  the notation.   Parameter
$\Psi$ viewed as  a real-valued mapping over $\xM$  is pathwise differentiable
at every $P\in\xM$ w.r.t.  the maximal  tangent set $\xS_{P} = L_{0}^{2} (P)$.
The efficient influence curve $D^*(P)$ of $\Psi$ at $P\in \xM$ is given by
\begin{eqnarray}
  \label{eq:EIC}
  D^*(P)(O) 
  &\equiv&  D_{2}^{*}   (\bQ,  \bG)   (O)  +   \left(\bQ(1,W)  -   \bQ(0,W)  -
           \Psi(P)\right) \quad \text{where}\\
  \notag
  D_{2}^{*} (\bQ, \bG) (O) 
  &\equiv& \frac{2A-1}{\ell\bG(A,W)}(Y-\bQ(A,W)).
\end{eqnarray}

Recall definition  \eqref{eq:remainder:intro}. It is  easy to check  that, for
every $P\in\xM$,
\begin{equation} 
  \label{eq:remainder}
  \Rem_{20}    (\bQ,    \bG)    =    E_{P_{0}}    \left[(2A-1)    \left(1    -
  \frac{\ell\bG_{0}(A,W)}{\ell\bG(A,W)}\right)         \left(\bQ(A,W)        -
  \bQ_0(A,W)\right) \right]. 
\end{equation}
Writing $\Rem_{20} (\bQ, \bG)$ instead of $\Rem_{20} (Q, \bG)$ slightly abuses
notation,  but  is  justified  because  integrating out  $A$  in  the  RHS  of
\eqref{eq:remainder} reveals that it only depends on $P_{0}$, $\bQ$ and $\bG$.
Furthermore, by the Cauchy-Schwartz inequality, it holds that
\begin{equation}
  \label{eq:bound:R20}
  \Rem_{20}(\bQ,    \bG)^{2}    \leq    P_{0}    (\bQ-\bQ_0)^{2}    \times    P_{0}
  \left(\frac{\bG-\bG_0}{\ell\bG}\right)^{2}.  
\end{equation}

\subsection{Uncooperative construction of a continuum of uncooperative TMLEs} 
\label{subsec:continuum:TMLEs}

\subsubsection*{Prerequisites.}

Let $\bQ_{n}^{0} \equiv \hat{\bQ} (P_{n})$  be an initial estimator of $\bQ_0$
and  $\{\bG_{n,h}\equiv \hat{\bG}_h(P_n)  : h  \in  \xH\}$ be  a continuum  of
candidate  estimators of  $\bG_0$ indexed  by a  real-valued tuning  parameter
$h \in \xH$, an open interval  of $\xR_{+}^{*}$.  By convention, we agree that
small values of $h$ correspond with  less bias for $\bG_{n,h}$ as an estimator
of $\bG_0$.   Specifically, denoting $L_{1}$  the valid loss function  for the
estimation of $\bG_{0}$ given by
\begin{equation}
  \label{eq:L1:loss}
  L_{1}(\bG)(A,W)  \equiv  -  \log  \ell\bG(A,W)  =  -A  \log  \bG(W)  -  (1-A)
  \log(1-\bG(W)),
\end{equation}
for every $\bG  \in \xG$, where $\ell\bG$ was  defined in~\eqref{eq:ellbG}, we
assume from  now on  that the  empirical risk  $h \mapsto  P_n L_{1}(\bG_{n,h})$
increases.

For  example, $\hat{\bG}_h$  could  correspond to  fitting  a logistic  linear
regression maximizing the log-likelihood under  the constraint that the sum of
the absolute values of the coefficients is smaller than or equal to $1/h$ with
$h \in \xH\equiv  \xR_{+}^{*}$.  We will refer to this  algorithm as the LASSO
logistic regression algorithm.  

\subsubsection*{Uncooperative TMLEs.}

Let  $Q_{W,n}$  be the  empirical  law  of  $\{W_{1}, \ldots,  W_{n}\}$.   Set
arbitrarily $h  \in \xH$ and let  $P_{n,h}^{0} \in \xM$ denote  any element of
$\xM$ such that  the marginal law of $W$ under  $P_{n,h}^{0}$ equals $Q_{W,n}$
and  the conditional  expectation of  $Y$ given  $(A,W)$ under  $P_{n}^{0}$ is
equal to $\bQ_{n}^{0}$, hence $Q_{n}^{0}  = (Q_{W,n}, \bQ_{n}^{0})$ on the one
hand; and  the conditional  expectation of $A$  given $W$  under $P_{n,h}^{0}$
coincide  with   $\bG_{n,h}$  on  the   other  hand.   Evaluating   $\Psi$  at
$P_{n,h}^{0}$ yields an estimator of $\Psi(P_{0})$,
\begin{equation*}
  \Psi(P_{n,h}^{0}) = \frac{1}{n}  \sum_{i=1}^{n} \left(\bQ_{n}^{0} (1, W_{i})
    - \bQ_{n}^{0} (0, W_{i}) \right),
\end{equation*}
which is not targeted toward the  inference of $\Psi(P_{0})$ in the sense that
none of  the known features  of $P_{n}^{0}$  was derived specifically  for the
sake of ultimately estimating $\Psi(P_{0})$. 

One  way  to  target  $\Psi(P_{n,h}^{0})$ toward  $\Psi(P_{0})$  is  to  build
$P_{n,h}^{*} \in \xM$ from $P_{n,h}^{0}$ in such a way that
\begin{equation*}
  P_{n} D^{*} (P_{n,h}^{*}) = o_{P} (1/\sqrt{n})
\end{equation*}
and to infer  $\Psi(P_{0})$ with $\Psi(P_{n,h}^{*})$. This can  be achieved by
``fluctuating'' $P_{n,h}^{0}$ in the following sense. For every $\bG \in \xG$,
introduce the so called ``clever covariate'' $\xC(\bG)$ given by
\begin{equation}
  \label{eq:clever}
  \xC(\bG) (A,W) \equiv \frac{2A-1}{\ell\bG(A,W)}.
\end{equation}
Now,  for  every $\varepsilon  \in  \xR$,  let $\bQ_{n,h,\varepsilon}^{0}$  be
characterized by
\begin{equation}
  \label{eq:fluct:ref}
  \logit    \left(\bQ_{n,h,\varepsilon}^{0}    (A,W)\right)   \equiv    \logit
  \left(\bQ_{n}^{0} (A,W)\right) + \varepsilon \xC(\bG_{n,h}) (A,W) 
\end{equation}
and  $P_{n,h,\varepsilon}^{0} \in  \xM$ be  defined like  $P_{n,h}^{0}$ except
that   the    conditional   expectation    of   $Y$   given    $(A,W)$   under
$P_{n,h,\varepsilon}^{0}$   equals    $\bQ_{n,h,\varepsilon}^{0}$   (and   not
$\bQ_{n}^{0}$).    Clearly,  $P_{n,h,\varepsilon}^{0}   =  P_{n,h}^{0}$   when
$\varepsilon=0$.  Moreover, denoting $L_{2}$ the loss function given by
\begin{equation*}
  L_{2}  (\bQ)   (O)  \equiv  -Y  \log   \bQ(A,W)  -  (1-Y)  \log   \left(1  -
    \bQ(A,W)\right) 
\end{equation*}
for every $\bQ$ induced by a $P \in \xM$, it holds that
\begin{equation*}
  \frac{d}{d\varepsilon}          L_{2}
  (\bQ_{n,h,\varepsilon}^{0}) (O) = - D_{2} (\bQ_{n,h,\varepsilon}^{0}, \bG_{n,h}) (O),
\end{equation*}
a  property  that  prompts  us   to  say  that  the  one-dimensional  submodel
$\{P_{n,h,\varepsilon}^{0} : \varepsilon \in \xR\} \subset \xM$ ``fluctuates''
$P_{n,h}^{0}$ ``in the direction of'' $D_{2} (\bQ_{n}^{0}, \bG_{n,h})$.

The optimal fluctuation  of $P_{n,h}^{0}$ along the above  submodel is indexed
by the minimizer of the empirical risk
\begin{equation}
  \label{eq:opt:eps:ref}
  \varepsilon_{n,h} \equiv \mathop{\argmin}_{\varepsilon  \in \xR} P_{n} L_{2}
  (\bQ_{n,h,\varepsilon}^{0}), 
\end{equation}
of      which      the      existence     is      assumed      (note      that
$\varepsilon  \mapsto   P_{n}  L_{2}  (\bQ_{n,h,\varepsilon}^{0})$   is  twice
differentiable        and       strictly        convex).        We        call
$P_{n,h}^{*} \equiv  P_{n,h,\varepsilon_{n,h}}^{0}$ the  TMLE of  $P_{0}$, and
the resulting estimator
\begin{equation}
  \label{eq:TMLE}
  \psi_{n,h}^{*}   \equiv    \Psi(P_{n,h}^{*})   =   \frac{1}{n}
  \sum_{i=1}^{n} \left(\bQ_{n,h,\varepsilon_{n,h}}^{0} (1, W_{i}) 
    - \bQ_{n,h,\varepsilon_{n,h}}^{0} (0, W_{i}) \right)
\end{equation}
the  TMLE  of $\Psi(P_{0})$.   It  is  readily  seen that  \eqref{eq:TMLE}  is
equivalent to
\begin{equation*}
  P_{n} \left(D^{*} (P_{n,h}^{*}) - D_{2}^{*} (\bQ_{n,h}^{*}, \bG_{n,h}) \right)
  = 0
\end{equation*}
where    $\bQ_{n,h}^{*}   \equiv    \bQ_{n,h,\varepsilon_{n,h}}^{0}$.    Since
$\varepsilon_{n,h}$       minimizes      the       differentiable      mapping
$\varepsilon \mapsto P_{n} L_{2} (Q_{n,h,\varepsilon}^{0})$, it holds moreover
that
\begin{equation}
  \label{eq:EIC:2:eqn:solved}
  P_{n} D_{2}^{*} (\bQ_{n,h}^{*}, \bG_{n,h}) = 0 
\end{equation}
which, combined with the previous display, yields 
\begin{equation}
  \label{eq:EIC:eqn:solved}
  P_{n} D^{*} (P_{n,h}^{*}) = 0;
\end{equation}
in   words,  $\psi_{n,h}^{*}$   is  targeted   toward  $\Psi(P_{0})$   indeed.
Furthermore,  in view  of \eqref{eq:remainder}  and \eqref{eq:EIC:eqn:solved},
$\psi_{n,h}^{*}$ satisfies
\begin{equation}
  \label{eq:TMLE:expansion}
  \psi_{n,h}^{*} - \Psi(P_{0}) = (P_{n}  - P_{0}) D^{*} (P_{n,h}^{*}) + \Rem_{20}
  (\bQ_{n,h}^{*}, \bG_{n,h}). 
\end{equation}
Finally,  the TMLEs  $\psi_{n,h}^{*}$  ($h \in  \xH$)  are said  uncooperative
because, although they share the same initial estimator $\bQ_{n}^{0}$, for any
two $h,  h' \in \xH$, $h  \neq h'$, the construction  of $\psi_{n,h}^{*}$ does
not capitalize on that of $\psi_{n,h'}^{*}$.  

\subsection{Selecting one of the uncooperative TMLEs}
\label{subsec:select:uncoop}

At this stage  of the procedure, a  crucial question is to select  one TMLE in
the  collection  of  uncooperative  TMLEs,   one  that  lends  itself  to  the
construction of a  CI for $\Psi(P_{0})$ with a given  asymptotic level. Such a
TMLE  necessarily   writes  as  $\psi_{n,h_{n}}^{*}$  for   some  well  chosen
$h_{n}\in \xH$.   This could possibly be  a deterministic (fixed in  $n$) or a
data-driven (random and $n$-dependent) element of $\xH$.

The risk $R_{1}$ generated by $L_{1}$ \eqref{eq:L1:loss} is given by
\begin{equation*}
  R_{1} (\bG, \bG_{0}) \equiv E_{Q_{0,W}} \left[\KL(\bG_{0}(W), \bG(W))\right],
\end{equation*}
where $\KL(p,q)$ is the Kullback-Leibler divergence between the Bernoulli laws
with parameters $p,q\in [0,1]$. By Pinsker's inequality, it holds that
\begin{equation*}
  0 \leq 2 P_{0} \left(\bG - \bG_{0}\right)^{2} \leq R_{1} (\bG, \bG_{0})
\end{equation*}
for all $\bG\in \xG$.  Therefore, if $\bG$ is bounded away  from zero and one,
then \eqref{eq:bound:R20} implies 
\begin{equation}
  \label{eq:bound:R20:bis}
  \Rem_{20}  (\bQ, \bG)^{2}  \lesssim P_{0}  (\bQ -  \bQ_{0})^{2} \times  R_{1}(\bG,
  \bG_{0}).  
\end{equation}

If the deterministic $h_{n} \in \xH$ is such that
\textit{(i)} there exist two rates $\rho_{1,n} = o(1)$ and $\rho_{2,n} = o(1)$
such  that  $R_{1}  (\bG_{n,h_{n}},   \bG_{0})  =  o_{P}(\rho_{1,n}^{2})$  and
$P_{0}   (\bQ_{n,h_{n}}^{*}   -   \bQ_{0})^{2}  =   o_{P}   (\rho_{2,n}^{2})$,
\textit{(ii)}
$P_{0} \left(D^{*} (P_{n,h_{n}}^{*}) -  D^{*} (P_{0})\right)^{2} = o_{P} (1)$,
\textit{(iii)} $D^{*} (P_{n,h_{n}}^{*})$ falls in a $P_{0}$-Donsker class with
$P_{0}$-probability        tending         to        one,        \textit{(iv)}
$o_{P}     (\rho_{1,n}\rho_{2,n})     =      o_{P}     (1/\sqrt{n})$,     then
\citep[][Lemma~19.24]{VdV98},           \eqref{eq:TMLE:expansion}          and
\eqref{eq:bound:R20:bis}  guarantee  that  \eqref{eq:asymp:exp:intro}  is  met
(with $\IF = D^{*}(P_{0})$).
[This argument will be used repeatedly  throughout the article.]  Thus, by the
central limit theorem, $\sqrt{n} (\psi_{n,h_{n}}^{*} - \Psi(P_{0}))$ converges
in     law    to     the     centered    Gaussian     law    with     variance
$\Var_{P_{0}}(D^{*}(P_{0})(O))$.
So,  if  the  synergy  between the  convergences  of  $\bQ_{n,h_{n}}^{*}$  and
$\bG_{n,h_{n}}$  to  their  respective   limits  $\bQ_{0}$  and  $\bG_{0}$  is
sufficient, then the TMLE $\psi_{n,h_{n}}^{*}$ can be used to build CIs.

The   argument  falls   apart   if  $o_{P}   (\rho_{1,n}\rho_{2,n})$  is   not
$o_{P}  (1/\sqrt{n})$ (or,  worse,  if the  $L^{2}(P_{0})$-limit $\bQ_{1}$  of
$\bQ_{n,h_{n}}^{*}$  is  not   $\bQ_{0}$,  because  we  do   not  expect  that
$R_{1} (\bG_{n,h_{n}}, \bG_{0}) = o_{P}(1/n)$).
In that  case, whether  or not it  is possible to  derive a  useful asymptotic
linear expansion of a TMLE $\psi_{n,h_{n}}^{*}$ similar to
\eqref{eq:asymp:exp:intro}  will depend  on whether  or not  we can  derive an
asymptotic                 linear                expansion                 for
$\sqrt{n}\,\Rem_{20} (\bQ_{n,h_{n}}^*,\bG_{n,h_{n}})$.

If $\bG_{n,h_{n}}$ was  derived by maximizing the likelihood  over a correctly
specified, finite-dimensional  and fine-tune-parameter-free  parametric model,
then    $\sqrt{n}\,\Rem_{20}    (\bQ_{n,h_{n}}^*,\bG_{n,h_{n}})$   would    be
asymptotically linear. Because of how we estimate $\bG_{0}$, we now argue that
there  is little  chance that  we can  select $h_{n}  \in \xH$  such that  the
remainder   term   $\sqrt{n}\,\Rem_{20}  (\bQ_{n,h_{n}}^*,\bG_{n,h_{n}})$   is
asymptotically linear.  A natural choice  would be to use the likelihood-based
cross-validation selector  $h_{n,\CV}$. Let  us recall how  it is  derived and
explain why we do not believe it will solve our problem.

Let $B_n\in  \{0,1\}^n$ be a  cross-validation scheme.  For  instance, $B_{n}$
could be  a $V$-fold cross-validation  scheme, \textit{i.e.}, a  random vector
taking $V$ different values $b_{1},  \ldots, b_{V} \in \{0,1\}^{n}$, each with
probability     $1/V$,    such     that     \textit{(i)}    the     proportion
$n^{-1}\sum_{i=1}^{n} b_{v}(i)$ of ones among  the coordinates of each $b_{v}$
is close  to $1/V$, and  \textit{(ii)} $\sum_{v=1}^{V}  b_{v}(i) = 1$  for all
$1 \leq i \leq n$.  Let $P_{n,B_n}^0$  be the empirical probability law of the
training subsample $\{O_i : B_n(i)=0, 1  \leq i \leq n\}$ and $P_{n,B_n}^1$ be
the    empirical    probability    law    of    the    validation    subsample
$\{O_{i}: B_n(i)=1, 1 \leq i \leq n\}$.  The likelihood-based cross-validation
selector $h_{n,\CV}$ of $h\in\xH$ is given by
\begin{equation}
  \label{eq:CV:selector}
  h_{n,\CV}  \equiv \mathop{\arg\min}_{h  \in  \xH} E_{B_n}  \left[P_{n,B_n}^1
    L_{1}(\hat{\bG}_h(P_{n,B_n}^0))\right].  
\end{equation}

Unfortunately,         we          do         not          expect         that
$\sqrt{n}\,\Rem_{20}   (\bQ_{n,h}^*,\bG_{n,h})$   is  asymptotically   linear.
Heuristically, $h_{n,\CV}$ trades off the  bias and variance of $\bG_{n,h}$ as
an estimator  of $\bG_0$,  whereas we  wish to  trade off  this bias  with the
variance  of   $\psi_{n,h}^*$.   Clearly,   the  variance  of   the  estimator
$\psi_{n,h}^* =  \Psi(P_{n,h}^{*})$, where $\Psi$  is a smooth  functional, is
significantly   smaller  than   that   of   the  infinite-dimensional   object
$\bG_{n,h}$.

\subsection{Collaborative construction of finitely many collaborative TMLEs } 
\label{subsec:continuum:CTMLEs}

The   take-home    message   of    Sections~\ref{subsec:continuum:TMLEs}   and
\ref{subsec:select:uncoop} is that  the \textit{uncooperative} construction of
a  continuum   of  standard   TMLEs  will  typically   fail  to   produce  one
asymptotically linear TMLE  if the product of the rates  of convergence of the
estimators  of $\bQ_{0}$  and $\bG_{0}$  to their  limits is  not fast  enough
(\textit{i.e.},  $o(1/\sqrt{n})$).  In  Sections~\ref{subsec:continuum:CTMLEs}
and  \ref{subsec:select:colla}, we  demonstrate  how a  \textit{collaborative}
construction of a  continuum of standard TMLEs can  produce one asymptotically
linear TMLE in this challenging situation, under appropriate assumptions.

\subsubsection*{Recursive construction.}


We  now  present  the  collaborative construction  of  \textit{finitely  many}
TMLEs.  In the  forthcoming theoretical  presentation, we  make on  the fly  a
series of assumptions. The most important ones will be emphasized.

We argued that the cross-validated selector $h_{n,\CV}$ \eqref{eq:CV:selector}
does    not    sufficiently    undersmooth     $\bG_{n,h}$    to    make    of
$\sqrt{n}\,\Rem_{20} (\bQ_{n,h}^*,\bG_{n,h})$  an asymptotically  linear term.
Since we have  assumed that $h \mapsto P_{n} L_{1}  (\bG_{n,h})$ increases, we
can focus  on those tuning parameters  $h$ in $\xH \cap  ]0,h_{n,\CV}]$, a set
assumed non-empty from now (an assumption that we call \As{1}{$P_{n},1$}).

The construction is recursive. It unfolds as follows.
\begin{description}
\item[Initialization.]  We  begin as  in Section~\ref{subsec:continuum:TMLEs}:
  for every  $h \in  \xH \cap ]0,h_{n,\CV}]$,  we build  $\bQ_{n,h}^{(*)}$ and
  $P_{n,h}^{(*)}$ using $\bQ_{n}^{0}$ as an initial estimator of $\bQ_{0}$ and
  $\bG_{n,h}$  as the  estimator of  $\bG_{0}$.   Note that  placing the  star
  symbol   between    parentheses   suggests   that    $\bQ_{n,h}^{(*)}$   and
  $P_{n,h}^{(*)}$  are   the  \textit{tentative}  $h$-specific   estimator  of
  $\bQ_{0}$ and TMLE.  Specifically, for every $h \in \xH \cap ]0,h_{n,\CV}]$,
  we   define   $\bQ_{n,h,\varepsilon}^{0}$    as   in   \eqref{eq:fluct:ref},
  $\varepsilon_{n,h,1}$ as in \eqref{eq:opt:eps:ref},  assuming that it exists
  (an   assumption    that   we    call   \As{2}{$P_{n},   1$}),    then   set
  $\bQ_{n,h}^{(*)}   \equiv    \bQ_{n,h,\varepsilon_{n,h,1}}^{0}$   and   find
  $P_{n,h}^{(*)}  \in   \xM$  such  that   the  marginal  law  of   $W$  under
  $P_{n,h}^{(*)}$ is the empirical law $Q_{W,n}$ of $\{W_{1}, \ldots, W_{n}\}$
  and the conditional  expectation of $Y$ given  $(A,W)$ under $P_{n,h}^{(*)}$
  equals $\bQ_{n,h}^{(*)}$, hence $Q_{n,h}^{(*)} = (Q_{W,n}, \bQ_{n,h}^{(*)})$
  on the  one hand;  and the  conditional expectation of  $A$ given  $W$ under
  $P_{n,h}^{(*)}$ coincides with $\bG_{n,h}$ on the other hand.

  We assume that $h \mapsto P_{n} L_{2} (Q_{n,h}^{(*)})$ is minimized globally
  at $h_{n,1}$ in the interior of $\xH \cap ]0,h_{n,\CV}]$ (an assumption that
  we call \As{3}{$P_{n},1$}).  If there are several minimizers, then $h_{n,1}$
  is   the   largest  of   them   by   choice.    Observe  that,   for   every
  $h \in \xH \cap ]0,h_{n,\CV}]$,
  \begin{equation*} 
    P_{n}   L_{2}  (\bQ_{n,h_{n,1}}^{(*)})  \leq   P_{n}  L_{2}
    (\bQ_{n,h}^{(*)}) \leq P_{n} L_{2} (\bQ_{n,h}^{0})
  \end{equation*} 
  and, in particular,
  \begin{equation*} 
    P_{n}   L_{2}   (\bQ_{n,h_{n,1}}^{(*)})  <   P_{n}   L_{2}
    (\bQ_{n,h_{n,\CV}}^{*}) \leq P_{n} L_{2} (\bQ_{n,h_{n,\CV}}^{0}).
  \end{equation*} 

  Let  us  assume now  that,  in  addition, $h  \mapsto  \varepsilon_{n,h,1}$,
  $h \mapsto  1/\bG_{n,h}(W_{i})$ and $h \mapsto  1/(1-\bG_{n,h}(W_{i}))$ (all
  $1 \leq i  \leq n$) are differentiable in an  open neighborhood of $h_{n,1}$
  (an assumption that we  call \As{4}{$P_{n},1$}).  Consequently, \textit{(i)}
  $\partial_{h_{n,1}} D^{*} (\bQ_{n,h_{n,1}}^{(*)}, \bG_{n,\Cdot}) (O_{i})$ is
  well defined  for each $1  \leq i  \leq n$ (see  \eqref{eq:dD:general}), and
  \textit{(ii)} $h \mapsto P_{n} L_{2} (\bQ_{n,h}^{(*)})$ is differentiable in
  that  neighborhood.   Moreover,  since   $h_{n,1}$  minimizes  the  previous
  mapping, we have
  \begin{eqnarray*} 
    0 
    &=&      -\left.\frac{d}{dt}     P_{n}     L_{2}
        (\bQ_{n,t}^{(*)})\right|_{t=h_{n,1}} \\ 
    &=&  \left(\left.\frac{d}{dt}
        \varepsilon_{n,t,1} \right|_{t=h_{n,1}}\right)  \times P_{n} D_{2}^{*}
        (\bQ_{n,h_{n,1}}^{(*)}, \bG_{n,h_{n,1}}) + \varepsilon_{n,h,1} \times
        P_{n} \partial_{h_{n,1}}  D^{*} (\bQ_{n,h_{n,1}}^{(*)},  \bG_{n,\Cdot}) \\
    &=&          \varepsilon_{n,h,1}          \times
        P_{n}      \partial_{h_{n,1}}      D^{*}
        (\bQ_{n,h_{n,1}}^{(*)}, \bG_{n,\Cdot}),
  \end{eqnarray*} 
  where the third equality holds because
  \begin{equation*} 
    P_{n} D_{2}^{*} (\bQ_{n,h_{n,1}}^{(*)}, \bG_{n,h_{n,1}}) =
    P_{n} D_{2}^{*} (\bQ_{n,h, \varepsilon_{n,h,1}}^{0}, \bG_{n,h_{n,1}}) = 0
  \end{equation*}     in     light    of     \eqref{eq:EIC:2:eqn:solved}.     If
  $\varepsilon_{n,h,1} \neq 0$ (an assumption  that we call \As{5}{$P_{n}, 1$}),
  then we thus have proven that the following equation is solved
  \begin{equation*}
    P_{n} \partial_{h_{n,1}}  D^{*} (\bQ_{n,h_{n,1}}^{(*)}, \bG_{n,\Cdot})
    = 0. 
  \end{equation*}

  To  complete  the  initialization,  we define  $h_{n,0}  \equiv  h_{n,\CV}$,
  $\bQ_{n,h_{n,1}}^{*}              \equiv             \bQ_{n,h_{n,1}}^{(*)}$,
  $Q_{n,h_{n,1}}^{*}                \equiv               Q_{n,h_{n,1}}^{(*)}$,
  $P_{n,h_{n,1}}^{*}                \equiv               P_{n,h_{n,1}}^{(*)}$,
  $\psi_{n,h_{n,1}}^{*}  \equiv \Psi(P_{n,h_{n,1}}^{*})$,  and note  that they
  satisfy
  \begin{gather*}
    P_{n} \partial_{h_{n,1}} D^{*} (\bQ_{n,h_{n,1}}^{*}, \bG_{n,\Cdot})
    = P_{n} D^{*} (P_{n,h_{n,1}}^{*}) = 0 \qquad \text{and}\\
    P_{n} L_{2} (\bQ_{n,h_{n,1}}^{*}) < P_{n} L_{2} (\bQ_{n,h_{n,0}}^{*})
  \end{gather*}
  (recall  how  \eqref{eq:EIC:2:eqn:solved} implied  \eqref{eq:EIC:eqn:solved}
  earlier).
\item[Recursion.]  Let $k \geq 2$ be arbitrarily chosen. Suppose that, for all
  $1   \leq    \ell   <   k$,    we   have   already   built    the   5-tuples
  $(h_{n,\ell},\bQ_{n,h_{n,\ell}}^{*},                   Q_{n,h_{n,\ell}}^{*},
  P_{n,h_{n,\ell}}^{*},     \psi_{n,h_{n,\ell}}^{*})$    under     assumptions
  \As{1}{$P_{n},\ell$}    to     \As{5}{$P_{n},\ell$},    and     also    that
  $\xH  \cap  ]0,h_{n,k-1}]  \neq  \emptyset$  (an  assumption  that  we  call
  \As{1}{$P_{n},k$}).    Let    us   now    present   the    construction   of
  $(h_{n,k}, \bQ_{n,h_{n,k}}^{*}, Q_{n,h_{n,k}}^{*}, P_{n,h_{n,k}}^{*})$ under
  assumptions   \As{1}{$P_{n},k$}    to   \As{5}{$P_{n},k$}.     Because   the
  presentation  is very  similar to  that of  the initialization,  it is  more
  laid out more directly.

  For  every  $h   \in  \xH  \cap  ]0,h_{n,k-1}]$,   we  build  \textit{again}
  $\bQ_{n,h}^{(*)}$      and      $P_{n,h}^{(*)}$      \textit{but}      using
  $\bQ_{n,h_{n,k-1}}^{*}$ as an initial estimator of $\bQ_{0}$ and $\bG_{n,h}$
  as    the    estimator    of    $\bG_{0}$.     Specifically,    for    every
  $h \in  \xH \cap ]0,h_{n,k-1}]$, we  define $\bQ_{n,h,\varepsilon}^{k-1}$ as
  in   \eqref{eq:fluct:ref}  with   $\bQ_{n,h_{n,k-1}}^{*}$  substituted   for
  $\bQ_{n}^{0}$,  $\varepsilon_{n,h,k}$  as   in  \eqref{eq:opt:eps:ref}  with
  $\bQ_{n,h,\varepsilon}^{k-1}$  substituted  for  $\bQ_{n,h,\varepsilon}^{0}$
  (\As{2}{$P_{n}, k$}  assumes the  existence of  $\varepsilon_{n,h,k}$), then
  set  $\bQ_{n,h}^{(*)}  \equiv  \bQ_{n,h,\varepsilon_{n,h,k}}^{k}$  and  find
  $P_{n,h}^{(*)}  \in   \xM$  such  that   the  marginal  law  of   $W$  under
  $P_{n,h}^{(*)}$ is the empirical law $Q_{W,n}$ of $\{W_{1}, \ldots, W_{n}\}$
  and the conditional  expectation of $Y$ given  $(A,W)$ under $P_{n,h}^{(*)}$
  equals $\bQ_{n,h}^{(*)}$, hence $Q_{n,h}^{(*)} = (Q_{W,n}, \bQ_{n,h}^{(*)})$
  on the  one hand;  and the  conditional expectation of  $A$ given  $W$ under
  $P_{n,h}^{(*)}$ coincides with $\bG_{n,h}$ on the other hand.

  We assume that $h \mapsto P_{n} L_{2} (Q_{n,h}^{(*)})$ is minimized globally
  at $h_{n,k}$ in the interior of $\xH \cap ]0,h_{n,k-1}]$ (an assumption that
  we call \As{3}{$P_{n},k$}).  If there are several minimizers, then $h_{n,k}$
  is  the  largest  of  them  by   choice.   Moreover,  we  also  assume  that
  $h   \mapsto  \varepsilon_{n,h,k}$,   $h  \mapsto   1/\bG_{n,h}(W_{i})$  and
  $h   \mapsto  1/(1-\bG_{n,h}(W_{i}))$   (all  $1   \leq  i   \leq  n$)   are
  differentiable in an  open neighborhood of $h_{n,k}$ (an  assumption that we
  call                    \As{4}{$P_{n},k$}).                    Consequently,
  $\partial_{h_{n,k}} D^{*}  (\bQ_{n,h_{n,k}}^{(*)}, \bG_{n,\Cdot})(O_{i})$ is
  well  defined  for each  $1  \leq  i  \leq n$  (see  \eqref{eq:dD:general}),
  $h  \mapsto  P_{n}  L_{2}   (\bQ_{n,h}^{(*)})$  is  differentiable  in  that
  neighborhood and, since $h_{n,k}$ minimizes the previous mapping,
  \begin{equation*}
    \varepsilon_{n,h,k}    \times     P_{n}    \partial_{h_{n,k}}    D^{*}
    (\bQ_{n,h_{n,k}}^{(*)}, \bG_{n,\Cdot}) = 0.
  \end{equation*}
  If   $\varepsilon_{n,h,k}   \neq   0$    (an   assumption   that   we   call
  \As{5}{$P_{n}, k$}), then it holds that
  \begin{equation*}
    P_{n} \partial_{h_{n,k}}  D^{*} (\bQ_{n,h_{n,k}}^{(*)}, \bG_{n,\Cdot})
    = 0. 
  \end{equation*}

  To   complete    the   presentation    and   the   recursion,    we   define
  $\bQ_{n,h_{n,k}}^{*}              \equiv             \bQ_{n,h_{n,k}}^{(*)}$,
  $Q_{n,h_{n,k}}^{*}                \equiv               Q_{n,h_{n,k}}^{(*)}$,
  $P_{n,h_{n,k}}^{*}                \equiv               P_{n,h_{n,k}}^{(*)}$,
  $\psi_{n,h_{n,k}}^{*}  \equiv \Psi(P_{n,h_{n,k}}^{*})$,  and note  that they
  satisfy
  \begin{gather}
    \label{eq:two:eqns:solved:recur}
    P_{n} \partial_{h_{n,k}} D^{*}  (\bQ_{n,h_{n,k}}^{*}, \bG_{n,\Cdot}) =
    P_{n} D^{*} (P_{n,h_{n,k}}^{*}) = 0
    \qquad    \text{and}\\
    \notag 
    P_{n}      L_{2}      (\bQ_{n,h_{n,\ell}}^{*})     <      P_{n}      L_{2}
    (\bQ_{n,h_{n,\ell-1}}^{*})
  \end{gather}
  for all $1 \leq \ell \leq k$.
\end{description}
We  discuss when  to  stop the  loop  in the  next  paragraph. The  collection
$\{P_{n,h_{n,k}}^{*}:  0 \leq  k  \leq  K_{n}\}$ of  TMLEs  is arguably  built
collaboratively,  as the  derivation of  every $P_{n,h_{n,\ell}}^{*}$  heavily
depends on $P_{n,h_{n,\ell-1}}^{*}$.

The loop is iterated until a stopping criterion is met.  The instantiations of
the  collaborative TMLE  laid out  in Section~\ref{sec:software}  rely on  the
LASSO logistic  regression algorithm.  It  is thus possible to  pre-specify an
upper  bound on  $K_{n}$. In  general,  we may  decide to  stop the  recursive
construction whenever a maximal number $K$  of iterations has been reached, or
$h_{n,k}  \leq  \hbar$,  or   $M$  successive  TMLEs  $\psi_{n,h_{n,k+m}}^{*}$
($0  \leq  m  <  M$)  all  belong  to  an  interval  of  length  smaller  than
$\eta_{n,k}$,  for  some  user-supplied  integers $K_{\max}$,  $M$  and  small
positive  numbers $h_{\min}$  and $\eta_{n,k}$,  the former  chosen such  that
$\xH \cap ]0,h_{\min}[$ is non-empty  and the latter possibly sample-size- and
data-driven.   The  choice   of  $K_{\max}$  would  typically   be  driven  by
considerations about the  computational time.  The choice  of $h_{\min}$ would
typically  depend  on  the  collection   $\{\hat{\bG}_h  :  h  \in  \xH\}$  of
$h$-specific algorithms, $h\leq h_{\min}$ meaning that too much undersmoothing
is certainly  at play when  using $\hat{\bG}_{h}$.  We would  suggest choosing
$M\equiv       3$        and       characterizing        $\eta_{n,k}$       by
$\eta_{n,k}^{2}   \equiv  \Upsilon_{P_{n}}(P_{n,h_{n,k}}^{*})   /  10n$   with
$\Upsilon_{P_{n}} : \xM \to \xR_{+}^{*}$ given by
\begin{equation*}
  \Upsilon_{P_{n}}(P)  \equiv E_{P_{n}}\left[D^{*}(P)(O)^{2}\right]  = \frac{1}{n}
  \sum_{i=1}^{n} D^{*}(P)(O_{i})^{2}.
\end{equation*}
The  definition   of  $\Upsilon_{P_{n}}$  is   justified  by  the   fact  that
$\Upsilon_{P_{n}} (D^{*}(P_{n,h_{n}}^{*}))$ estimates  the asymptotic variance
of  the   TMLE  $\Psi(P_{n,h_{n}}^{*})$   in  the   context  where   we  prove
\eqref{eq:asymp:exp:intro}    (with     $\IF    =    D^{*}     (P_{0})$)    in
Section~\ref{subsec:select:uncoop}.

\subsection{Selecting one of the finitely many collaborative TMLEs }
\label{subsec:select:colla}

It  remains  to  determine  which  TMLE to  select  among  the  collection  of
collaborative        TMLEs        that         we        constructed        in
Section~\ref{subsec:continuum:CTMLEs}.   Again, the  selection  hinges on  the
cross-validation principle.

The recursive construction  described in Section~\ref{subsec:continuum:CTMLEs}
can  be applied  to the  empirical  measure $\bbP_{n}$  of any  subset of  the
complete   data    set.    Starting   from   $h_{n,\CV}$    (as   defined   in
\eqref{eq:CV:selector}  even when  $\bbP_{n}$ differs  from $P_{n}$),  let the
5-tuple
$(\bbH_{n,1},        \bbbQ_{n,\bbH_{n,1}}^{*},        \bbQ_{n,\bbH_{n,1}}^{*},
\bbP_{n,\bbH_{n,1}}^{*}, \Psi(\bbP_{n,\bbH_{n,1}}^{*}))$  be defined  like the
5-tuple
$(h_{n,1},\bQ_{n,h_{n,1}}^{*},      Q_{n,h_{n,1}}^{*},      P_{n,h_{n,1}}^{*},
\psi_{n,h_{n,1}}^{*})$   with  $\bbP_{n}$   substituted  for   $P_{n}$,  under
assumptions   \As{1}{$\bbP_{n},   1$}   to   \As{5}{$\bbP_{n},   1$}.    Then,
recursively,                                                               let
$(\bbH_{n,k},        \bbbQ_{n,\bbH_{n,k}}^{*},        \bbQ_{n,\bbH_{n,k}}^{*},
\bbP_{n,\bbH_{n,k}}^{*},   \Psi(\bbP_{n,\bbH_{n,k}}^{*}))$  be   defined  like
$(h_{n,k},\bQ_{n,h_{n,k}}^{*},      Q_{n,h_{n,k}}^{*},      P_{n,h_{n,k}}^{*},
\psi_{n,h_{n,k}}^{*})$   with  $\bbP_{n}$   substituted  for   $P_{n}$,  under
assumptions  \As{1}{$\bbP_{n}, k$}  to \As{5}{$\bbP_{n},  k$}.  The  recursive
construction  is stopped  when $\bbK_{n}$  5-tuples have  been derived,  where
$\bbK_{n}$ is defined like $K_{n}$ with $\bbP_{n}$ substituted for $P_{n}$.

The                                                                 collection
\begin{equation}
  \label{eq:bbCTMLE}
  \left\{(\bbH_{n,k},       \bbbQ_{n,\bbH_{n,k}}^{*},       \bbQ_{n,\bbH_{n,k}}^{*},
    \bbP_{n,\bbH_{n,k}}^{*},  \Psi(\bbP_{n,\bbH_{n,k}}^{*})) :  1 \leq  k \leq
    \bbK_{n}\right\} 
\end{equation}
of  $\bbK_{n}$  collaborative   TMLEs  is  used  to  define   a  continuum  of
collaborative TMLEs in the following straightforward way.  The challenge is to
associate                               a                              4-tuple
$(\bbbQ_{n,h}^{*},  \bbQ_{n,h}^{*}, \bbP_{n,h}^{*},  \Psi(\bbP_{n,h}^{*}))$ to
any $h \in \xH \cap ]0, h_{n,\CV}]$.  To do so, we simply let $\bbH_{n}(h)$ be
the element of  $\{\bbH_{n,k} : 1 \leq  k \leq \bbK_{n}\}$ that  is closest to
$h$ (with a preference for the larger of the two closer ones when $h$ is right
in the middle), that is, formally, we set 
\begin{equation}
  \label{eq:bbH}
  \bbH_{n}(h) \equiv \max  \Big\{\bbH_{n,k} : |h - \bbH_{n,k}| =  \min \{|h -
  \bbH_{n,\ell}| : 1 \leq \ell \leq \bbK_{n}\}\Big\}
\end{equation}
and associate to $h$ the corresponding 4-tuple
\begin{equation*}
  (\bbbQ_{n,\bbH_{n}(h)}^{*},                        \bbQ_{n,\bbH_{n}(h)}^{*},
  \bbP_{n,\bbH_{n}(h)}^{*}, \Psi(\bbP_{n,\bbH_{n}(h)}^{*})). 
\end{equation*}

Let    $B_{n}$    be    the     cross-validation    scheme    introduced    in
Section~\ref{subsec:select:uncoop}.   By convention,  let  the  $\max$ of  the
empty set be 0. The collaborative TMLE that we select is
\begin{equation}
  \label{eq:CTMLE}
  (\bQ_{n,\hnk}^{*}, Q_{n,\hnk}^{*}, P_{n,\hnk}^{*}, \psi_{n,\hnk}^{*}) 
\end{equation}
where $\kappa_{n}$ is given by 
\begin{equation}
  \label{eq:kappa}
  \kappa_{n} \equiv  1 \vee  \max \left\{1  \leq k \leq  K_{n} :  h_{n,k} \geq
    \mathop{\arg\min}_{h    \in    \xH    \cap   ]0,    h_{n,\CV}]}    E_{B_n}
    \left[P_{n,B_n}^1     L_{2}(\bbbQ_{n,\bbH_{n}(h)}^{*}\big|_{\bbP_{n}     =
        P_{n,B_n}^0}))\right]\right\}. 
\end{equation}
In words,  $\kappa_{n}$ is the unique  element of $\{1, \ldots,  K_{n}\}$ such
that $\hnk$ is the smallest element of $\{h_{n,1}, \ldots, h_{n,K_{n}}\}$ that
is  larger than  the  minimizer  of the  cross-validated  $L_{2}$-risk of  the
collaborative TMLE,  if there  exists such  an element,  and 1  otherwise.  In
\eqref{eq:kappa},  $\bbbQ_{n,\bbH_{n}(h)}^{*}\big|_{\bbP_{n}  =  P_{n,B_n}^0}$
equals $\bbbQ_{n,\bbH_{n}(h)}^{*}$ when $\bbP_{n} = P_{n,B_n}^0$.

The contrast  between $\hnk$ and  $h_{n,\CV}$ is  stark. At first  glance, the
main  difference is  that the  role play  by cross-validated  $L_{1}$-risks of
algorithms  to  estimate  $\bG_{0}$  in \eqref{eq:CV:selector}  is  played  by
cross-validated  $L_{2}$-risks   of  algorithms   to  estimate   $\bQ_{0}$  in
\eqref{eq:kappa}.  A closer examination reveals that the difference is deeper.
Replacing       $L_{1}       (\hat{\bG}_{h}       (P_{n,B_{n}}^{0}))$       by
$L_{2}  (\bQ_{n,B_{n}h}^{0*})$   (with  $\bQ_{n,B_{n},h}^{0*}$   defined  like
$\bQ_{n,h}^{*}$   in   Section~\ref{subsec:continuum:TMLEs}   but   based   on
$P_{n,B_{n}}^{0}$  instead  of  $P_{n}$)  would  not  make  of  the  resulting
alternative cross-validated selector  of $h$ a good candidate:  because of the
inherent lack of cooperation between the uncooperative TMLEs $\psi_{n,h}^{*}$,
the resulting estimator  of $\bG_{0}$ would not even be  consistent. This fact
motivates the general  C-TMLE methodology, of which  the present instantiation
includes  a   twist  consisting  in   solving  two  critical   equations,  see
\eqref{eq:two:eqns:solved:recur}.

\subsection{Asymptotics}
\label{subsec:asymptot}

The   study  of   the  asymptotic   properties  of   the  collaborative   TMLE
$\psi_{n,\hnk}^{*}$ hinges on  Theorem~\ref{theo:high:level}. We first specify
two pseudo-distances  $d_{\xG}$ and $d_{\xQ}$ on  $\xG$ and $\xQ$ in  light of
requirement  \eqref{eq:remainder:bound:intro}.  On  the one  hand, because  we
will  eventually assume  that $\bG_{0}$  is bounded  away from  zero and  one,
\eqref{eq:remainder} yields that  we can choose $d_{\xG}$ such  that, for each
$\bG_{1}, \bG_{2} \in \xG$,
\begin{equation*}
  d_{\xG} (\bG_{1},\bG_{2})^{2} \equiv P_{0} (\bG_{1}-\bG_{2})^{2}.  
\end{equation*}
On     the      other     hand,      note     that      any     data-dependent
$Q_{n}  \equiv  (Q_{W,n},  \bQ_{n})  \in   \xQ$  naturally  gives  rise  to  a
substitution estimator $\psi_{n}$ of $\psi_{0}$:
\begin{equation*}
  \psi_{n} \equiv  E_{Q_{W,n}} \left(\bQ_{n} (1,  W) - \bQ_{n}  (0,W)\right) =
  \frac{1}{n}    \sum_{i=1}^{n}   \left(\bQ_{n}    (1,   W_{i})    -   \bQ_{n}
    (0,W_{i})\right). 
\end{equation*}
It is easy to check (see the appendix) that the following result holds.
\begin{lemma}
  \label{lem:easy}
  Assume  that  $\bG_{0}$  is  bounded  away   from  zero  and  one  and  that
  $\bQ_{n}(1,\cdot) - \bQ_{n}(0,\cdot)$ falls  in a $P_{0}$-Donsker class with
  $P_{0}$-probability          tending           to          one.           If
  $P_{0}  (\bQ_{n}  - \bQ_{1})^{2}  =  o_{P}  (1)$  for some  $\bQ_{1}$,  then
  $\psi_{n} = P_{0} (\bQ_{1}(1,\cdot) - \bQ_{1}(0,\cdot)) + o_{P} (1)$.
\end{lemma}
Since we  always estimate the  marginal law  of $W$ under  $P_{0}$, $Q_{W,0}$,
with   its  empirical   counterpart  $Q_{W,n}$,   we  can   thus  define   the
pseudo-distance $d_{\xQ}$ by setting, for each $Q_{1}, Q_{2} \in \xQ$,
\begin{equation*}
  d_{\xQ} (Q_{1}, Q_{2})^{2} \equiv P_{0}  (\bQ_{1} - \bQ_{2})^{2},
\end{equation*}
an expression  that does  not depend  on the first  components of  $Q_{1}$ and
$Q_{2}$.

Consider the following inter-dependent assumptions.   The first one is related
to \textbf{A1}$(Q,h,c)$ and completes \As{5}{$P_{n}, h_{n,\kappa_{n}}$}.
\begin{description}
\item[\textbf{C1}] There exists a universal  constant $C_{1} \in ]0,1/2[$ such
  that $\bG_{0}$  and any by-product $\bG_{n,h}$  of algorithm $\hat{\bG}_{h}$
  (any $h \in \xH$) trained on the empirical measure $P_{n}$ take their values
  in  $[C_{1},  1-C_{1}]$.   Moreover,   there  exists  an  open  neighborhood
  $\xT \subset \xH$ of $h_{n,\kappa_{n}}$ and a universal constant $C_{2} > 0$
  such  that  $t \mapsto  \bG_{n,t}(W)$  is  twice differentiable  over  $\xT$
  ($P_{0}$-almost surely) and, $P_{0}$-almost surely,
  \begin{equation*}
    \sup_{h  \in  \xT}  \left|\frac{d}{dt}  \bG_{n,t}  (W)|_{t=h}\right|  \vee
    \sup_{h \in \xT} \left|\frac{d^{2}}{dt^{2}} \bG_{n,t} (W)|_{t=h}\right| 
    \leq C_{2}.  
  \end{equation*}
\end{description}
When \textbf{C1}  is met,  we denote $\bG_{n,h}'(W)$  the first  derivative of
$t \mapsto \bG_{n,t}(W)$ at $h \in \xH$.
\begin{description}
\item[\textbf{C2}]  Both $\bG_{n,  \hnk}$ and  $\bQ_{n,\hnk}^{*}$ converge  in
  $L^{2}(P_{0})$, to $\bG_{0}$ and $\bQ_{1}$ respectively.  Moreover, it holds
  that   $P_{0}  (\bG_{n,\hnk}   -  \bG_{0})^{2}   =  o_{P}(1/\sqrt{n})$   and
  $P_{0}  (\bG_{n,\hnk}  -  \bG_{0})^{2}   \times  P_{0}  (\bQ_{n,\hnk}^{*}  -
  \bQ_{1})^{2} = o_{P}(1/n)$.
\item[\textbf{C3}]        Assumption         \textbf{A4}        is        met,
  $(\hnk       -       \th_{n})^{2}       =       o_{P}(1/\sqrt{n})$       and
  $(\hnk  -  \th_{n})^{2}  \times  P_{0} (\bQ_{n,\hnk}^{*}  -  \bQ_{1})^{2}  =
  o_{P}(1/n)$.
\item[\textbf{C4}]     With     $P_{0}$-probability    tending     to     one,
  $\bQ_{n,\hnk}^{*}$,  $\bG_{n,\hnk}$,  $\bG_{n,\th_{n}}$ and  $\bG_{n,\hnk}'$
  fall in $P_{0}$-Donsker classes.
\end{description}

We    are    now   in    a    position    to    state   the    corollary    of
Theorem~\ref{theo:high:level} that describes the  asymptotic properties of the
collaborative TMLE targeting the average treatment effect.

\begin{corollary}[Asymptotics of the collaborative TMLE -- targeting the average
  treatment effect]
  \label{theo:specific}
  Suppose that  assumptions \As{1}{$\cdot,  \cdot$} to  \As{5}{$\cdot, \cdot$}
  that     we    made     in    Sections~\ref{subsec:continuum:CTMLEs}     and
  \ref{subsec:select:colla} when constructing the  collaborative TMLE given in
  \eqref{eq:CTMLE}  are  met.   In   addition,  suppose  that  \textbf{C1}  to
  \textbf{C4} are satisfied. Then
  \begin{equation*}
    \psi_{n,\hnk}^{*}  - \Psi(P_{0})  =  (P_{n} -  P_{0}) \left(D^{*}  (Q_{1},
      \bG_{0}) + \Delta(P_{1})\right) + o_{P} (1/\sqrt{n}). 
  \end{equation*}
\end{corollary}

By    the    central   limit    theorem,    the    corrolary   implies    that
$\sqrt{n} (\psi_{n,\hnk}^{*} - \Psi(P_{0}))$ converges  in law to the centered
Gaussian              law             with              a             variance
$\sigma^{2}  \equiv  P_{0}  (D^{*}  (Q_{1},  \bG_{0})  +  \Delta(P_{1}))^{2}$.
Therefore,  provided  that  we  can  estimate  $\sigma^{2}$  consistently  (or
conservatively), we  can build CIs  for $\Psi(P_{0})$ with a  given asymptotic
level.       Sections~\ref{sec:software},       \ref{sec:experiments}      and
\ref{sec:transfer}   investigate   the   practical   implementation   of   the
collaborative TMLE and its performances in a simulation study.

\section{Collaborative TMLE  for continuous tuning when  inferring the average
  treatment effect: practical implementation}
\label{sec:software}

In  this  section,  we  describe  the  practical  implementation  of  the  two
instantiations of  the collaborative TMLE  algorithm presented and  studied in
Section~\ref{sec:ctmle_con}.     In    both    of   them,    the    collection
$\{\hat{\bG}_{h} : h  \in \xH\}$ is embodied in \verb|R|  by the \verb|glmnet|
algorithm~\citep{glmnet2010}.   The nature  of  algorithm  $\hat{\bQ}$ is  left
unspecified. As for $\bQ_{n}^{0}$, it is obtained once and for all by training
$\hat{\bQ}$ on $P_{n}$ at the  beginining of the procedure. More specifically,
we never evaluate $\hat{\bQ}(\bbP_{n})$ for $\bbP_{n} \neq P_{n}$.

\subsection{LASSO-C-TMLE}
\label{sec:lasso:ctmle}

We now  describe our LASSO-C-TMLE  algorithm.  Recall that  $\bbP_{n}$ denotes
the  empirical measure  of a  generic subset  of the  complete data  set.  The
following  algorithm   implements  the  theoretical  procedure   laid  out  in
Sections~\ref{subsec:continuum:CTMLEs} and \ref{subsec:select:colla}.
\begin{enumerate}
\item\label{enum:one}              Build               a              sequence
  $\{\bG_{n, h} \equiv \hat{\bG}_{h} (P_{n}):  h \in \xH_{100}\}$ by computing
  a discretized version of the path of the LASSO logistic regression of $A$ on
  $W$  with a  regularization parameter  $h$  ranging in  the set  $\xH_{100}$
  provided   by  \verb|cv.glmnet|   with  options   \verb|nlambda=100|  (hence
  $\text{card}(\xH_{100})     =    100$)     and    \verb|nfolds=10|.      Set
  $h_{\min}  \equiv   \min  \xH_{100}$  and   let  $h_{n,\CV}$  be   equal  to
  \verb|lambda.min|.
\item                    Build                   a                    sequence
  $\{\bbG_{n,  h}  \equiv  \hat{\bG}_{h}  (\bbP_{n}):  h  \in  \xH_{100}  \cap
  [h_{\min}, h_{n,\CV})\}$ by  computing a discretized version of  the path of
  the LASSO logistic regression of $A$  on $W$ with a regularization parameter
  $h$ ranging in $\xH_{100}  \cap [h_{\min}, h_{n,\CV})$ using \texttt{glmnet}
  with a  \texttt{lambda} set to  $\xH_{100} \cap [h_{\min},  h_{n,\CV})$ from
  step~\ref{enum:one}.
\item[] Set $k\equiv 1$ and $\bbH_{n,k-1} \equiv h_{n,\CV}$.
\item                \label{enum:three}               For                every
  $h \in \xH_{100} \cap [h_{\min}, \bbH_{n,k-1})$, determine $\bbbQ_{n,h}^{k}$
  by fluctuating  $\bbbQ_{n}^{k-1}$ based on $\bbG_{n,h}$  (and $\bbP_{n}$) as
  Section~\ref{subsec:continuum:TMLEs}.
\item\label{enum:four}    Identify     the    minimizer     $\bbH_{n,k}$    of
  $h      \mapsto       \bbP_{n}      L_{2}       (\bbbQ_{n,h}^{k})$      over
  $\xH_{100}    \cap    [h_{\min},    \bbH_{n,k-1})$,   define    and    store
  $\bbbQ_{n,h}^{*}        \equiv       \bbbQ_{n,h}^{k}$        for       every
  $h  \in  \xH_{100}  \cap  [\bbH_{n,k}, \bbH_{n,k-1})$,  and  finally  define
  $\bbbQ_n^{k} \equiv \bbbQ_{n,h_{n,k}}^{k}$.
\item\label{enum:five}   As    long   as   $\bbH_{n,k}   >    h_{\min}$,   set
  $k  \leftarrow k+1$  and repeat  steps~\ref{enum:three} and  \ref{enum:four}
  recursively.
\end{enumerate}
The  algorithm necessarily  converges in  a  finite number  of repetitions  of
step~\ref{enum:four}. Let $\bbK_{n}$  be the number of  repetitions. For every
$1           \leq           k          \leq           \bbK_{n}$,           set
$\bbQ_{n,\bbH_{n,k}}^{*}  \equiv  (\bbQ_{W,n},  \bbbQ_{n,\bbH_{n,k}}^{*})  \in
\xQ$ (its  first component is the  empirical law of $W$  under $\bbP_{n}$) and
let $\bbP_{n,\bbH_{n,k}}^{*} \in  \xM$ be any element of model  $\xM$ of which
the   $Q$-component   equals    $\bbQ_{n,\bbH_{n,k}}^{*}$.    The   collection
\eqref{eq:bbCTMLE}   of    $\bbK_{n}$   collaborative   TMLEs    and   mapping
$h  \mapsto  \bbH_{n}  (h)$  over  $\xH_{100}  [h_{\min},  h_{n,\CV})$  as  in
\eqref{eq:bbH} are thus now well defined.

Recall the  definition of  the cross-validation  scheme $B_{n}$  introduced in
Section~\ref{subsec:select:uncoop}. Set
\begin{equation*}
  \hbar(P_{n})  \equiv  \mathop{\arg\min}_{h  \in  \xH_{100}  \cap  [h_{\min},
    h_{n,\CV})}         E_{B_{n}}          \left[P_{n,B_{n}}^{1}         L_{2}
    (\left.\bbbQ_{n,\bbH_{n}(h)}\right|_{\bbP_{n} = P_{n,B_{n}}^{0}})\right]
\end{equation*}
and     run    once     steps~\ref{enum:one}    to     \ref{enum:five}    with
$\bbP_{n} \equiv P_{n}$, hence the collection 
\begin{equation*}
  \left\{(h_{n,k},       \bQ_{n,h_{n,k}}^{*},       Q_{n,h_{n,k}}^{*},
    P_{n,h_{n,k}}^{*},  \Psi(P_{n,h_{n,k}}^{*})) :  1 \leq  k \leq
    K_{n}\right\} 
\end{equation*}
of collaborative TMLEs. Finally, set
\begin{equation*}
  \kappa_{n} \equiv  1 \vee  \max \left\{1  \leq k \leq  K_{n} :  h_{n,k} \geq
    \hbar(P_{n})\right\}. 
\end{equation*}
The  collaborative  TMLE  that  we  select,  our  LASSO-C-TMLE  estimator,  is
$\Psi(P_{n,h_{n,\kappa_{n}}}^{*})$, as in \eqref{eq:CTMLE}.

\subsection{LASSO-PSEUDO-C-TMLE}
\label{sec:lasso:pseudo:ctmle}

The LASSO-C-TMLE procedure described in Section~\ref{sec:lasso:ctmle} is quite
demanding  computationally.   It  is  thus  tempting to  try  and  develop  an
alternative algorithm that would mimick LASSO-C-TMLE but be simpler.  

In  Section~\ref{subsec:select:colla},  we   emphasized  (see  comment  before
statement of theorem)  that one of the  keys of LASSO-C-TMLE is  to ensure the
existence of $h_{n} \in \xH$ and $\bQ_{n,h_{n}}^{*}$ such that
\begin{equation*}
  P_{n} \partial_{h_{n}} D^{*} (\bQ_{n,h_{n}}^{*}, \bG_{n,\Cdot}) = 0. 
\end{equation*}
If   we   knew   how   to    compute   the   derivative   $\bG_{n,h}'(W)$   of
$t  \mapsto \bG_{n,t}(W)$  at $t=h$,  then this  could be  easily achieved  by
enriching         the        fluctuation         of        the         initial
$\bQ_{n}^{0}   \equiv  \hat{\bQ}   (P_{n})$.    Specifically,   in  light   of
\eqref{eq:clever}  and \eqref{eq:fluct:ref},  given any  $h\in \xH$,  we would
define
\begin{equation}
  \label{eq:clever:two}
  \xC_{h}^{+}  (\bG_{n,\Cdot})   (A,W)  \equiv   \xC  (\bG_{n,h})   (A,W)  \left(1,
    \bG_{n,h}'(W)\right), 
\end{equation}
introduce     $\bQ_{n,h,\varepsilon^{+}}^{0}$     characterized    for     any
$\varepsilon^{+} \in \xR^{2}$ by
\begin{equation*}
  \logit    \left(\bQ_{n,h,\varepsilon^{+}}^{0}    (A,W)\right)   \equiv    \logit
  \left(\bQ_{n}^{0} (A,W)\right) + \xC_{h}^{+}(\bG_{n,\Cdot}) (A,W) \varepsilon^{+}
\end{equation*}
and      $P_{n,h,\varepsilon^{+}}^{0}      \in     \xM$      defined      like
$P_{n,h,\varepsilon}^{0}$ except that the conditional expectation of $Y$ given
$(A,W)$          under           $P_{n,h,\varepsilon^{+}}^{0}$          equals
$\bQ_{n,h,\varepsilon^{+}}^{0}$ (and  not $\bQ_{n,h,\varepsilon}^{0}$).  Then,
the optimal  fluctuation would be  indexed by  the minimizer of  the empirical
risk
\begin{equation*}
  \varepsilon_{n,h}^{+} \equiv  \mathop{\argmin}_{\varepsilon^{+} \in \xR^{2}}
  P_{n}   L_{2}   (\bQ_{n,h,\varepsilon^{+}}^{0}). 
\end{equation*}
It  would  result  in $\bQ_{n,h}^{*+}  \equiv  \bQ_{n,h,\varepsilon^{+}}^{0}$,
$P_{n,h}^{*+} \equiv P_{n,h,\varepsilon_{n,h}^{+}}^{0}$ and the TMLE
\begin{equation*}
  \psi_{n,h}^{*+}   \equiv    \Psi(P_{n,h}^{*+})   =   \frac{1}{n}
  \sum_{i=1}^{n}  \left(\bQ_{n,h}^{*+} (1,  W_{i}) -  \bQ_{n,h}^{*+}
    (0, W_{i}) \right)
\end{equation*}
where, by construction, we would have
\begin{equation*}
  P_{n}  D^{*}   (\bQ_{n,h}^{*+},  \bG_{n,h})   =  P_{n}   \partial_{h}  D^{*}
  (\bQ_{n,h}^{*+}, \bG_{n,h}) = 0. 
\end{equation*}

The  LASSO-PSEUDO-C-TMLE  algorithm that  we  now  describe adapts  the  above
procedure. It unfolds as follows.
\begin{enumerate}
\item                    Build                   a                    sequence
  $\{\bG_{n, h} \equiv \hat{\bG}_{h} (P_{n}):  h \in \xH_{100}\}$ by computing
  a discretized version of the path of the LASSO logistic regression of $A$ on
  $W$  with a  regularization parameter  $h$  ranging in  the set  $\xH_{100}$
  provided   by  \verb|cv.glmnet|   with   option  \verb|nlambda=100|   (hence
  $\text{card}(\xH_{100}) = 100$).  Let $h_{n}$ be equal to \verb|lambda.min|.
  Our estimator of $\bG_{0}$ is $\bG_{n,h_{n}}$.
\item                            Choose                            arbitrarily
  $h_{n}^{+} \in  \mathop{\arg\min} \{|h -  h_{n}| :  h \in \xH_{100},  h \neq
  h_{n}\}$ and, for every $1 \leq i \leq n$, define
  \begin{equation*}
    \bG_{n,h_{n}}^{\prime +} (W_{i}) \equiv \frac{\bG_{n,h_{n}^{+}} (W_{i}) -
      \bG_{n,h_{n}} (W_{i})}{h_{n}^{+} - h_{n}}, 
  \end{equation*}
  a     rudimentary    numerical     approximation    of     the    derivative
  $\bG_{n,h_{n}}' (W_{i})$ of $t \mapsto \bG_{n,t} (W_{i})$ at $t = h_{n}$.
\item  Determine  $\bQ_{n,h_{n}}^{*+}$  and  $P_{n,h_{n}}^{*+}$  as  described
  above,  with  $h=h_{n}$  and   $\bG_{n,h_{n}}^{\prime  +}$  substituted  for
  $\bG_{n,h_{n}}'$ in \eqref{eq:clever:two}.
\end{enumerate}
The LASSO-PSEUDO-C-TMLE estimator is $\psi_{n,h_{n}}^{*+}$.

\section{Main simulation study}
\label{sec:experiments}

In this section, we present the results of a multi-faceted simulation study of
the behaviors and performances of  the two instantiations of the collaborative
TMLE   described   in  Section~\ref{sec:software}.    Section~\ref{subsec:sim}
specifies  the synthetic  data-generating  distribution $P_{0}$  that we  use,
Section~\ref{subsec:compet}     introduces    the     competing    estimators,
Section~\ref{subsec:strategy} outlines the structure  of the simulation study,
and   Section~\ref{subsec:results}   gathers    its   results.    Written   in
\verb|R|~\citep{R},  our  code  makes   extensive  use  of  the  \verb|C-TMLE|
package~\citep{gruber2010ctmle}.

\subsection{Synthetic data-generating distribution}
\label{subsec:sim}

Our synthetic data-generating distribution  $P_{0} = \Pi_{0,p,\delta}$ depends
on two fine-tune  parameters: the dimension~$p$ of the  baseline covariate $W$
and a nonnegative constant $\delta \geq 0$.  Sampling $O \equiv (W,A,Y)$ under
$\Pi_{0,p,\delta}$ unfolds sequentially along the following steps.
\begin{enumerate}
\item  Sample  $\tilde{W}$  from  the  centered  Gaussian  law  on  $\xR^{M}$,
  $M  =  \ceil{p/10}$,  of  which   the  covariance  matrix  $\Sigma$  is  the
  block-diagonal matrix $(A_{kl})_{1 \leq k, l \leq M}$ where: $A_{11}$ is the
  $10\times  10$  identity  matrix;  each  $A_{kk}$ for  $1<k\leq  M$  is  the
  block-diagonal matrix $(B_{k,st})_{1 \leq s,t \leq 4}$ with
  \begin{equation*}
    B_{k,11} = \left(
      \begin{array}{ccc}
        1&0&.25\\
        0&1&.25\\
        .25&.25&1\\
      \end{array}
    \right), \quad
    B_{k,22} = B_{k,33} = \left(
      \begin{array}{cc}
        1&.5\\
        .5&1\\
      \end{array}
    \right), \quad
    B_{k,44} = \left(
      \begin{array}{ccc}
        1&.5&0\\
        .5&1&0\\
        0&0&1\\
      \end{array}
    \right)
  \end{equation*}
  and $B_{k,st}$ is a zero matrix for $1  \leq s \neq t \leq 4$; each $A_{kl}$
  for  $1  \leq  k\neq  l  \leq  M$   is  a  zero  matrix.   If  $p=10$,  then
  $\Sigma = A_{11}$  and we set $W  \equiv \tilde{W}$.  If $M >  10p$, then we
  set $W \equiv (\tilde{W}_{1}, \ldots, \tilde{W}_{p})^{\top}$.
\item Sample $A$ conditionally on $W$ from the Bernoulli law with paramater
  \begin{equation*}
    \bG_{0}  (W)  \equiv  \expit \left(\delta  +  \sum_{k=1}^{p}  \beta_{k}
      W_{k}\right),  
  \end{equation*}
  where $(\beta_1, \cdots, \beta_p) = (1, 1, 3/(p-2), \ldots, 3/(p-2))$.
\item Sample $\tilde{Y}$  conditionally on $(A,W)$ from the  Gaussian law with
  (conditional) variance 1/25 and expectation
  \begin{equation*}
    f_{0}(A,W) \equiv \frac{2}{5} (1 + W_{1} + W_{2} + W_{5} + W_{6} + W_{8} +
    A), 
  \end{equation*}
  then define $Y \equiv \expit(\tilde{Y})$.
\end{enumerate}

The  covariance matrix  $\Sigma$ induces  a loose  dependence structure.   The
components of  $\tilde{W}$ can  be gathered  in $1+4\times  (M-1)$ independent
groups, one group  consisting of $10+(M-1)$ independent  random variables, and
the other  groups consisting of  either two  or three mildly  dependent random
variables (with correlations equal to either 0.25 or 0.5).  Neither
\begin{equation*}
  \bQ_{0}     (A,W)     \equiv     \int_{[0,1]}     \frac{e^{-[\logit(u)     -
      f_{0}(A,W)]^{2}/50}}{10\sqrt{\pi} (1 - u)} du 
\end{equation*}
nor $\Psi(\Pi_{0,p,\delta})$  has a closed form  expression.  Independently of
$p$ and $\delta$, $\Psi(\Pi_{0,p,\delta}) \approx 0.0799$.

\subsection{Competing estimators}
\label{subsec:compet}

Let  $O_{1}, \ldots,  O_{n}$ be  independent draws  from $P_{0}$.   Recall the
characterization  of $\hat{\bG}_{h}$  ($h  \in \xH_{100}$)  and definition  of
$h_{n,\CV}$ given  in Section~\ref{sec:lasso:ctmle},  step~\ref{enum:one}. Let
the algorithm $\hat{\bQ}$  for the estimation of $\bQ_{0}$  consist in fitting
the                                working                               model
$\{\bQ_{\theta} :  \theta =  (\theta_{0}, \theta_{1}),  \theta_{0}, \theta_{1}
\in \xR^{8}\}$ where $\bQ_{\theta}$ is given by
\begin{equation*}
  \bQ_{\theta}  (A,W) \equiv  \Phi\left((A \theta_{1}^{\top}  + (1-A)\theta_{0}^{\top}))
    (W_{3}, \ldots, W_{10})^{\top}\right)
\end{equation*}
with $\Phi$ the  distribution function of the standard normal  law.  Note that
the working model is necessarily mis-specified, notably because of the absence
of $W_{1}$ and $W_{2}$ in the  above definition.  Recall that $\bQ_{n}^{0}$ is
obtained by training $\hat{\bQ}$  on the whole data set once  and for all.  To
emphasize,  $\hat{\bQ}$  is  never   re-trained  during  the  cross-validation
procedure.  This is consistent  with implementation the original instantiation
of the C-TMLE algorithm and of its scalable instantiations.

We   compare   the   LASSO-C-TMLE  and   LASSO-PSEUDO-C-TMLE   estimators   of
$\Psi(\Pi_{0,p,\delta})$  from Section~\ref{sec:software}  with the  following
commonly used competitors:
\begin{itemize}
\item the unadjusted estimator:
  \begin{equation*}
    \psi_n^{\text{unadj}}\equiv   \frac{\sum_{i=1}^n    A_iY_i}{\sum_{i=1}^n   A_i}   -
    \frac{\sum_{i=1}^n (1-A_i) Y_i}{\sum_{i=1}^n (1-A_i)};
  \end{equation*}
\item the so called G-comp estimator \citep{robins1986new}:
  \begin{equation*}
    \psi_n^{\text{G-comp}} \equiv \frac{1}{n} \sum_{i=1}^n \left(\bQ_n^0(1,W_i) - 
      \bQ_n^0(0,W_i)\right);
  \end{equation*}
\item          the          so         called          IPTW          estimator
  \citep{horvitz1952generalization,robins2000marginal}:
  \begin{equation*}
    \psi_n^{\text{IPTW}}  \equiv  \frac{1}{n}   \sum_{i=1}^n  \frac{  (2A_i  -
      1)Y_{i}}{\bG_{n,h_{\CV}}(A_i,W_i)};
  \end{equation*}
\item the so-called A-IPTW estimator \citep{robins1995semiparametric}:
  \begin{equation*}
    \psi_n^{\text{A-IPTW}} \equiv \frac{1}{n} \sum_{i=1}^n
    \frac{(2A_i-1)}{\bG_{n,h_{\CV}}(A_i,W_i)}
    \left(Y_i-\bQ^0_n(W_i,A_i)\right)      +     \frac{1}{n}      \sum_{i=1}^n
    \left(Q^0_n(1,W_i)-Q^0_n(0,W_i)\right); 
  \end{equation*}
\item   and   the   plain   TMLE   estimator   $\psi_{n,h_{n,\CV}}^{*}$,   see
  \eqref{eq:TMLE}.
\end{itemize}

\subsection{Outline of the structure of the simulation study}
\label{subsec:strategy}

We consider six different scenarios. In  each of them, we repeat independently
$B=200$ times the  following steps: for each $(n, p,  \delta)$ in a collection
of scenario-specific triplets,
\begin{enumerate}
\item  simulate  a  data  set  of  $n$  independent  observations  drawn  from
  $\Pi_{0,p,\delta}$;
\item   derive  the   LASSO-C-TMLE  and   LASSO-PSEUDO-C-TMLE  estimators   of
  Sections~\ref{sec:lasso:ctmle} and  \ref{sec:lasso:pseudo:ctmle} as  well as
  the competing estimators presented in Section~\ref{subsec:compet};
\item    for    the     double-robust    estimators    only,    \textit{i.e.},
  $\psi_n^{\text{A-IPTW}}$, $\psi_{n,h_{n,\CV}}^{*}$ and our two collaborative
  TMLEs,  construct  95\%   CIs  and  check  whether  or  not   each  of  them
  contains~$\Psi(\Pi_{0,p,\delta})$.
\end{enumerate}

\subsubsection*{Building  confidence  intervals  based  on  the  collaborative
  TMLEs.}

By   Corollary~\ref{theo:specific},   the    asymptotic   variances   of   our
collaborative TMLEs both write as
\begin{equation}
  \label{eq:asymp:var}
  \Var_{P_{0}}          \left[\left(D^{*}(Q_{1},           \bG_{0})          +
      \Delta(P_{1})\right)(O)\right]. 
\end{equation}
Because   $\Delta(P_{1})$    is   difficult    to   estimate,    we   estimate
\eqref{eq:asymp:var} with  the empirical variance  of $D^{*}(P_{n,\hnk}^{*})$,
\textit{i.e.}, with
\begin{equation*}
  P_{n}     D^{*}(P_{n,\hnk}^{*})^{2}     =     \frac{1}{n}     \sum_{i=1}^{n}
  D^{*}(P_{n,\hnk}^{*}) (O_{i})^{2} 
\end{equation*}
(recall that  $P_{n} D^{*}(P_{n,\hnk}^{*}) = 0$  by construction).  Therefore,
the 95\% CIs based on our collaborative TMLEs take the form
\begin{equation*}
  \psi_{n,\hnk}^{*} \pm 1.96 \sqrt{P_{n}     D^{*}(P_{n,\hnk}^{*})^{2}/n}.
\end{equation*}
We anticipate that  the asymptotic variances are  over-estimated, resulting in
CIs  that  are  too  wide.   However, we  also  anticipate  that  the  omitted
correction term  is of second  order relative to main  term, or, put  in other
words,    that    the     difference    between    \eqref{eq:asymp:var}    and
$\Var_{P_{0}} (D^{*}(Q_{1}, \bG_{0})(O))$ is small.

\subsubsection*{Six scenarios.}

The three first  scenarios investigate what happens when $\delta  = 0$ and the
number  of covariates  $p$ increases  as a  function of  sample size  $n$.  In
scenario~1,  $p  =  0.2  \times  n$  and  we  increase  $n$.   In  scenario~2,
$p  = \floor{2.83  \times  \sqrt{n}}$  and we  increase  $n$.  In  scenario~3,
$p = \floor{7.6 \times \log n}$ and  we increase~$n$.  The values of the pairs
$(p,n)$ used  in these scenarios  are presented in  Table~\ref{table:np}.  The
constants  0.2, 2.83  and 7.6  are chosen  so  that $p  = 40$  at sample  size
$n = 200$ in the three scenarios.

\setcounter{table}{-1}
\begin{table}[H]
  \centering
  \caption{Values of $p$ and $n$ in scenarios 1, 2 and 3.}
  \label{table:np}
  \begin{tabular}{l|rrrrrrrrrr}
    \hline
    $n$ & 200 & 400 & 600 & 800 & 1000 & 1200 & 1400 & 1600 & 1800 & 2000 \\
    \hline
    $p$ in scenario~1 & 40 & 80 & 120 & 160 & 200 & 240 & 280 & 320 & 360 & 400 \\ 
    $p$ in scenario~2 & 40 & 56 & 69 & 80 & 89 & 98 & 105 & 113 & 120 & 126 \\ 
    $p$ in scenario~3 & 40 & 45 & 48 & 50 & 52 & 53 & 55 & 56 & 56 & 57 \\ 
    \hline
  \end{tabular}
\end{table}

In scenarios~4 and 5, we still set $\delta  = 0$ and either keep $p$ fixed and
increase $n$  (scenario~4) or  keep $n$ fixed  and increase  $p$ (scenario~5).
Finally, in scenario~6, we keep $n$ and $p$ fixed and challenge the positivity
assumptions  that $\bG_{0}$  is bounded  away from  0 and  1 by  progressively
increasing $\delta$.

In each scenario and for all estimators, we report in a table the average bias
(bias,  multiplied by  10), standard  error (SE,  multiplied by  10) and  mean
squared    error   (MSE,    multiplied    by   100)    across   the    $B=200$
repetitions. Specifically, if  $\{\phi_{n}^{(b)} : 1 \leq b \leq  B\}$ are the
$B$ realizations of an estimator  of $\psi_{0} = \Psi(\Pi_{0,p,\delta})$ based
on   $n$   independent   draws   from   $\Pi_{0,p,\delta}$,   then   we   call
$\overline{\phi}_{n}^{1:B}  \equiv  B^{-1}  \sum_{b=1}^{B}  (\phi_{n}^{(b)}  -
\psi_{0})$                  the                 average                  bias,
$(B^{-1}             \sum_{b=1}^{B}              (\phi_{n}^{(b)}             -
\overline{\phi}_{n}^{1:B})^{2})^{1/2}$     the     standard     error,     and
$(\overline{\phi}_{n}^{1:B})^{2}  +  B^{-1} \sum_{b=1}^{B}  (\phi_{n}^{(b)}  -
\overline{\phi}_{n}^{1:B})^{2}$ the mean squared error.

We also represent  in a series of  figures how MSE, the  empirical coverage of
the 95\% CIs and their widths evolve  as the sample size (scenarios~1 to 4) or
number  of   covariates  (scenario~5)   or  parameter   $\delta$  (scenario~6)
increase. To  ease comparisons, all  similar figures  share the same  $x$- and
$y$-axes.

\subsection{Results}
\label{subsec:results}

\subsubsection*{Scenarios~1, 2, and 3: increasing $\boldsymbol{n}$ and setting
  $\boldsymbol{p = 0.2 n, \floor{2.83 \sqrt{n}}, \floor{7.6 \log n}}$.}
\label{sec:sim:scenario:one:two:three}

The results  of the three  simulation studies under  scenarios~1, 2 and  3 are
best  presented   and  commented  upon  altogether.    Figure~\ref{fig:N}  and
Table~\ref{table:N}  summarize   the  numerical  findings   under  scenario~1;
Figure~\ref{fig:Nsqrt}  and  Table~\ref{table:Nsqrt} summarize  the  numerical
findings  under scenario~2;  Figure~\ref{fig:Nlog} and  Table~\ref{table:Nlog}
summarize the numerical findings under scenario~3.

\begin{figure}[p]
  \centering
  \begin{subfigure}[t]{0.45\textwidth}
    \includegraphics[width=\textwidth]{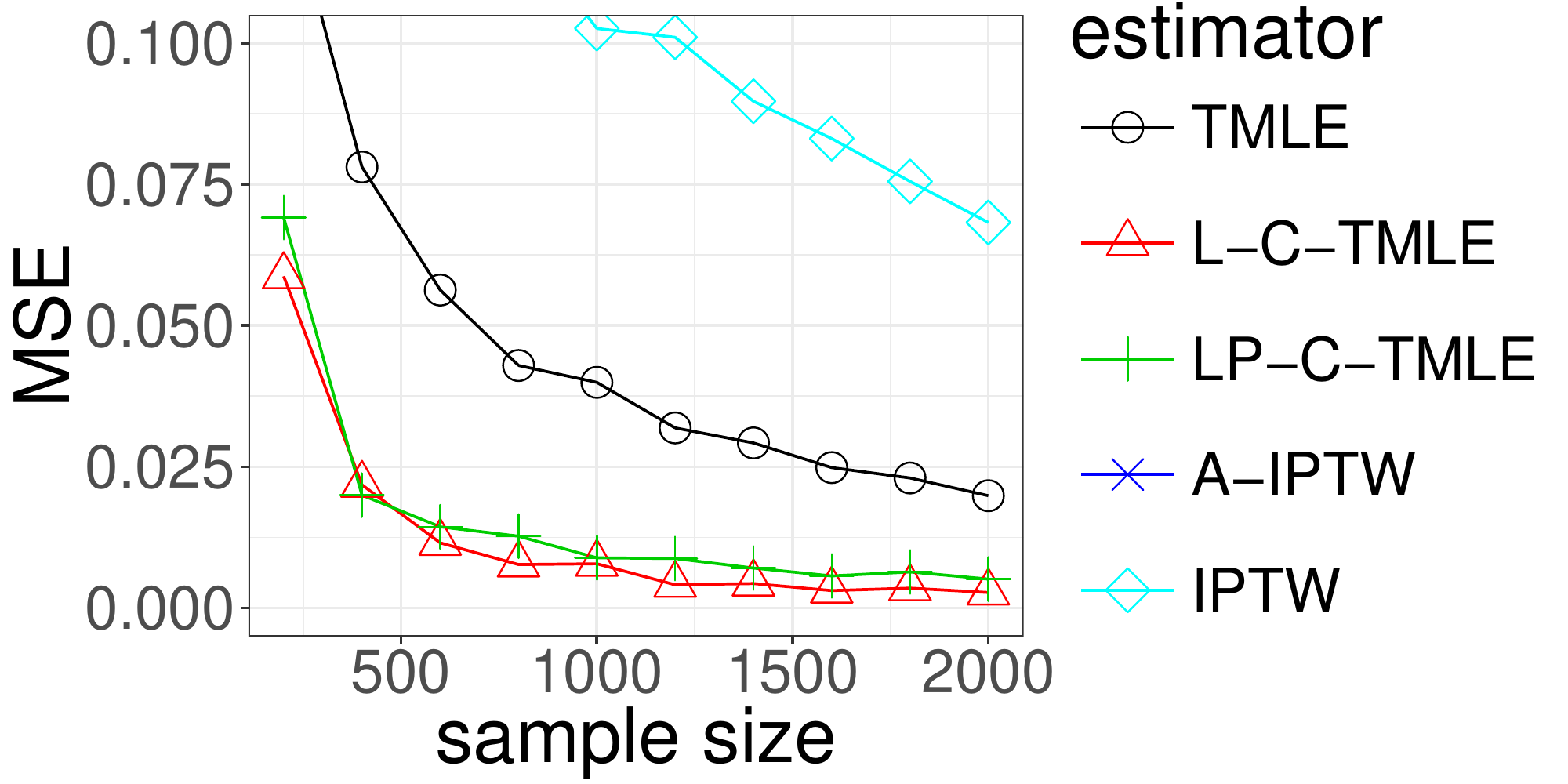}
    \caption{MSE for  five of  the seven  estimators.  The  MSE of  the A-IPTW
      estimator is so  large that it does not fit  in the picture. \textit{MSE
        is multiplied by 100.} }
    \label{fig:plot_N}
  \end{subfigure}\\
  \begin{subfigure}[t]{0.45\textwidth}
    \includegraphics[width=\textwidth]{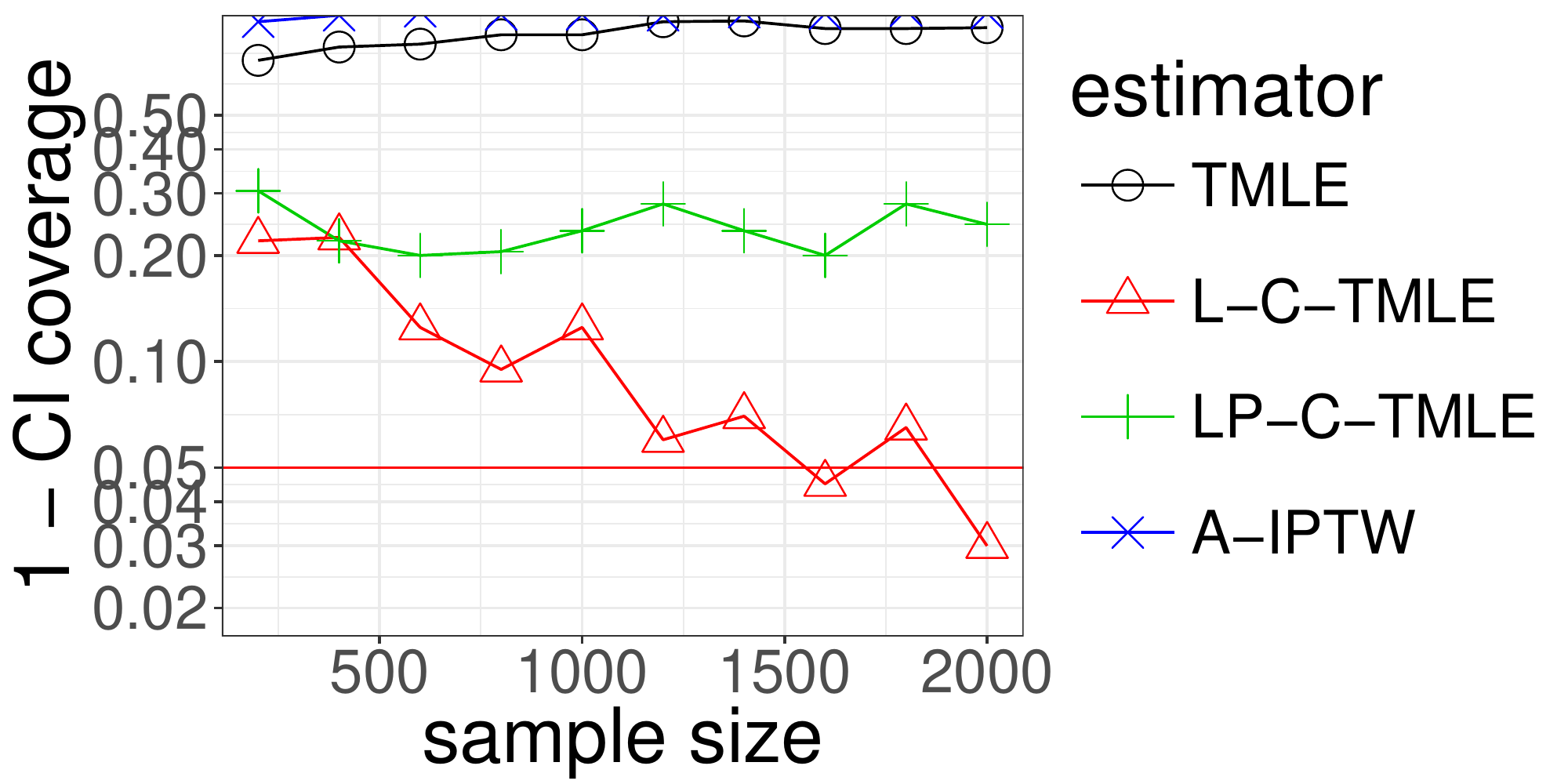}
    \caption{Coverage of 95\% CIs based on the double-robust estimators.}
    \label{fig:plot_N_ci1}
  \end{subfigure}
  \hspace{10mm}
  \begin{subfigure}[t]{0.45\textwidth}
    \includegraphics[width=\textwidth]{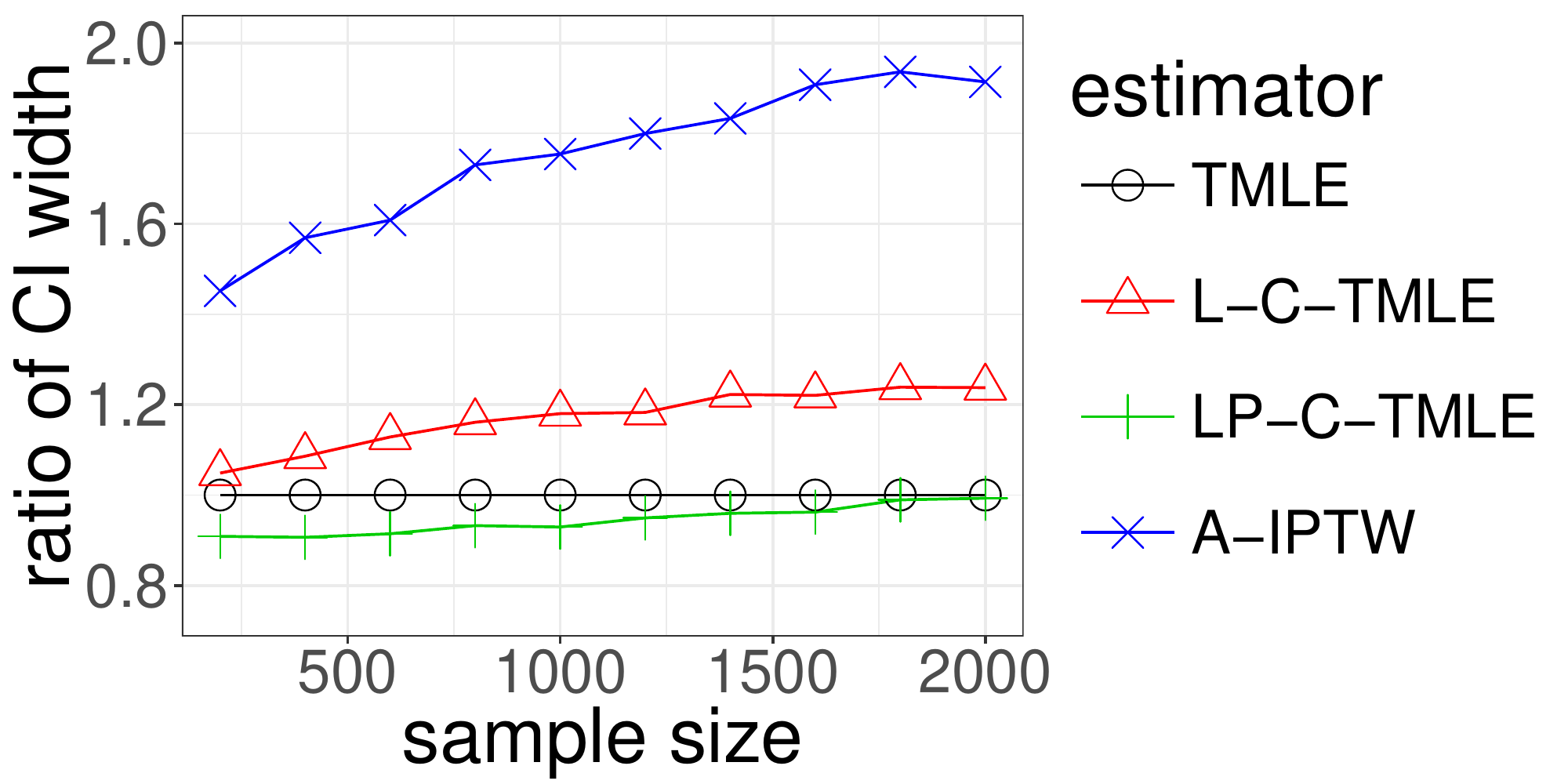}
    \caption{Relative width of 95\% CIs  based on the double-robust estimators
      w.r.t.  that of the plain TMLE, $\psi_{n,h_{n,\CV}}^{*}$.}
    \label{fig:plot_N_ci2}
  \end{subfigure}
  \caption{\textbf{Scenario~1.}  We  fix the ratio  $p/n = 0.2$,  and increase
    the sample size $n$ from 200 to 2000.}
  \label{fig:N}
\end{figure}

\begin{table}[p]
  \centering
   \begin{tabular}{l|l|rrrrrrr}
    \hline $n$  & &  $\psi_{n}^{\text{unadj}}$ &  $\psi_{n}^{\text{G-comp}}$ &
                                                                            $\psi_{n}^{\text{IPTW}}$
     &  $\psi_{n}^{\text{A-IPTW}}$ &  $\psi_{n,h_{n,\CV}}^{*}$  &  \scriptsize{L-C-TMLE} &  \scriptsize{LP-C-TMLE}
     \\\hline 
     200 & bias & 1.259 & 1.212 & 0.435 & 0.632 & 0.327 & -0.007 & 0.039 \\ 
                & SE & 0.236 & 0.152 & 0.320 & 0.137 & 0.151 & 0.242 & 0.260 \\ 
                & MSE & 1.641 & 1.491 & 0.291 & 0.418 & 0.130 & 0.059 & 0.069 \\
                & ratio &&& & 1.414 & 0.882 & 0.577 & 0.466 \\ \hline
     1000   & bias & 1.171 & 1.206 & 0.279 & 0.407 & 0.189 & 0.032 & 0.026 \\ 
                & SE & 0.110 & 0.066 & 0.157 & 0.061 & 0.064 & 0.083 & 0.091 \\ 
                & MSE & 1.382 & 1.459 & 0.103 & 0.169 & 0.040 & 0.008 & 0.009 \\
                & ratio &&& & 1.765 & 0.969 & 0.882 & 0.632 \\ \hline
     2000   & bias & 1.175 & 1.217 & 0.232 & 0.339 & 0.134 & 0.014 & 0.020 \\ 
                & SE & 0.076 & 0.047 & 0.120 & 0.050 & 0.045 & 0.050 & 0.069 \\ 
                & MSE & 1.386 & 1.483 & 0.068 & 0.118 & 0.020 & 0.003 & 0.005 \\
                & ratio &&& & 1.666 & 0.959 & 1.062 & 0.626 \\ \hline
  \end{tabular}
  \caption{\textbf{Scenario~1.} The  performance of  each estimator  at sample
    size $n  \in \{200, 1000,  2000\}$, with ratio  $p/n = 0.2$.   The columns
    named  L-C-TMLE   and  LP-C-TMLE   correspond  to  the   LASSO-C-TMLE  and
    LASSO-PSEUDO-C-TMLE estimators, respectively.   Rows \textit{ratio} report
    the ratios of  the average of the SE estimates  across the $B$ repetitions
    to the empirical SE. \textit{Bias and SE  are multiplied by 10, and MSE is
      multiplied by 100.}}
  \label{table:N}
\end{table}

\begin{figure}[p]
  \centering
  \begin{subfigure}[t]{0.45\textwidth}
    \includegraphics[width=\textwidth]{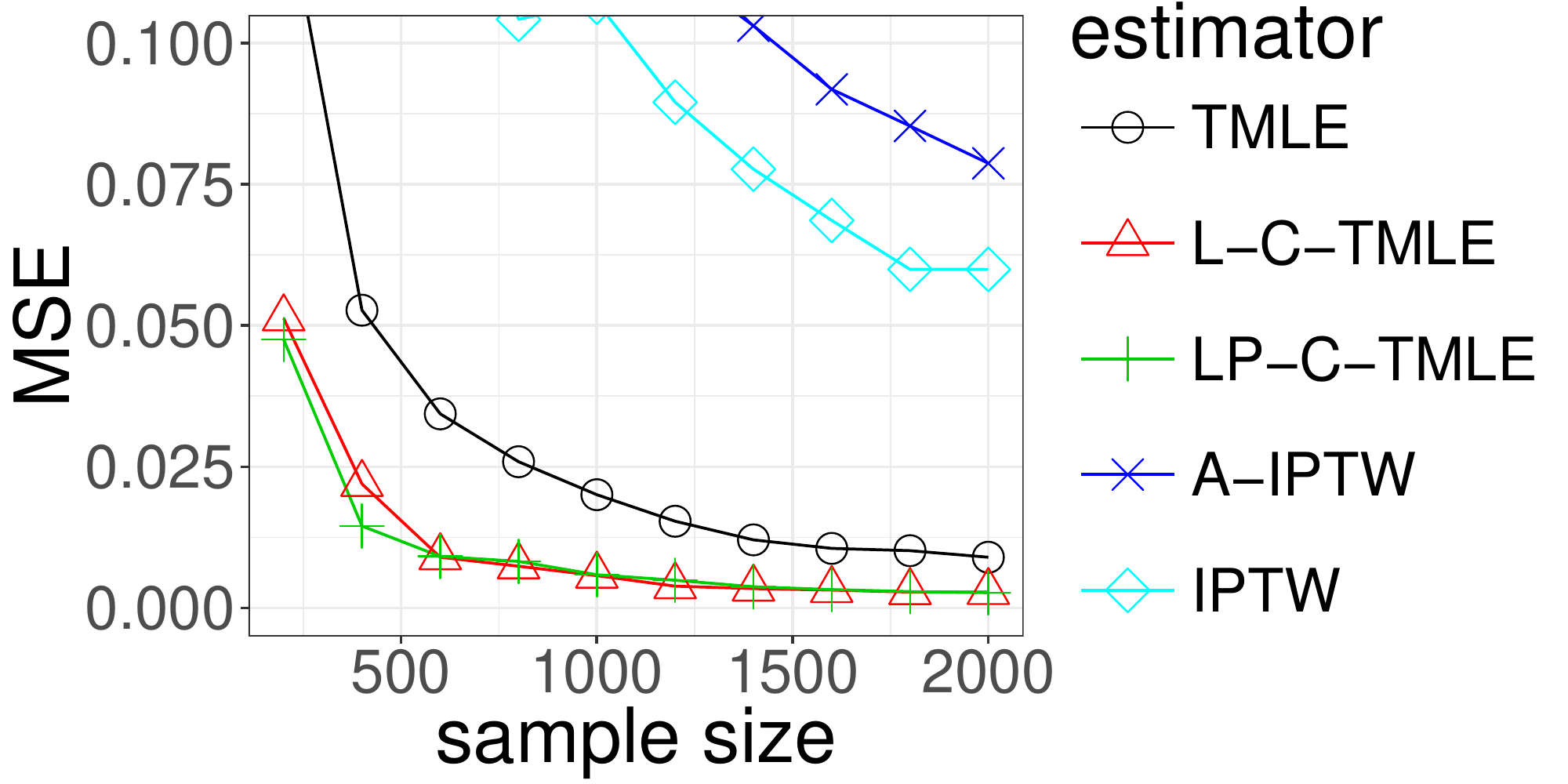}
    \caption{MSE for five of the seven estimators.}
    \label{fig:plot_Nsqrt}
  \end{subfigure}\\
  \begin{subfigure}[t]{0.45\textwidth}
    \includegraphics[width=\textwidth]{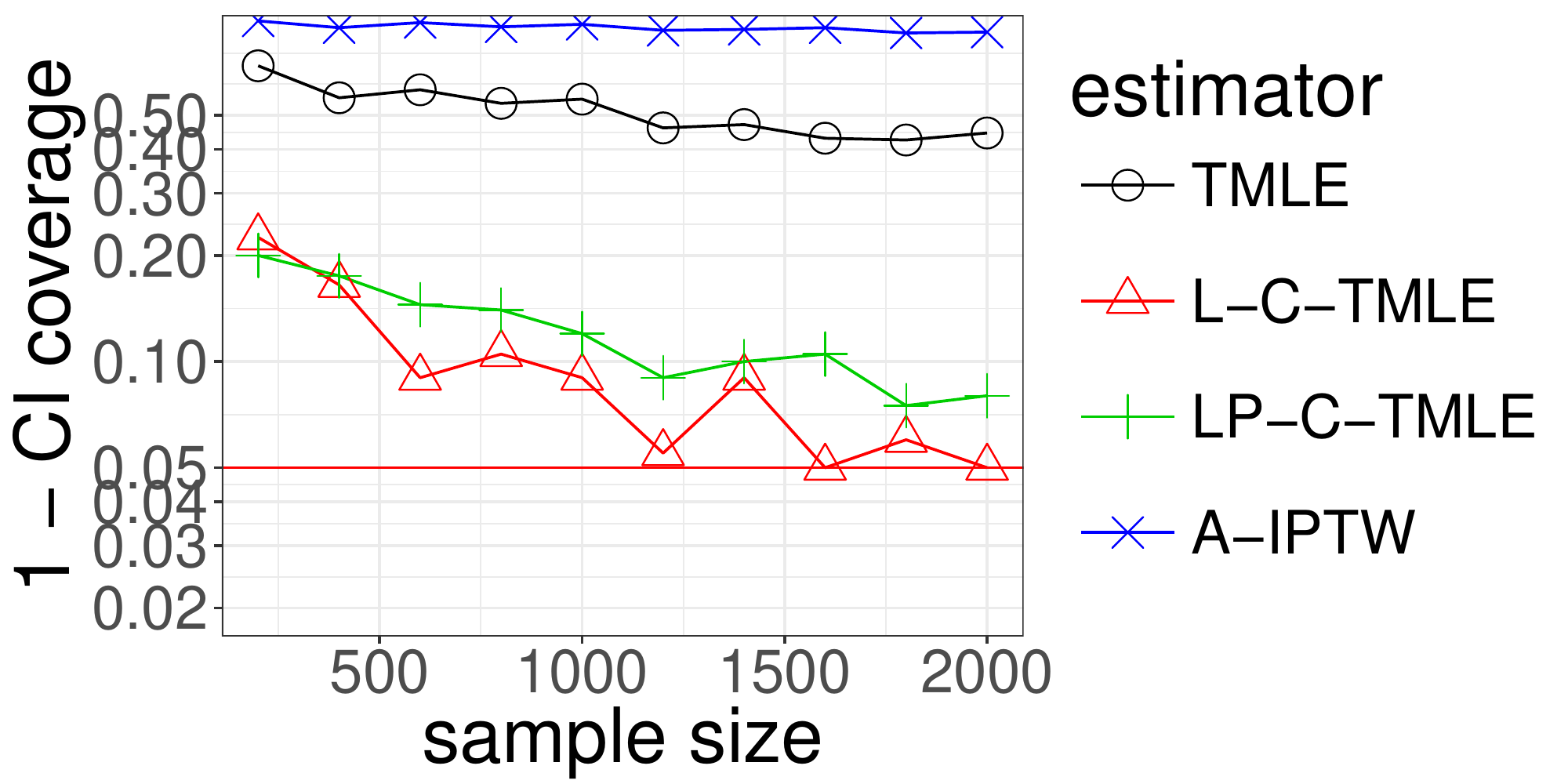}
    \caption{Coverage    of   95\%    CIs   based    on   the    double-robust
      estimators. \textit{MSE is multiplied by 100.} }
    \label{fig:plot_Nsqrt_ci1}
  \end{subfigure}
  \hspace{10mm}
  \begin{subfigure}[t]{0.45\textwidth}
    \includegraphics[width=\textwidth]{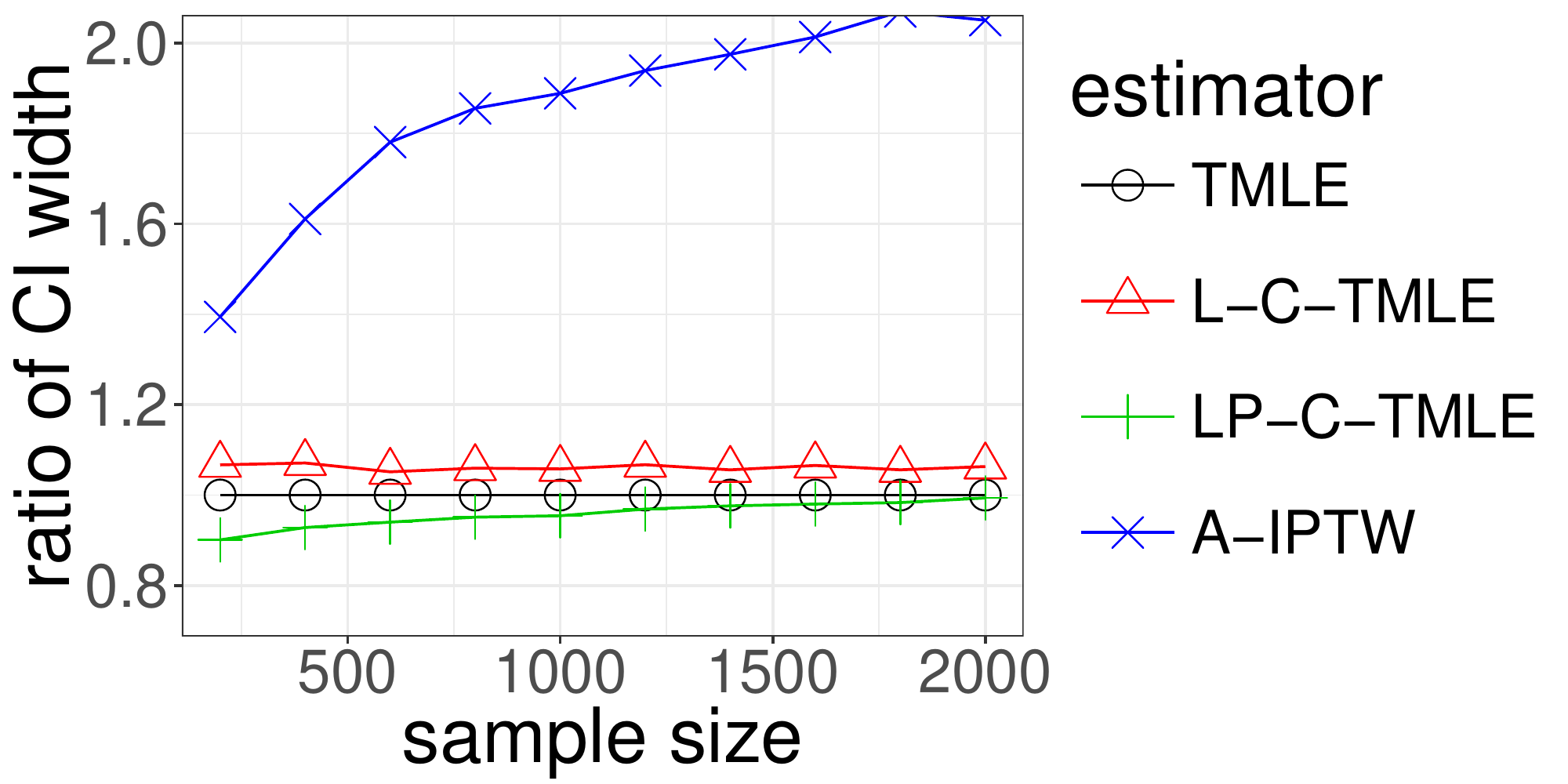}
    \caption{Relative width of 95\% CIs  based on the double-robust estimators
      w.r.t.  that of the plain TMLE, $\psi_{n,h_{n,\CV}}^{*}$.}
    \label{fig:plot_Nsqrt_ci2}
  \end{subfigure}
  \caption{\textbf{Scenario~2.} We  increase the sample  size $n$ from  200 to
    2000 and set $p = \floor{2.83\sqrt{n}}$.}
  \label{fig:Nsqrt}
\end{figure}

\begin{table}[p]
  \centering
   \begin{tabular}{l|l|rrrrrrr}
     \hline $n$  & & $\psi_{n}^{\text{unadj}}$ &  $\psi_{n}^{\text{G-comp}}$ &
                                                                               $\psi_{n}^{\text{IPTW}}$ & $\psi_{n}^{\text{A-IPTW}}$ & $\psi_{n,h_{n,\CV}}^{*}$ & \scriptsize{L-C-TMLE} &  \scriptsize{LP-C-TMLE} \\\hline
  200  & bias & 1.242 & 1.173 & 0.469 & 0.614 & 0.322 & 0.014 & 0.018 \\ 
   & SE & 0.221 & 0.157 & 0.278 & 0.135 & 0.156 & 0.226 & 0.217 \\ 
 & MSE & 1.592 & 1.401 & 0.297 & 0.395 & 0.128 & 0.051 & 0.048 \\
  & ratio &&&  & 1.377 & 0.857 & 0.630 & 0.553 \\ \hline
 1000 & bias & 1.216 & 1.214 & 0.271 & 0.361 & 0.126 & 0.003 & 0.019 \\
    & SE & 0.104 & 0.068 & 0.184 & 0.077 & 0.064 & 0.076 & 0.074 \\ 
 & MSE & 1.489 & 1.479 & 0.107 & 0.136 & 0.020 & 0.006 & 0.006 \\
  & ratio &&&  & 1.505 & 0.965 & 0.862 & 0.790 \\ \hline
  2000 & bias & 1.192 & 1.214 & 0.201 & 0.274 & 0.080 & 0.007 & 0.018 \\ 
   & SE & 0.075 & 0.051 & 0.140 & 0.061 & 0.051 & 0.053 & 0.049 \\ 
  & MSE & 1.426 & 1.477 & 0.060 & 0.079 & 0.009 & 0.003 & 0.003 \\
   & ratio &&& &  1.488 & 0.870 & 0.898 & 0.904 \\ \hline
  \end{tabular}
  \caption{\textbf{Scenario~2.} The  performance of  each estimator  at sample
    size    $n    \in   \{200,    1000,    2000\}$,    with   ratio    $p    =
    \floor{2.83\sqrt{n}}$. The columns named L-C-TMLE and LP-C-TMLE correspond
    to  the  LASSO-C-TMLE  and LASSO-PSEUDO-C-TMLE  estimators,  respectively.
    Rows \textit{ratio} report  the ratios of the average of  the SE estimates
    across the $B$  repetitions to the empirical SE.  \textit{Bias  and SE are
      multiplied by 10, and MSE is multiplied by 100.}}
  \label{table:Nsqrt}
\end{table}

\begin{figure}[p]
  \centering
  \begin{subfigure}[t]{0.45\textwidth}
    \includegraphics[width=\textwidth]{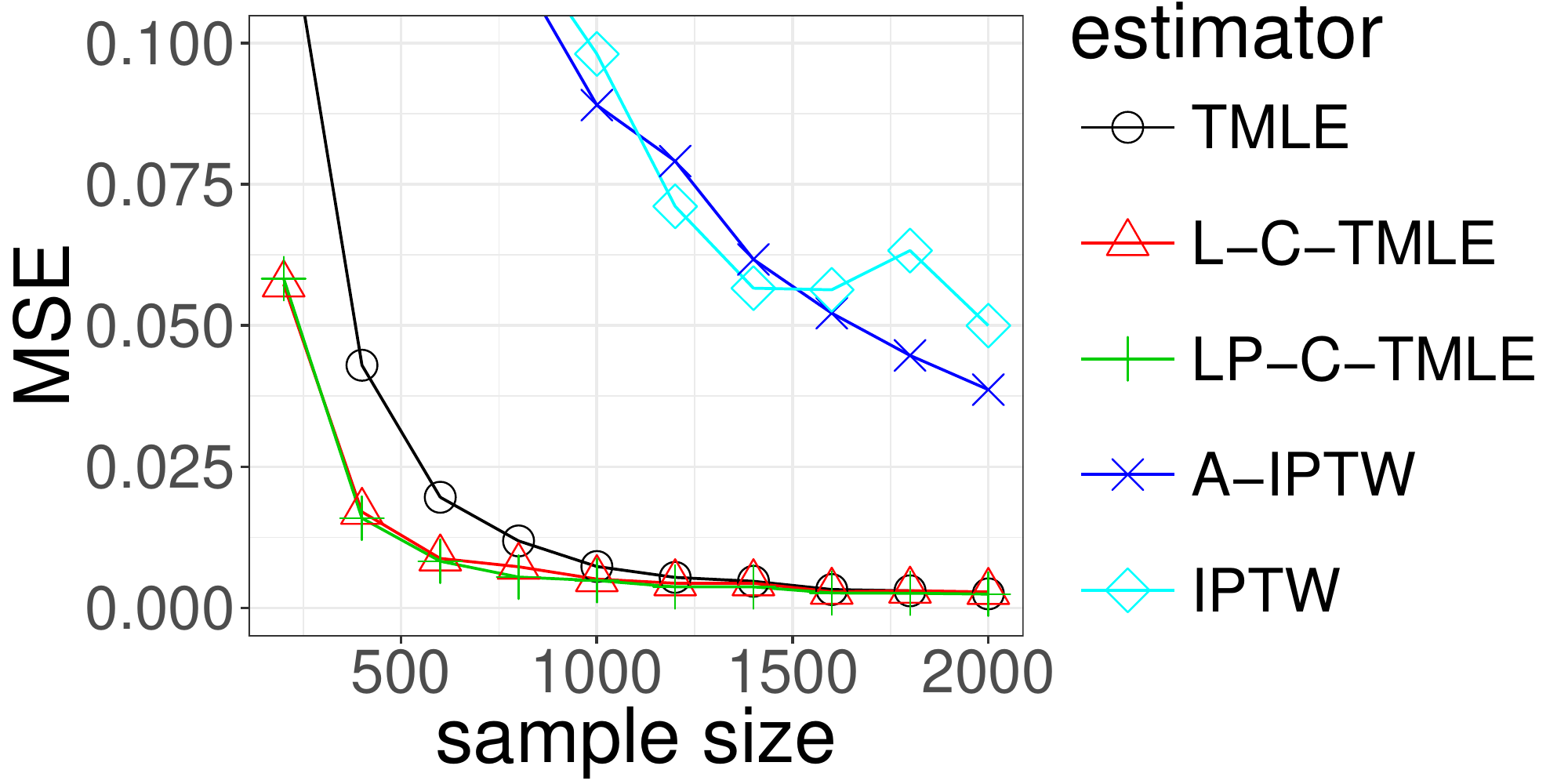}
    \caption{MSE for five  of the seven estimators.  \textit{MSE is multiplied
        by 100.} }
    \label{fig:plot_Nlog}
  \end{subfigure}\\
  \begin{subfigure}[t]{0.45\textwidth}
    \includegraphics[width=\textwidth]{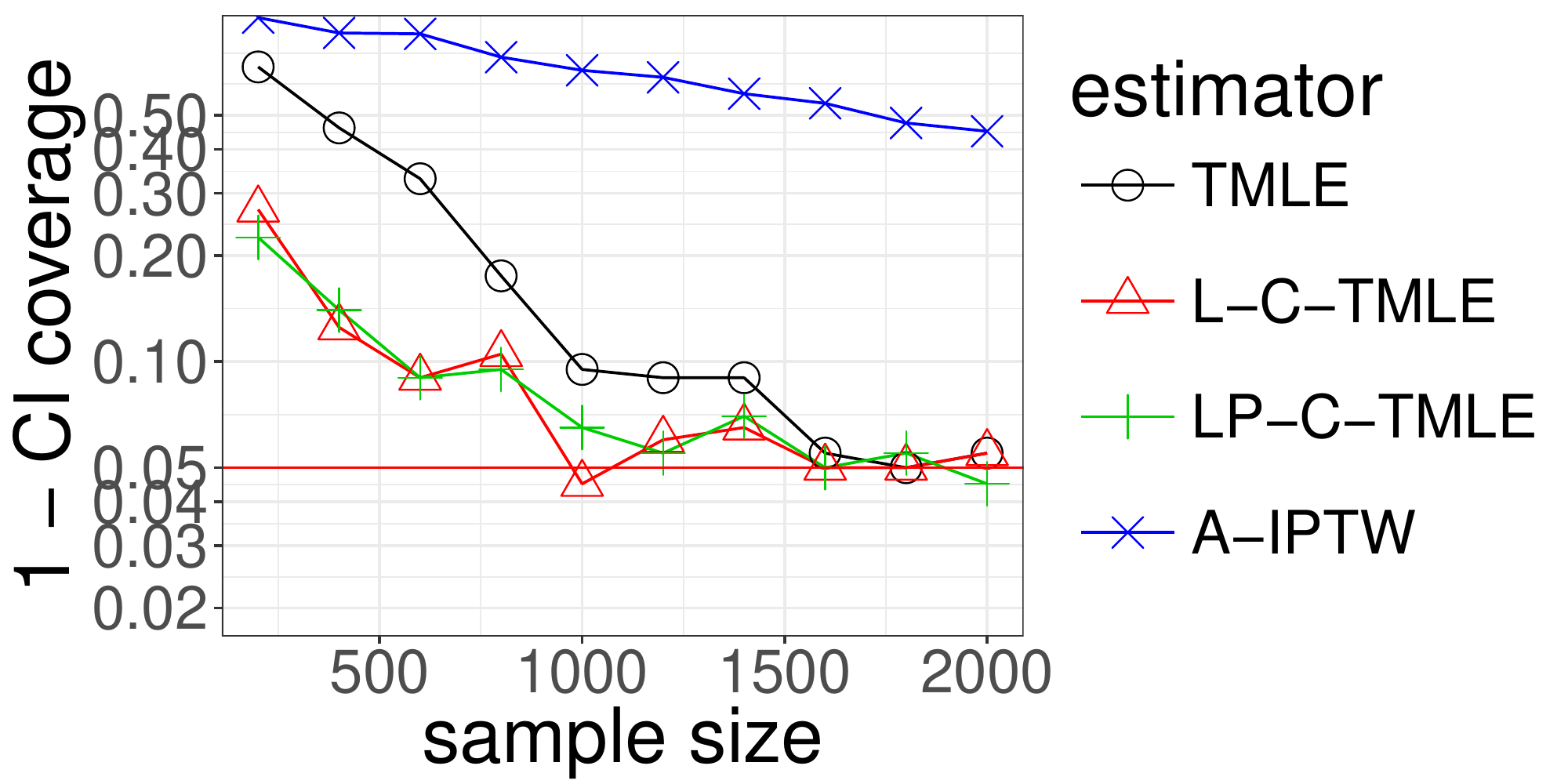}
    \caption{Coverage of 95\% CIs based on the double-robust estimators.}
    \label{fig:plot_Nlog_ci1}
  \end{subfigure}
    \hspace{10mm}
      \begin{subfigure}[t]{0.45\textwidth}
    \includegraphics[width=\textwidth]{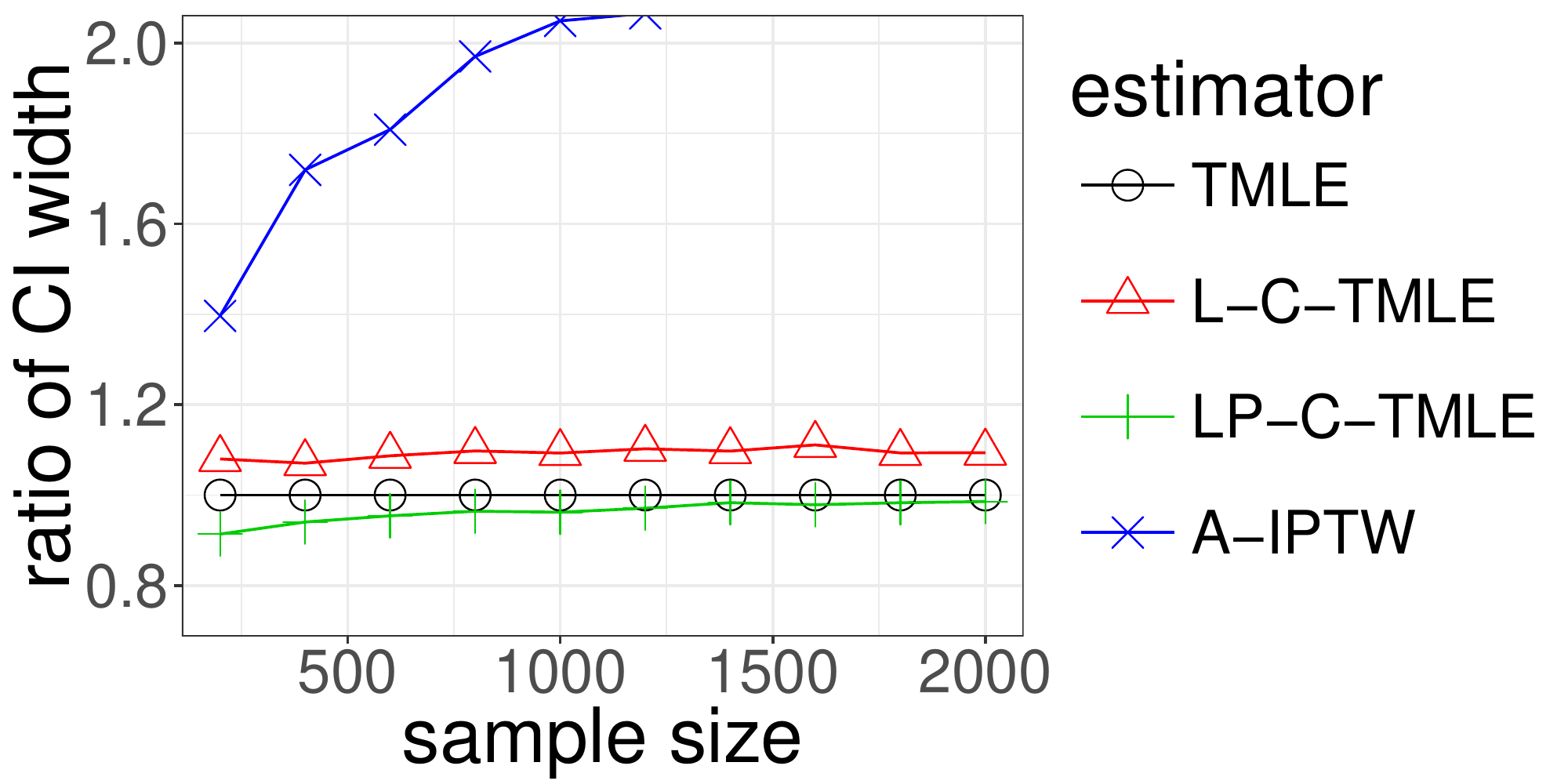}
    \caption{Relative width of 95\% CIs  based on the double-robust estimators
      w.r.t.  that of the plain TMLE, $\psi_{n,h_{n,\CV}}^{*}$.}
    \label{fig:plot_Nlog_ci2}
  \end{subfigure}
  \caption{\textbf{Scenario~3.}  We increase  the sample size $n$  from 200 to
    2000 and keep $p= \floor{7.6\log(n)}$.}
  \label{fig:Nlog}
\end{figure}

\begin{table}[p]
  \centering
   \begin{tabular}{l|l|rrrrrrr}
     \hline $n$  & & $\psi_{n}^{\text{unadj}}$ &  $\psi_{n}^{\text{G-comp}}$ &
                                                                             $\psi_{n}^{\text{IPTW}}$
     &  $\psi_{n}^{\text{A-IPTW}}$ &  $\psi_{n,h_{n,\CV}}^{*}$  &  \scriptsize{L-C-TMLE} &  \scriptsize{LP-C-TMLE}
     \\\hline
 200& bias & 1.244 & 1.190 & 0.427 & 0.634 & 0.314 & -0.009 & 0.014 \\ 
   & SE & 0.228 & 0.148 & 0.316 & 0.136 & 0.170 & 0.239 & 0.241 \\ 
 & MSE & 1.600 & 1.439 & 0.282 & 0.420 & 0.128 & 0.057 & 0.058 \\
 & ratio &&&  & 1.363 & 0.777 & 0.598 & 0.502 \\ \hline
 1000  & bias & 1.242 & 1.208 & 0.224 & 0.286 & 0.056 & -0.008 & -0.006 \\ 
   & SE & 0.114 & 0.077 & 0.219 & 0.086 & 0.065 & 0.071 & 0.069 \\ 
 & MSE & 1.555 & 1.465 & 0.098 & 0.089 & 0.007 & 0.005 & 0.005 \\
  & ratio &&& & 1.544 & 0.995 & 0.992 & 0.894 \\ \hline
 2000  & bias & 1.227 & 1.205 & 0.159 & 0.185 & 0.019 & -0.006 & -0.003 \\ 
   & SE & 0.074 & 0.050 & 0.157 & 0.066 & 0.046 & 0.053 & 0.050 \\ 
 & MSE & 1.511 & 1.453 & 0.050 & 0.039 & 0.003 & 0.003 & 0.002 \\
  & ratio &&& & 1.552 & 1.023 & 0.978 & 0.942 \\  \hline
  \end{tabular}
  \caption{\textbf{Scenario~3.} The  performance of  each estimator  at sample
    size  $n \in  \{200, 1000,  2000\}$, with  $p =  \floor{7.6\log(n)}$.  The
    columns named  L-C-TMLE and LP-C-TMLE  correspond to the  LASSO-C-TMLE and
    LASSO-PSEUDO-C-TMLE estimators, respectively.   Rows \textit{ratio} report
    the ratios of  the average of the SE estimates  across the $B$ repetitions
    to the empirical SE.  \textit{Bias and SE are multiplied by 10, and MSE is
      multiplied by 100.}}
  \label{table:Nlog}
\end{table}

Figures~\ref{fig:plot_N}, \ref{fig:plot_Nsqrt}  and \ref{fig:plot_Nlog} reveal
a general trend: MSE decreases as  sample size $n$ increases, despite the fact
that the  number of covariates $p$  also increases (at different  $n$-rates in
each   scenario).  Overall,   LASSO-C-TMLE  and   LASSO-PSEUDO-C-TMLE  perform
similarly and  better than TMLE;  TMLE outperforms IPTW, and  IPTW outperforms
A-IPTW.  Moreover,  the gap  between LASSO-C-TMLE, LASSO-PSEUDO-C-TMLE  on the
one hand  and TMLE on the  other hand \textit{(i)}~reduces as  sample size $n$
increases (in each scenario), and  \textit{(ii)}~reduces as $p$ decreases (for
each sample size $n$, across scenarios).

Judging by  Tables~\ref{table:N}, \ref{table:Nsqrt} and  \ref{table:Nlog}, the
unadjusted, G-comp, IPTW  and A-IPTW estimators are strongly  biased. The TMLE
estimator is  strongly biased too,  even for large  sample size $n$,  when the
number of covariates $p$ is not sufficiently small. Note however that the bias
of TMLE vanishes at sample size $n = 2000$ in scenario~2 (then, $p = 126$) and
at  sample  size   $n  \in  \{1000,  2000\}$  (then,  $p   \in  \{52,  57\}$).
Double-robustness   is  in   action.   In   contrast,  the   LASSO-C-TMLE  and
LASSO-PSEUDO-C-TMLE   estimators  are   both  essentially   unbiased  in   all
configurations.

Tables~\ref{table:N}, \ref{table:Nsqrt} and  \ref{table:Nlog} also reveal that
the variance  of the  TMLE estimator  tends to  be smaller  than those  of the
LASSO-C-TMLE  and LASSO-PSEUDO-C-TMLE  estimators,  those  last two  variances
being very similar. Moreover, the gap between them tends to diminish as sample
size $n$ increases, in all scenarios. 

We  now  turn  to Figures~\ref{fig:plot_N_ci1},  \ref{fig:plot_Nsqrt_ci1}  and
\ref{fig:plot_Nlog_ci1}. The LASSO-C-TMLE estimator  performs best in terms of
empirical  coverage,  followed by  the  LASSO-PSEUDO-C-TMLE,  TMLE and  A-IPTW
estimators,  in that  order.   At  moderate sample  size,  the superiority  of
LASSO-C-TMLE-based CIs over  the others is striking.  However,  even they fail
to provide  the wished  coverage except  when sample size  $n$ is  large (say,
larger than $1500$).

As a  side note, we  recall that if  $S$ is drawn  from the Binomial  law with
parameter  $(B,  q)   =  (200,  q)$,  then  $S  \leq   185$  with  probability
approximately 8\% for $q = 95\%$, 22\% for $q = 94\%$ and 43\% for $q = 93\%$.
In  this light,  an  empirical coverage  of  7.5\% is  not  that abnormal  for
$B = 200$ independent  CIs of exact coverage $q = 93\%$, and  even $q = 94\%$.
Moreover, we  anticipated to get conservative  CIs because of how  we estimate
the    asymptotic    variance    of   the    LASSO-C-TMLE    estimator,    see
Section~\ref{subsec:strategy}.   The  fact  that   the  ``ratio''  entries  of
Table~\ref{table:Nlog}, scenario~3, are that close to one for the LASSO-C-TMLE
estimator at sample size $n \in \{1000, 2000\}$ (not to mention at sample size
$n=2000$   in    Table~\ref{table:N},   scenario~1)   reveals    that   little
over-estimation of  the asymptotic variance  is at  play.  Finally, we  see in
Figures~\ref{fig:plot_N_ci2},           \ref{fig:plot_Nsqrt_ci2}           and
\ref{fig:plot_Nlog_ci2}   that  the   CIs  based   on  the   LASSO-C-TMLE  and
LASSO-PSEUDO-C-TMLE estimators are systematically  slightly wider and slightly
narrower than those based on the  TMLE estimator, all much narrower than those
based on the A-IPTW estimator.

\subsubsection*{Scenario~4:  keeping  $\boldsymbol{p}$  fixed  and  increasing
  $\boldsymbol{n}$.}

Figure~\ref{fig:fixp}  and  Table~\ref{table:fixp}   summarize  the  numerical
findings under scenario~4, where the number of covariates $p$ is set to 40 and
sample size $n$ goes from 200 to 2000 by steps of 200.

We   observe   the   same    trend   in   Figure~\ref{fig:plot_fixp}   as   in
Figures~\ref{fig:plot_N},  \ref{fig:plot_Nsqrt}  and \ref{fig:plot_Nlog}:  MSE
decreases   as  sample   size   $n$  increases;   overall,  LASSO-C-TMLE   and
LASSO-PSEUDO-C-TMLE perform  similarly and better than  TMLE; TMLE outperforms
IPTW, and  IPTW outperforms A-IPTW.   Moreover, the gap  between LASSO-C-TMLE,
LASSO-PSEUDO-C-TMLE  on the  one  hand and  TMLE on  the  other hand  vanishes
completely as sample size increases, whereas  it only got smaller in scenarios
1, 2, 3.

Judging  by Table~\ref{table:fixp},  the unadjusted,  G-comp, IPTW  and A-IPTW
estimators    are   strongly    biased    whereas    the   LASSO-C-TMLE    and
LASSO-PSEUDO-C-TMLE  estimators are  both essentially  unbiased even  at small
sample size $n=200$.  The TMLE estimator is strongly biased too at sample size
$n=200$, but much less so as $n$ increases, with no bias at all at $n = 2000$.
Again,  double-robustness is  in action.   Moreover,  there is  little if  any
difference between  the LASSO-C-TMLE, LASSO-PSEUDO-C-TMLE and  TMLE estimators
in   terms  of   bias,  SE   and   MSE  for   sample  size   $n  \in   \{1000,
2000\}$. 

Figure~\ref{fig:ci:fixp}  reveals that,  at sample  sizes $n  \geq 1000$,  the
empirical   coverage   of   the   CIs    based   on   the   LASSO-C-TMLE   and
LASSO-PSEUDO-C-TMLE estimators is satisfactory, and that CIs based on the TMLE
estimator may  provide more  coverage than wished.   By Table~\ref{table:fixp}
(\textit{ratio}  rows),  the   estimation  of  the  actual   variance  of  the
LASSO-C-TMLE   and   TMLE   estimator   is   quite   good   at   sample   size
$n \in  \{1000, 2000\}$. Apparently,  the variance of  the LASSO-PSEUDO-C-TMLE
estimator  is  under-estimated  at  sample  size  $n=1000$,  but  much  better
estimated at sample size $n=2000$.

\begin{figure}[p]
  \centering
  \begin{subfigure}[t]{0.45\textwidth}
    \includegraphics[width=\textwidth]{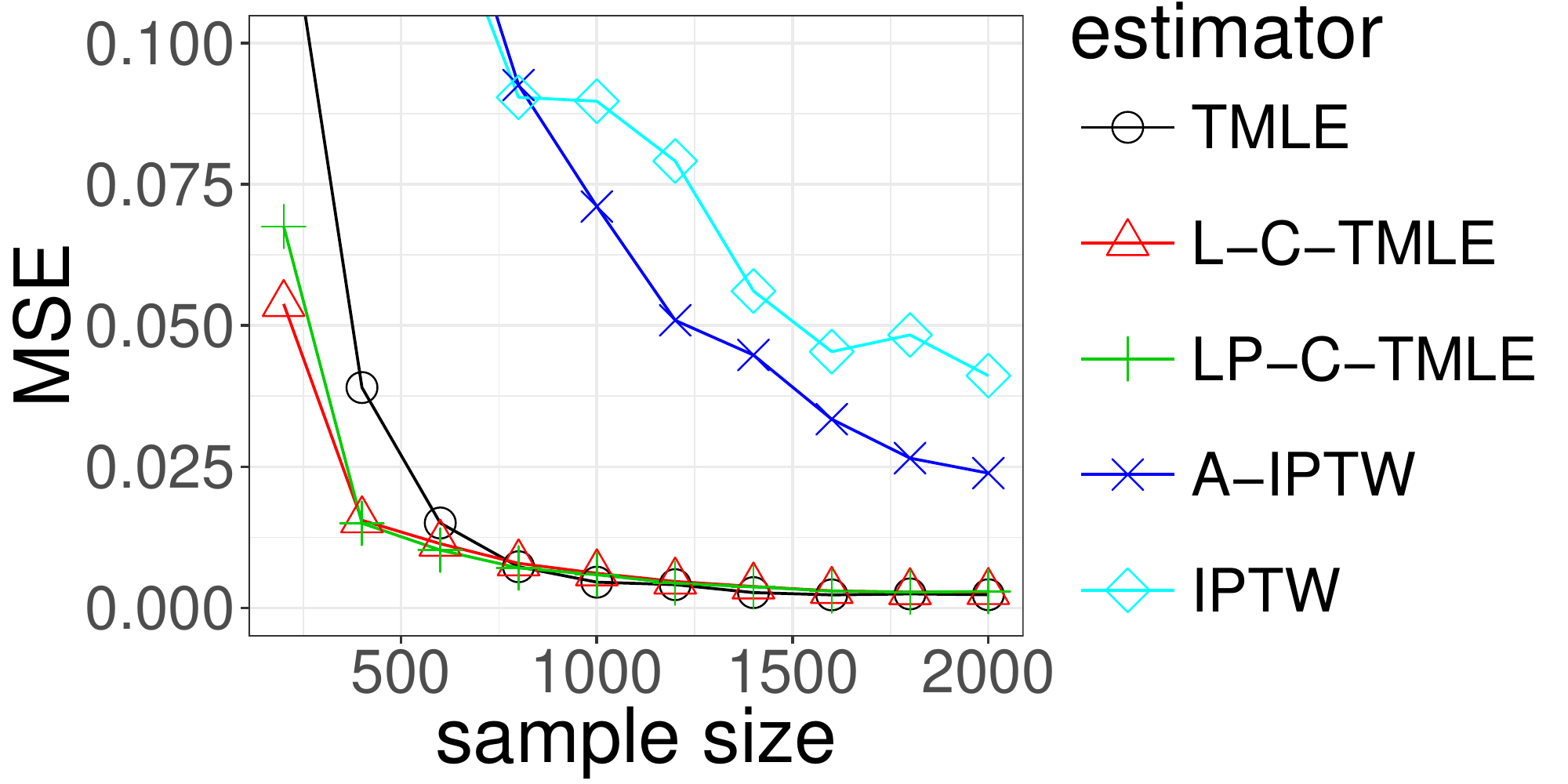}
    \caption{MSE for five  of the seven estimators.  \textit{MSE is multiplied
        by 100.} }
    \label{fig:plot_fixp}
  \end{subfigure}\\
  \begin{subfigure}[t]{0.45\textwidth}
    \includegraphics[width=\textwidth]{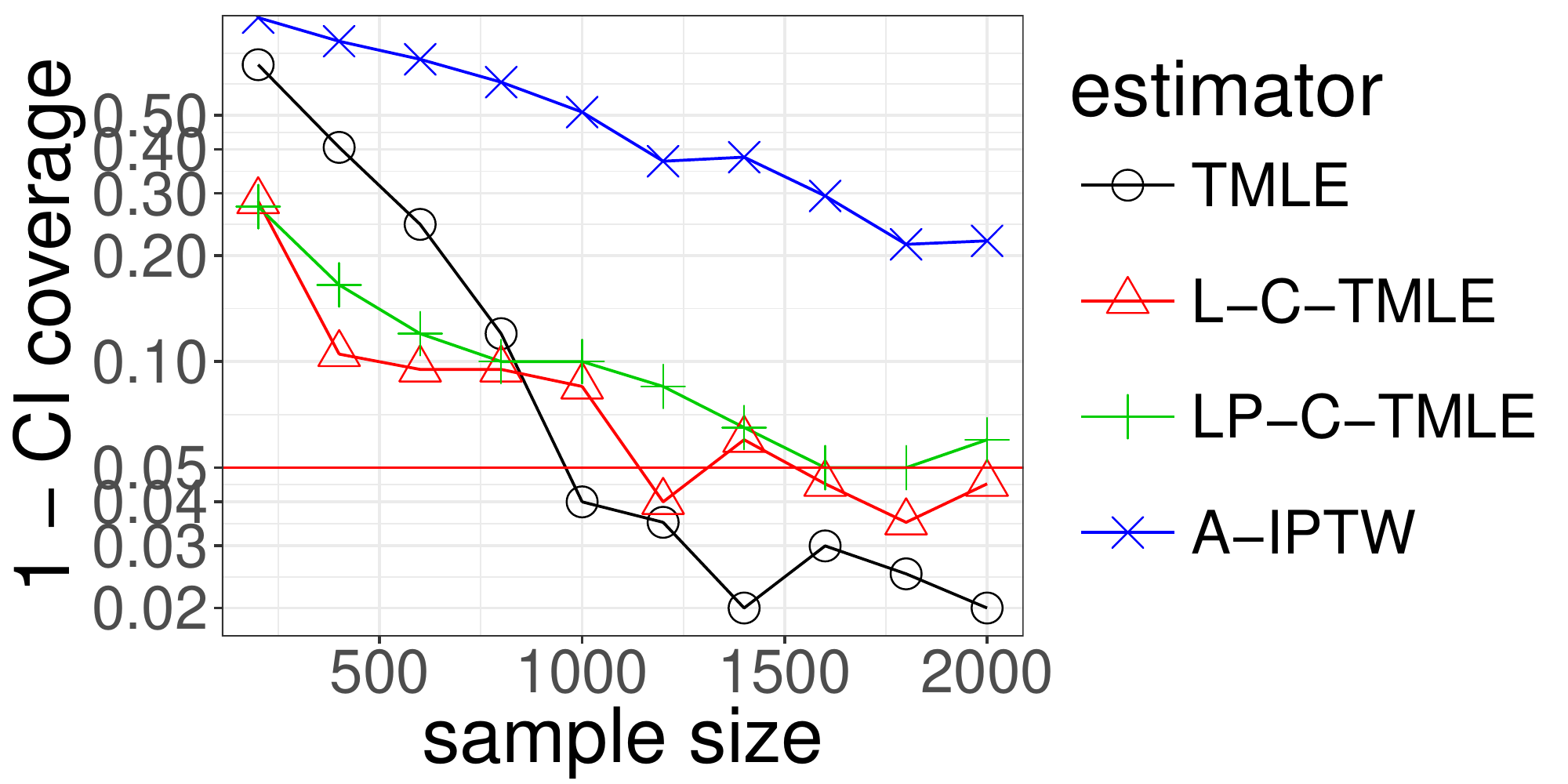}
    \caption{Coverage of 95\% CIs based on the double-robust estimators.}
    \label{fig:ci:fixp}
  \end{subfigure}
  \hspace{10mm}
  \begin{subfigure}[t]{0.45\textwidth}
    \includegraphics[width=\textwidth]{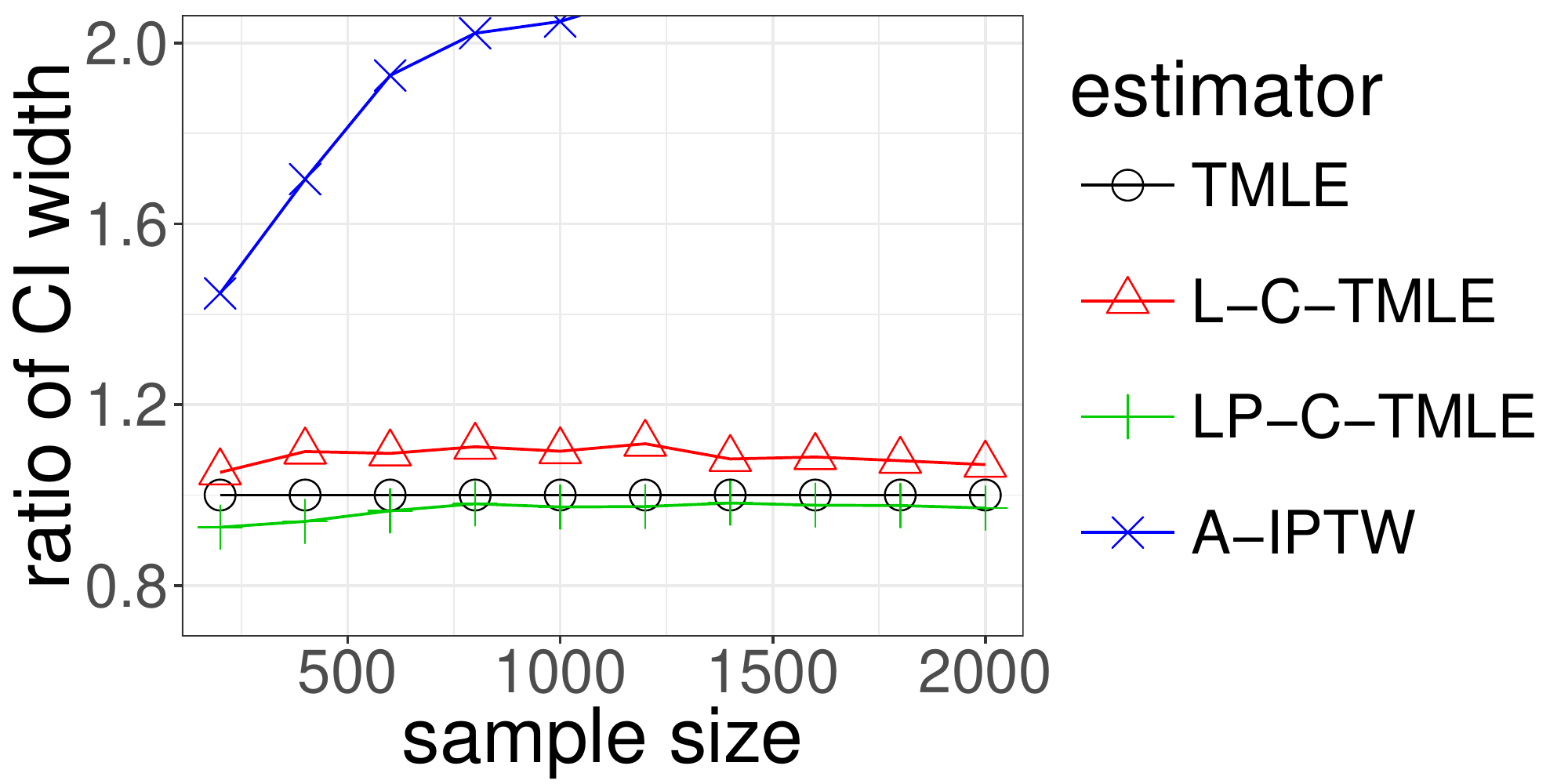}
    \caption{Relative  width  of  95\%  CIs  based  on  the  double-robust
      estimators     w.r.t.       that     of     the      plain     TMLE,
      $\psi_{n,h_{n,\CV}}^{*}$.}
  \end{subfigure}
  \caption{\textbf{Scenario~4.}  We  fix the number of  covariates $p=40$,
    and increase the sample size $n$ from 200 to 2000.}
  \label{fig:fixp}
\end{figure}

\begin{table}[p]
  \centering
  \begin{tabular}{l|l|rrrrrrrr}
    \hline $n$ & & $\psi_{n}^{\text{unadj}}$ & $\psi_{n}^{\text{G-comp}}$ & $\psi_{n}^{\text{IPTW}}$ & $\psi_{n}^{\text{A-IPTW}}$ & $\psi_{n,h_{n,\CV}}^{*}$ \scriptsize{L-C-TMLE} &  \scriptsize{LP-C-TMLE}
    \\ \hline
200 & bias & 1.286 & 1.215 & 0.485 & 0.649 & 0.317 & 0.020 & 0.020 \\ 
  & SE & 0.226 & 0.159 & 0.345 & 0.147 & 0.172 & 0.231 & 0.259 \\ 
& MSE & 1.705 & 1.501 & 0.354 & 0.443 & 0.130 & 0.054 & 0.068 \\
&  ratio &&& & 1.283 & 0.761 & 0.594 & 0.469 \\ \hline
 1000 & bias & 1.259 & 1.191 & 0.213 & 0.251 & 0.028 & -0.020 & -0.021 \\ 
  & SE & 0.109 & 0.073 & 0.211 & 0.090 & 0.062 & 0.076 & 0.074 \\ 
 & MSE & 1.597 & 1.425 & 0.090 & 0.071 & 0.005 & 0.006 & 0.006 \\
 & ratio &&& & 1.490 & 1.062 & 0.947 & 0.866 \\  \hline
 2000 & bias & 1.260 & 1.202 & 0.117 & 0.140 & -0.002 & -0.001 & -0.002 \\ 
  & SE & 0.080 & 0.049 & 0.165 & 0.066 & 0.046 & 0.049 & 0.048 \\ 
 & MSE & 1.595 & 1.448 & 0.041 & 0.024 & 0.002 & 0.002 & 0.002 \\
 & ratio &&& & 1.673 & 1.062 & 1.014 & 0.974 \\ 
  \end{tabular}
  \caption{\textbf{Scenario~4.} The  performance of  each estimator  at sample
    size  $n \in  \{200,  1000, 2000\}$,  with  $p =  40$.  The columns  named
    L-C-TMLE   and    LP-C-TMLE   correspond    to   the    LASSO-C-TMLE   and
    LASSO-PSEUDO-C-TMLE estimators, respectively.   Rows \textit{ratio} report
    the ratios of  the average of the SE estimates  across the $B$ repetitions
    to the empirical SE.  \textit{Bias and SE are multiplied by 10, and MSE is
      multiplied by 100.}}
  \label{table:fixp}
\end{table}

\subsubsection*{Scenario~5:  keeping  $\boldsymbol{n}$  fixed  and  increasing
  $\boldsymbol{p}$.}

Figure~\ref{fig:scen5}  and  Table~\ref{table:scen5} summarize  the  numerical
findings under scenario~5,  where the sample size  $n$ is set to  1000 and the
number of covariates $p$ ranges over $\{50, 75, 100, 150, 200\}$.

The take home message of Figure~\ref{fig:plot_ratio} is that, in terms of MSE,
the  LASSO-C-TMLE  and  LASSO-PSEUDO-C-TMLE  estimators  outperform  the  TMLE
estimator,      which      outperforms      the      IPTW      and      A-IPTW
estimators. Figure~\ref{fig:ci:scen5} further shows  that the above message is
also valid  when considering the  empirical coverage of  the CIs based  on the
different  estimators. As  the number  of  covariates $p$  increases, all  the
empirical  coverage  degrade.  However,  the  CIs  based on  the  LASSO-C-TMLE
estimator behave remarkably better than those based on the LASSO-PSEUDO-C-TMLE
estimator, which are themselves superior to those based on the TMLE estimator.

Examining  Table~\ref{table:scen5}  helps  to better  understand  the  general
pattern. The unadjusted,  G-comp, IPTW and A-IPTW estimators  are too strongly
biased to compete. The TMLE estimator  performs rather well when the number of
covariates  $p$  equals  50,  like the  LASSO-C-TMLE  and  LASSO-PSEUDO-C-TMLE
estimators. However, when $p \in \{100, 200\}$, then the TMLE estimator is too
biased to compete too -- even  double-robustness does not help \textit{yet} at
the moderate  sample size of  $n = 1000$.   In contrast, the  LASSO-C-TMLE and
LASSO-PSEUDO-C-TMLE   estimators  are   essentially   unbiased,  and   exhibit
relatively  small  variances  (compared  to  all  the  variances  reported  in
Tables~\ref{table:N},           \ref{table:Nsqrt},           \ref{table:Nlog},
\ref{table:fixp}. Finally, let us note that  the estimation of the variance of
the  LASSO-C-TMLE estimator  is rather  good (see  the \textit{ratio}  rows of
Table~\ref{table:scen5}),  as  opposed   to  that  of  the   variance  of  the
LASSO-PSEUDO-C-TMLE estimator.

\begin{figure}[p]
  \centering
  \begin{subfigure}[t]{0.45\textwidth}
    \includegraphics[width=\textwidth]{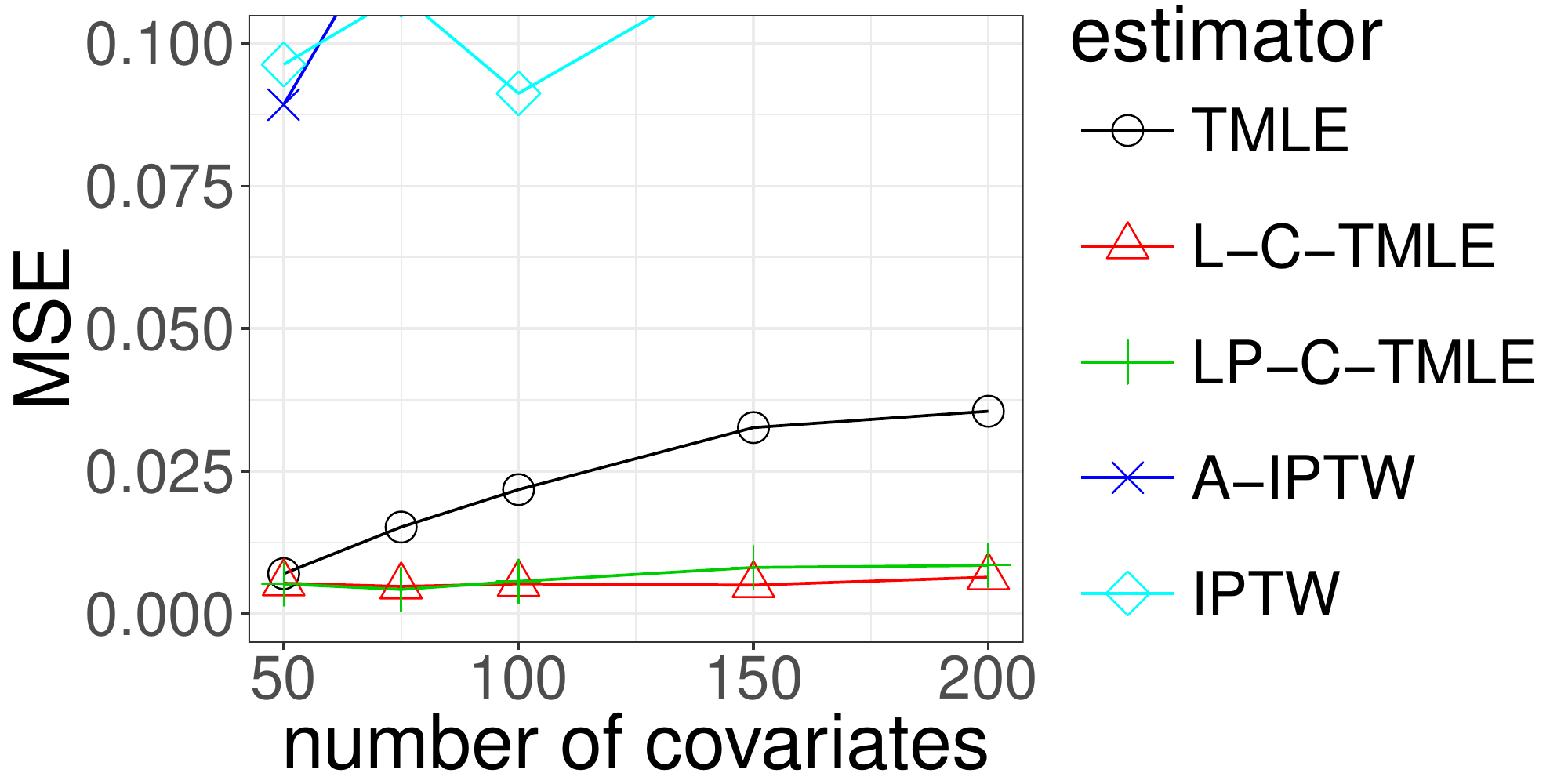}
    \caption{MSE for five  of the seven estimators.  \textit{MSE is multiplied
        by 100.} }
    \label{fig:plot_ratio}
  \end{subfigure}\\
  \begin{subfigure}[t]{0.45\textwidth}
    \includegraphics[width=\textwidth]{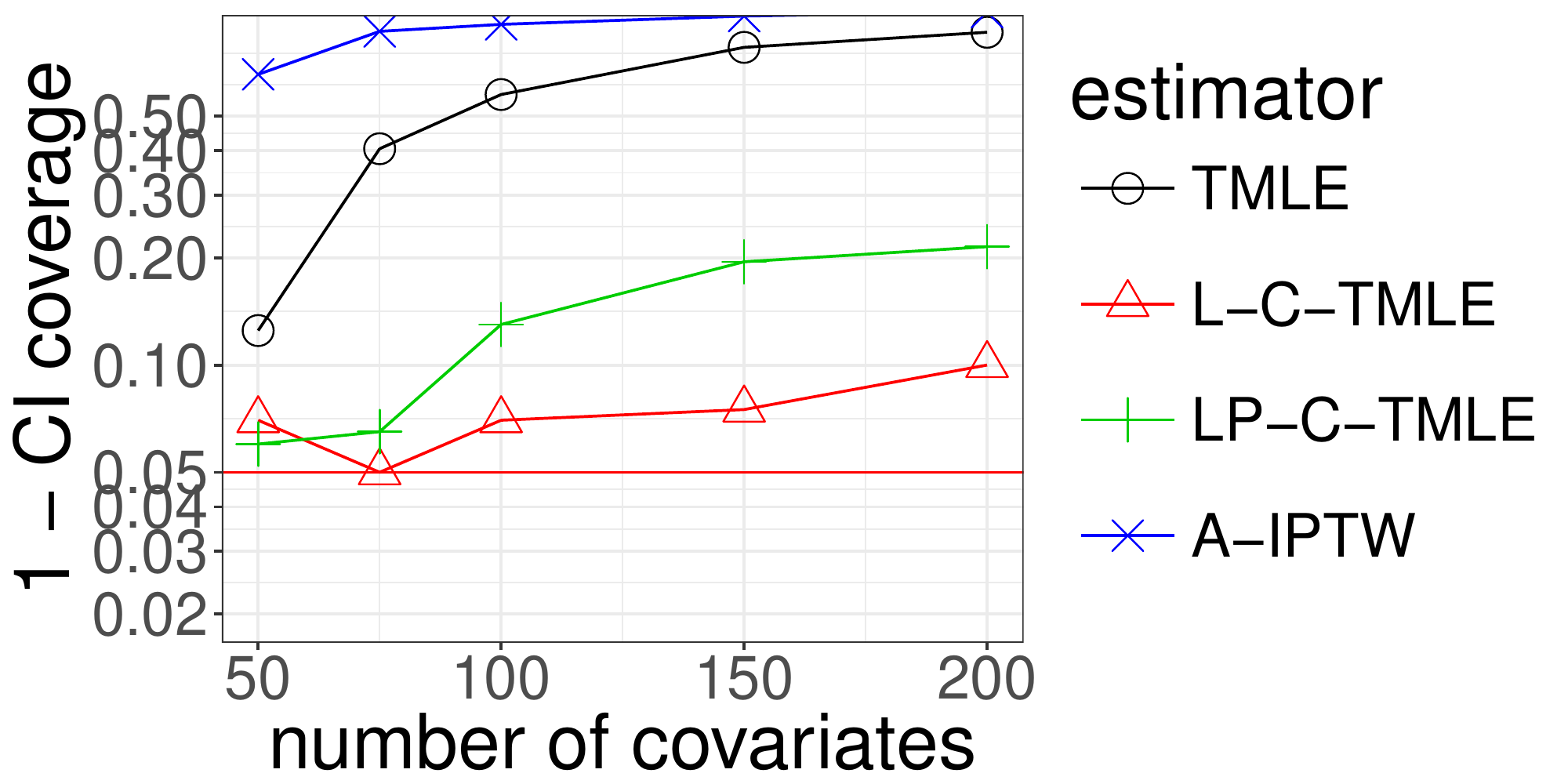}
    \caption{Coverage of 95\% CIs based on the double-robust estimators.}
    \label{fig:ci:scen5}
  \end{subfigure}
  \hspace{10mm}
  \begin{subfigure}[t]{0.45\textwidth}
    \includegraphics[width=\textwidth]{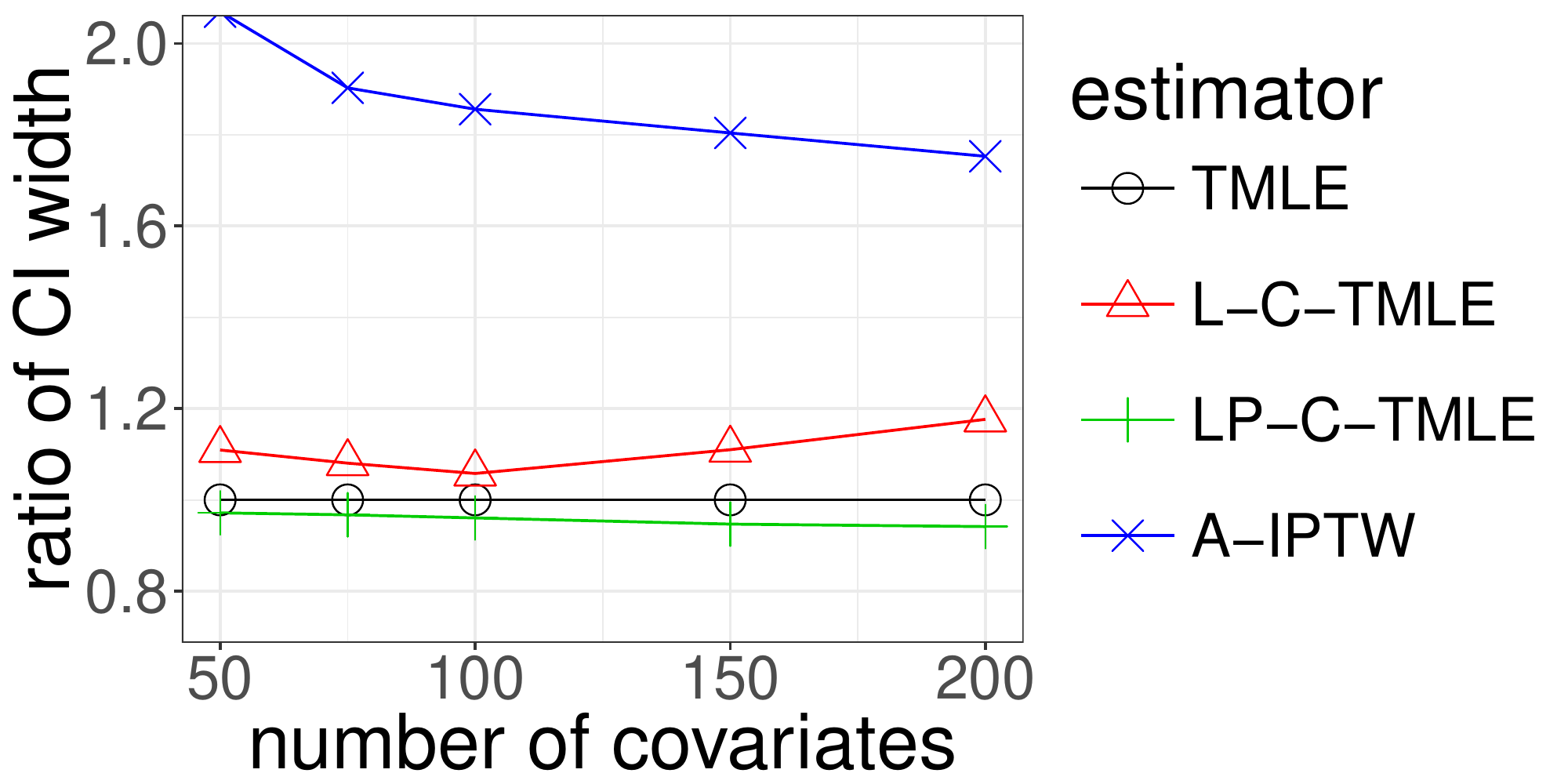}
    \caption{Relative  width  of  95\%  CIs  based  on  the  double-robust
      estimators     w.r.t.       that     of     the      plain     TMLE,
      $\psi_{n,h_{n,\CV}}^{*}$.}
  \end{subfigure}
  \caption{\textbf{Scenario~5.}  We  fix  the sample  size  $n=1000$,  and
    increase the number of covariates $p$ from 50 to 200.}
  \label{fig:scen5}
\end{figure}

\begin{table}[p]
  \centering
  \begin{tabular}{l|l|rrrrrrr}
    \hline $p$ & & $\psi_{n}^{\text{unadj}}$ & $\psi_{n}^{\text{G-comp}}$ & $\psi_{n}^{\text{IPTW}}$ & $\psi_{n}^{\text{A-IPTW}}$ & $\psi_{n,h_{n,\CV}}^{*}$ & \scriptsize{L-C-TMLE} &  \scriptsize{LP-C-TMLE}  \\
    \hline
50 & bias & 1.237 & 1.205 & 0.228 & 0.285 & 0.053 & -0.017 & -0.008 \\ 
   & SE & 0.107 & 0.072 & 0.211 & 0.090 & 0.065 & 0.071 & 0.072 \\ 
& MSE & 1.542 & 1.457 & 0.096 & 0.089 & 0.007 & 0.005 & 0.005 \\
& ratio &&&  & 1.472 & 0.983 & 0.998 & 0.864 \\ \hline
100   & bias & 1.179 & 1.199 & 0.252 & 0.357 & 0.130 & 0.005 & 0.013 \\ 
   & SE & 0.102 & 0.069 & 0.167 & 0.073 & 0.071 & 0.072 & 0.075 \\ 
& MSE & 1.402 & 1.443 & 0.091 & 0.133 & 0.022 & 0.005 & 0.006 \\
& ratio &&& & 1.561 & 0.867 & 0.896 & 0.791 \\ \hline
 200 & bias & 1.190 & 1.221 & 0.297 & 0.417 & 0.179 & 0.024 & 0.024 \\ 
   & SE & 0.107 & 0.067 & 0.159 & 0.067 & 0.060 & 0.077 & 0.089 \\ 
 & MSE & 1.428 & 1.494 & 0.114 & 0.178 & 0.036 & 0.006 & 0.009 \\
 & ratio &&& & 1.591 & 1.009 & 0.933 & 0.645 \\  \hline
  \end{tabular}
  \caption{\textbf{Scenario~5.} The  performance of  each estimator  at sample
    size $n = 1000$, with $p \in \{50, 100, 200\}$. The columns named L-C-TMLE
    and  LP-C-TMLE  correspond  to the  LASSO-C-TMLE  and  LASSO-PSEUDO-C-TMLE
    estimators, respectively.   Rows \textit{ratio}  report the ratios  of the
    average of  the SE estimates across  the $B$ repetitions to  the empirical
    SE.  \textit{Bias  and SE are multiplied  by 10, and MSE  is multiplied by
      100.} }
  \label{table:scen5}
\end{table}

\subsubsection*{Scenario  6:  keeping  $\boldsymbol{n}$  and  $\boldsymbol{p}$
  fixed and challenging the positivity assumption.}

In this  scenario, we study how  the level of posivitity  violation influences
the performance  of the estimators,  at small sample size  $n = 100$  and with
$p=50$         covariates,         by         progressively         increasing
$\delta  \in  \{0.5  +  k/10  :  0  \leq  k  \leq  15\}$.   Figure~\ref{fig:g}
illustrates how the  positivity violation is challenged.  We  recover the fact
that  $\delta  \mapsto  \Pi_{0,50,\delta}  (A =  1|W)$  is  increasing.   When
$\delta = 2$, the law is highly  skewed to 1, and the positivity assumption is
practically violated.   Figures~\ref{fig:plot_positivity}, \ref{fig:CI:scen6},
\ref{fig:ratio:CI:scen6}   and   Table~\ref{table:positivity}  summarize   the
numerical findings under scenario~6.

We see  in Figure~\ref{fig:scen6}  that, overall, the  TMLE estimator  is much
more affected than the LASSO-C-TMLE  and LASSO-PSEUDO-C-TMLE estimators by the
near violation of the positivity assumption at sample size $n = 500$, and that
the LASSO-C-TMLE and LASSO-PSEUDO-C-TMLE  estimators behave similarly in terms
of MSE  and empirical  coverage. Judging by  Table~\ref{table:positivity}, The
unadjusted, G-comp, IPTW,  A-IPTW and TMLE estimators are  too strongly biased
to  compete  with the  nearly  unbiased  LASSO-C-TMLE and  LASSO-PSEUDO-C-TMLE
estimators. The rather poor performance in  terms of empirical coverage of the
CIs based on the LASSO-C-TMLE and LASSO-PSEUDO-C-TMLE estimators may be due to
the apparent failure in estimating well their variance (see the \textit{ratio}
rows of Table~\ref{table:positivity}). 

\begin{figure}[p]
  \centering
  \begin{subfigure}[t]{0.45\textwidth}
    \includegraphics[width=\textwidth]{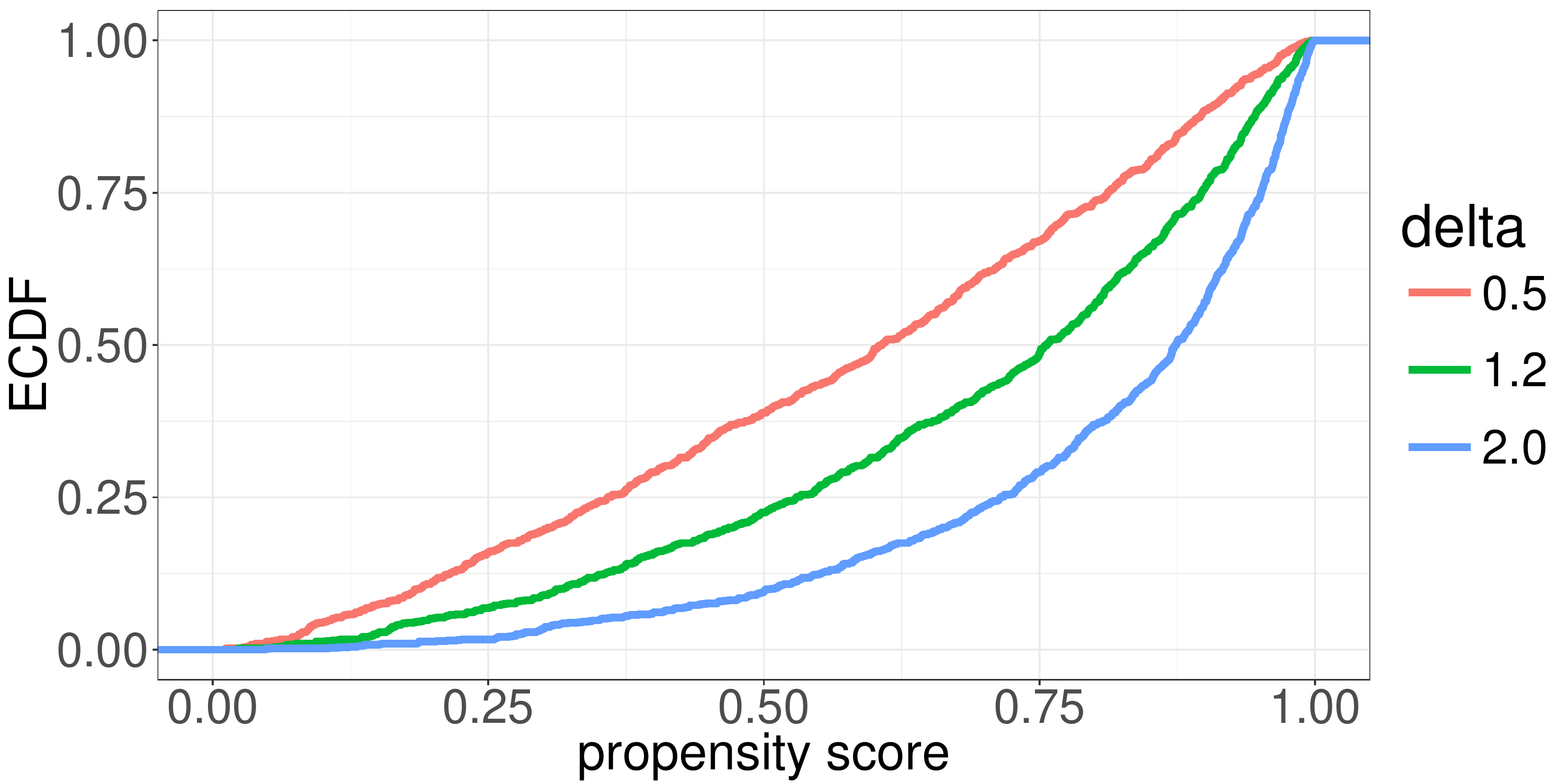}
    \caption{For  every $\delta  \in \{0.5,  1.2, 2\}$,  we simulate  $n=1000$
      observations $(W_{1},  A_{1}, Y_{1})$,  \ldots, $(W_{n},  A_{n}, Y_{n})$
      from                     $\Pi_{0,50,\delta}$,                    compute
      $\{\Pi_{0,50,\delta}(A=1 | W=W_i) : 1 \leq i \leq n\}$, and finally plot
      the corresponding empirical cumulative distribution.  }
    \label{fig:g}
  \end{subfigure}
  \hspace{10mm}
  \begin{subfigure}[t]{0.45\textwidth}
    \includegraphics[width=\textwidth]{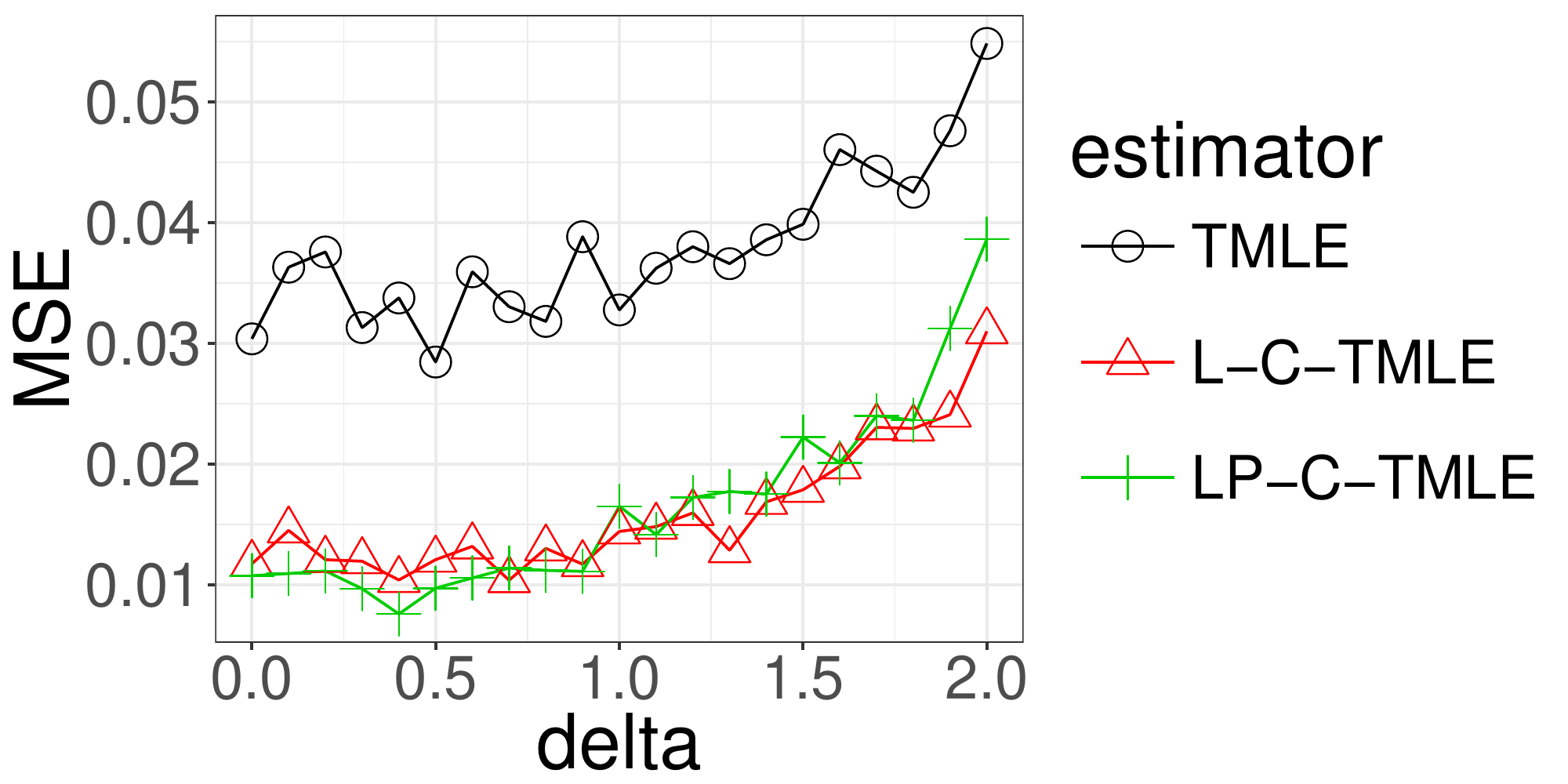}
    \caption{MSE for three of the  seven estimators. \textit{MSE is multiplied
        by 100.} }
    \label{fig:plot_positivity}
  \end{subfigure}\\
  \begin{subfigure}[t]{0.45\textwidth}
    \includegraphics[width=\textwidth]{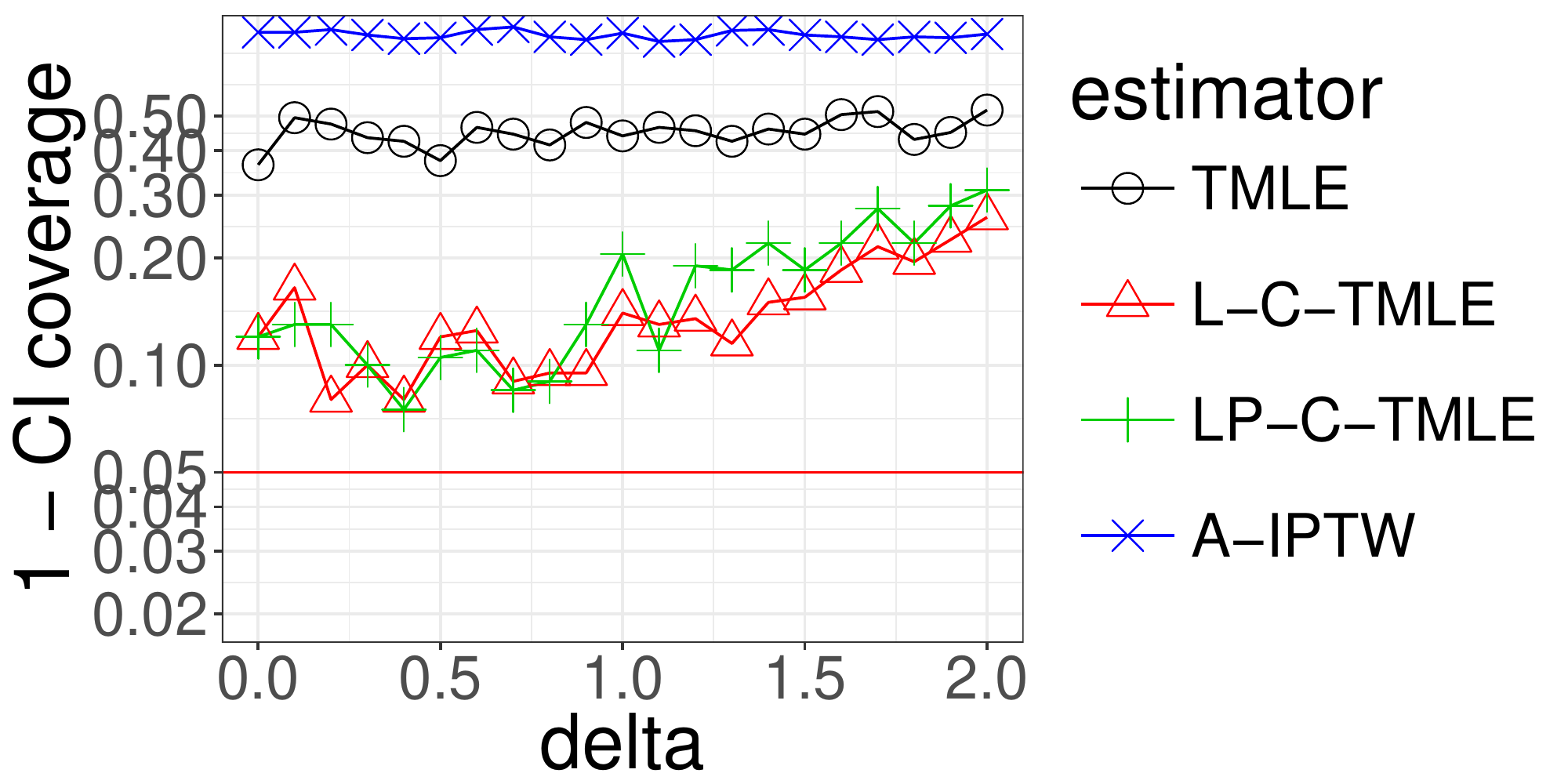}
    \caption{Coverage of 95\% CIs based on the double-robust estimators.}
    \label{fig:CI:scen6}
  \end{subfigure}
     \hspace{10mm}
    \begin{subfigure}[t]{0.45\textwidth}
    \includegraphics[width=\textwidth]{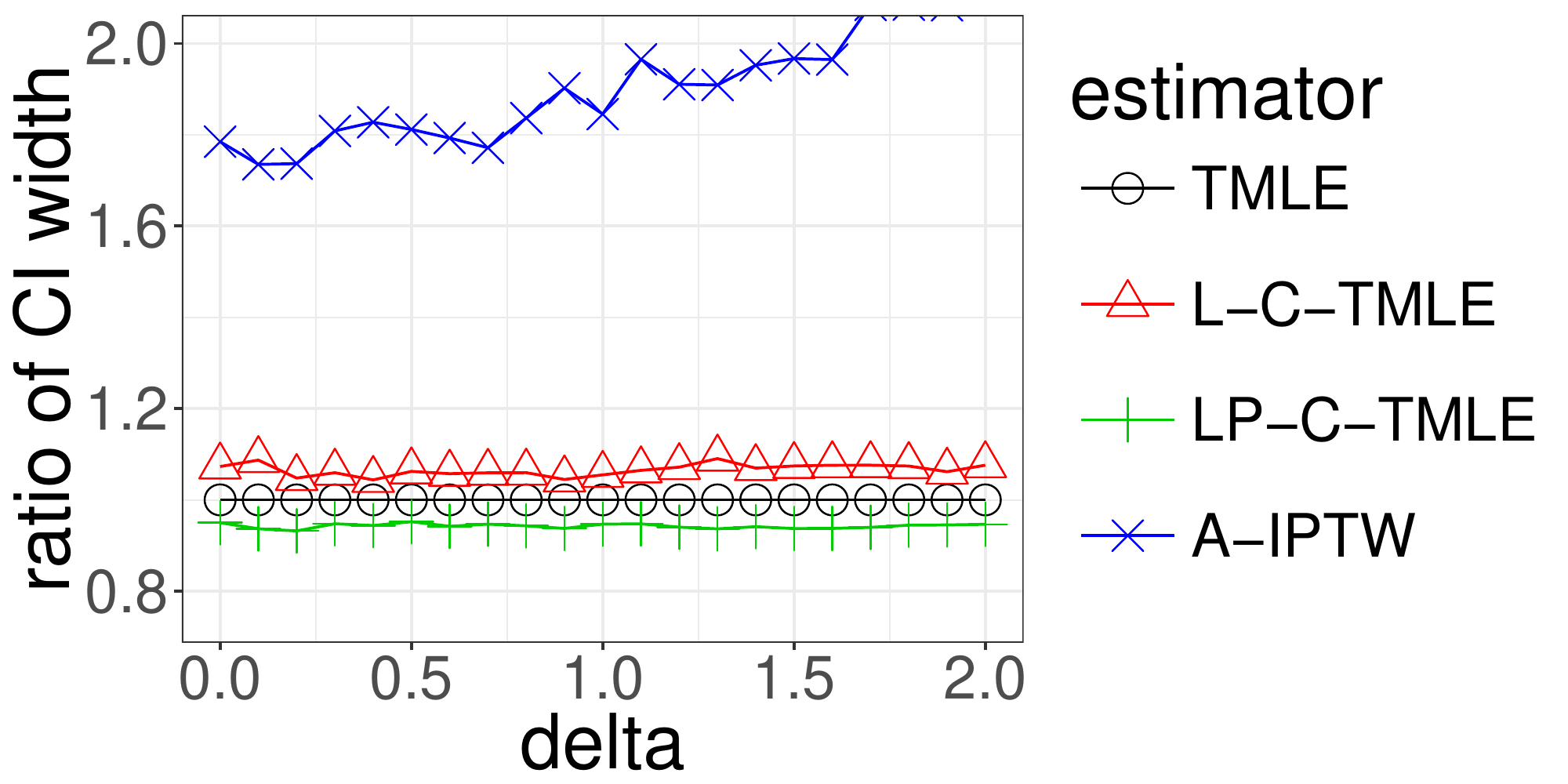}
    \caption{Relative width of 95\% CIs  based on the double-robust estimators
      w.r.t.  that of the plain TMLE, $\psi_{n,h_{n,\CV}}^{*}$.}
    \label{fig:ratio:CI:scen6}
    \end{subfigure}
    \caption{\textbf{Scenario~6.}   We   fix  $n=500,   p  =  50$,   and  vary
      $\delta \in \{0.5 + k/10: 0 \leq k \leq 15\}$.  As the MSEs for IPTW and
      A-IPTW  are  too  large,  we  only  plot the  MSEs  of  the  plain  TMLE
      $\psi_{n,h_{n,\CV}}^{*}$  and  the  two   collaborative  TMLEs  to  ease
      comparisons.}
    \label{fig:scen6}
\end{figure}

\begin{table}[p]
  \centering
  \begin{tabular}{l|l|rrrrrrrr}
    \hline $\delta$ & & $\psi_{n}^{\text{unadj}}$ & $\psi_{n}^{\text{G-comp}}$
    &     $\psi_{n}^{\text{IPTW}}$      &     $\psi_{n}^{\text{A-IPTW}}$     &
                                                                               $\psi_{n,h_{n,\CV}}^{*}$
    & \scriptsize{L-C-TMLE} &  \scriptsize{LP-C-TMLE} \\ \hline 
1.0 & bias & 1.283 & 1.252 & 0.995 & 0.484 & 0.145 & 0.006 & 0.009 \\ 
    & SE &  0.170 & 0.105 & 0.319 & 0.117 & 0.108 & 0.120 & 0.128 \\ 
    & MSE & 1.675 & 1.579 & 1.092 & 0.248 & 0.033 & 0.014 & 0.017 \\ 
        & ratio   &&& & 1.385 & 0.816 & 0.774 & 0.650 \\ \hline
    2.0 & bias & 1.391 & 1.340 & 1.620 & 0.625 & 0.185 & 0.053 & 0.063 \\ 
    & SE & 0.223 & 0.142 & 0.409 & 0.154 & 0.143 & 0.168 & 0.186 \\ 
    & MSE & 1.984 & 1.817 & 2.791 & 0.414 & 0.055 & 0.031 & 0.039 \\ 
    & ratio   &&&& 1.283 & 0.650 & 0.597 & 0.474 \\ \hline
  \end{tabular}
  \caption{\textbf{Scenario~6.} The  performance of  each estimator  at sample
    size  $n=500$, with  $p=50$ and  $\delta \in  \{1,2\}$. The  columns named
    L-C-TMLE   and    LP-C-TMLE   correspond    to   the    LASSO-C-TMLE   and
    LASSO-PSEUDO-C-TMLE estimators, respectively.   Rows \textit{ratio} report
    the ratios of  the average of the SE estimates  across the $B$ repetitions
    to the empirical SE.  \textit{Bias and SE are multiplied by 10, and MSE is
      multiplied by 100.} }
  \label{table:positivity}
\end{table}

\section{Secondary simulation study: LASSO-C-TMLE as a fine-tuning procedure} 
\label{sec:transfer}

In  this shorter  section, we  describe  a second,  less ambitious  simulation
study. Its aim is to evaluate the interest in using the LASSO-C-TMLE procedure
as a fine-tuning procedure.  Specifically we  wish to investigate, in the same
context   as  in   Section~\ref{sec:experiments},  how   the  rivals   of  the
LASSO-C-TMLE  estimator  that  also  rely   on  the  estimation  of  $\bG_{0}$
(\textit{i.e.},  the IPTW,  A-IPTW, TMLE  and LASSO-PSEUDO-C-TMLE  estimators)
perform when  they are provided  with the estimator $\bG_{n,\hnk}$  indexed by
the data-adaptive, targeted, fine-tune parameter $\hnk$. 

We thus  choose to repeat independently  $B = 200$ times  the following steps:
for each number of covariates $p \in \{100,200\}$,
\begin{enumerate}
\item simulate  a data set of  $n = 1000$ independent  observations drawn from
  $\Pi_{0,p,0}$;
\item derive the LASSO-C-TMLE estimator of Sections~\ref{sec:lasso:ctmle}
\item       derive       the      LASSO-PSEUDO-C-TMLE       estimator       of
  Section~\ref{sec:lasso:pseudo:ctmle} as  well as the competing  IPTW, A-IPTW
  and     TMLE     estimators     \textit{exactly}     as     presented     in
  Section~\ref{subsec:compet}, \textit{and also  using $\bG_{n,\hnk}$ in place
    of  $\bG_{n,h_{n}}$  (LASSO-PSEUDO-C-TMLE)  and  $\bG_{n,h_{n,\CV}}$  (the
    others)}.
\end{enumerate}

The results are reported in  Table~\ref{tab:transfer}. A clear pattern emerges
from Table~\ref{tab:transfer}:  the bias is systematically  reduced when using
$\bG_{n,\hnk}$   in   place   of   $\bG_{n,h_{n,\CV}}$   or   $\bG_{n,h_{n}}$.
Nevertheless,  the  MSE of  the  IPTW  estimator  increases, with  a  two-fold
increase when  the number of  covariates $p =  200$.  In contrast,  the A-IPTW
estimator benefits  more from the substitution,  with a stark decrease  of the
MSE on  top of that of  the bias, the latter  being still far too  large. This
makes only more  remarkable the fact that the TMLE  estimator greatly benefits
from  the substitution  on all  fronts,  bias and  MSE. On  the contrary,  the
benefit for the  LASSO-PSEUDO-C-TMLE estimator is not  convincing. In summary,
$\bG_{n,\hnk}$ is targeted  even ``out of context'',  \textit{i.e.}, even when
it is  used to  build a plain  TMLE estimator as  opposed to  the full-fledged
C-TMLE estimator.

\begin{table}[p]
  \centering
  \begin{tabular}{l|l|rrrrrr}
    \cline{1-6}
    $p$    &  &  $\psi_{n}^{\text{IPTW}}$   &    $(\psi_{n}^{\text{IPTW}})'$    &  $\psi_{n}^{\text{A-IPTW}}$ & $(\psi_{n}^{\text{A-IPTW}})'$ \\ 
    \cline{1-6}
    100 & bias & 0.252 & 0.078 & 0.357 & 0.151 & \\ 
           & SE & 0.167 & 0.325 & 0.073 & 0.167 \\ 
           & MSE  & 0.091 & 0.112 & 0.133 & 0.050 \\
    200 & bias  & 0.297 & 0.106 & 0.417 & 0.152 \\ 
           & SE  & 0.159 & 0.480 & 0.067 & 0.150 \\ 
           & MSE  & 0.114 & 0.242 & 0.178 & 0.045 \\ 
    \hline
    \hline
    $p$ & & $\psi_{n,h_{n,\CV}}^{*}$ & $\psi_{n,\hnk}^{*}$ & \scriptsize{L-C-TMLE} &  \scriptsize{LP-C-TMLE}
                                                           & \scriptsize{LP-C-TMLE}$'$ &\\ \hline
100 & bias & 0.130 & 0.017 & 0.005 & 0.013 & -0.000 \\ 
   & SE & 0.071 & 0.101 & 0.072 & 0.075 & 0.094 \\ 
   & MSE & 0.022 & 0.011 & 0.005 & 0.006 & 0.009 \\
  \hline
  200 & bias & 0.179 & -0.037 & 0.024 & 0.024 & -0.072 \\ 
   & SE & 0.060 & 0.140 & 0.077 & 0.089 & 0.148 \\ 
   & MSE & 0.036 & 0.021 & 0.006 & 0.009 & 0.027 \\ 


\end{tabular}
\caption{\textbf{Using   LASSO-C-TMLE  as   a  fine-tuning   procedure.}   The
  performance   of    each   estimator   at   sample    size   $n=1000$   with
  $p \in \{100,  200\}$. The prime symbol indicates the  use of $\bG_{n,\hnk}$
  as  an   estimator  of   $\bG_{0}$  in   place  of   $\bG_{n,h_{n,\CV}}$  or
  $\bG_{n,h_{n}}$.  \textit{Bias and  SE  are  multiplied by  10,  and MSE  is
    multiplied by 100.}}
  \label{tab:transfer}
\end{table}

\section{Discussion}
\label{sec:discuss}

We study the inference of the value of a smooth statistical parameter at a law
$P_{0}$ from which we sample $n$ independent observations, in situations where
\textit{(i)}  we  rely  on  a  machine  learning  algorithm  fine-tuned  by  a
real-valued parameter $h$ to estimate  the $G$-component $\bG_{0}$ of $P_{0}$,
possibly  consistently,  and  \textit{(ii)}  the   product  of  the  rates  of
convergence of  the estimators of  the $Q$-  and $G$-components of  $P_{0}$ to
their targets may be slower than the convenient $o(1/\sqrt{n})$.  A plain TMLE
with an $h$ chosen by cross-validation  would typically not lend itself to the
construction  of a  CI,  because  the selection  of  $h$  would trade-off  its
empirical bias with something akin to  the empirical variance of the estimator
of $\bG_{0}$ as opposed to that of  the TMLE.  We develop a collaborative TMLE
procedure that succeeds in achieving  the relevant trade-off: under high-level
empirical processes conditions, and if there exists an oracle $h$ that makes a
bulky   remainder   term  asymptotically   Gaussian,   then   the  C-TMLE   is
asymptotically  Gaussian hence  amenable to  building a  CI provided  that its
asymptotic variance can be estimated too.

The  construction of  the  C-TMLE  and the  main  result  about its  empirical
behavior are illustrated  with the inference of the  average treatment effect,
both theoretically and numerically. In the simulation study, the $G$-component is
estimated by  the LASSO, and  $h$ is the bound  on the $\ell^{1}$-norm  of the
candidate  coefficients.  Overall,  the  resulting  LASSO-C-TMLE estimator  is
superior to all  its competitors, including a plain  TMLE estimator. Evaluated
in terms of empirical bias, standard error, mean squared error and coverage of
CIs, the  superiority is striking in  small and moderate sample  sizes.  It is
also strong  when the number of  covariates increases, or when  the positivity
assumption  is  increasingly  challenged,   thus  making  the  inference  task
progressively even more delicate.

The simulation study suggests that the CIs  based on the C-TMLE do not provide
the wished coverage, especially in small  sample sizes. Obviously, this may be
explained by the need for the C-TMLE estimator to reach its asymptotic regime.
More subtly,  this may also  be related to high-level  assumption \textbf{A4},
that  states the  existence of  an oracle  $h$ making  a bulky  remainder term
asymptotically Gaussian.  The assumption may fail to hold in practice. We will
devote  future  research  to  understanding  better  \textbf{A4}  and  finding
strategies to avoir relying on it.

In conclusion, we believe that the present study further demonstrates the high
versatility  and   potential  of  the  collaborative   targeted  minimum  loss
estimation methodology. For (relative) simplicity, we focused on the inference
of  a   smooth,  real-valued   statistical  parameter  from   independent  and
identically  distributed  observations,  assuming that  the  machine  learning
algorithm is fine-tuned by a  real-valued parameter.  Our instantiation of the
collaborative targeted minimum loss estimation  methodology can be extended to
other  statistical parameters,  sampling schemes,  and fine-tuning  of machine
learning algorithms.

\appendix

\section{Proofs}

Sections~\ref{subsec:lem:A4},                    \ref{subsec:theo:high:level},
\ref{subsec:lem:easy}   and   \ref{subsec:theo:specific}  respectively   prove
Lemma~\ref{lem:A4},  Theorem~\ref{theo:high:level},  Lemma~\ref{lem:easy}  and
Corollary~\ref{theo:specific}.

\subsection{Proof of Lemma~\ref{lem:A4}}
\label{subsec:lem:A4}

\begin{proof}
  By \eqref{eq:remainder:intro},
  \begin{eqnarray*}
    \Psi(P_{n,h_{n}}^{*}) - \Psi(P_{0})   +  P_{0}   D^{*}  (Q_{n,h_{n}}^{*},
    \bG_{n,h_{n}}) 
    &=& \Rem_{20} (Q_{n,h_{n}}^{*},  \bG_{n,h_{n}}),\\
    \Psi(P_{0}) - \Psi(P_{0}) + P_{0} D^{*} (Q_{1}, \bG_{n,h_{n}}) 
    &=& \Rem_{20} (Q_{1},  \bG_{n,h_{n}}),
  \end{eqnarray*}
  hence, by \eqref{eq:asymp:exp:intro} and \eqref{eq:hla3:three},
  \begin{align*}
    P_{0} 
    \big(D^{*} (Q_{n,h_{n}}^{*}, \bG_{n,h_{n}}) - 
    & D^{*} (Q_{1}, \bG_{n,h_{n}}) \big) \\ 
    &  =   \Rem_{20}  (Q_{n,h_{n}}^{*},   \bG_{n,h_{n}})  -   \Rem_{20}  (Q_{1},
      \bG_{n,h_{n}}) - \left(\Psi(P_{n,h_{n}}^{*})  - \Psi(P_{0})\right) \\
    &=- (P_{n} - P_{0}) D^{*} (P_{0}) + o_{P}     (1/\sqrt{n}).
  \end{align*}
  But  \eqref{eq:hla2:one}  and  $\Phi_{0}  (\bG_{0}) =  P_{0}  D^{*}  (Q_{1},
  \bG_{0}) = 0$ also imply   that
  \begin{eqnarray*}
  -\Phi_{0} (\bG_{n,h_{n}}) + \Phi_{0} (\bG_{0}) 
  & = & P_{0} \left(D^{*} (Q_{n,h_{n}}^{*}, \bG_{n,h_{n}}) - D^{*} (Q_{1},
        \bG_{n,h_{n}})\right) \\
  && + (P_{n}  - P_{0}) \left(D^{*} (Q_{n,h_{n}}^{*},  \bG_{n,h_{n}}) - D^{*}
     (Q_{1}, \bG_{0})\right)\\ 
  && + (P_{n} - P_{0}) D^{*} (Q_{1}, \bG_{0}) + o_{P} (1/\sqrt{n})
  \end{eqnarray*}
  which, combined with the previous display and \eqref{eq:hla3:one}, yield
  \begin{equation*}
    \Phi_{0} (\bG_{n,h_{n}}) - \Phi_{0} (\bG_{0})  = (P_{n} - P_{0}) \left(D^{*}
      (P_{0})   -  D^{*}   (Q_{1},  \bG_{0})\right)   +  o_{P}(1/\sqrt{n})   =
    o_{P}(1/\sqrt{n}) 
  \end{equation*}
  showing  that  \eqref{eq:hla4}  is  satisfied with  $\th_{n}  =  h_{n}$  and
  $\Delta(P_{1}) \equiv 0$. 
\end{proof}

\subsection{Proof of Theorem~\ref{theo:high:level}}
\label{subsec:theo:high:level}

\begin{proof}
  The proof unfolds in two parts.

  \textit{Step  one:  extracting the  would-be  first  order term.}   Equality
  \eqref{eq:hla2:one} in \textbf{A2} rewrites as
  \begin{eqnarray*}
    o_{P}(1/\sqrt{n})
    &=& P_{n} D^{*} (Q_{n,h_{n}}^{*}, \bG_{n,h_{n}}) \\
    &=& (P_{n} -  P_{0}) D^{*} (Q_{n,h_{n}}^{*}, \bG_{n,h_{n}})  + P_{0} D^{*}
        (Q_{n,h_{n}}^{*}, \bG_{n,h_{n}}) \\ 
    &=&  \left[(P_{n}  - P_{0})  D^{*}  (Q_{1},  \bG_{0})  + (P_{n}  -  P_{0})
        \left(D^{*}   (Q_{n,h_{n}}^{*},   \bG_{n,h_{n}})   -   D^{*}   (Q_{1},
        \bG_{0})\right)\right] \\
    && + \left[P_{0} D^{*} (Q_{n,h_{n}}^{*}, \bG_{0}) + P_{0}
       \left(D^{*} (Q_{n,h_{n}}^{*}, \bG_{n,h_{n}}) - D^{*}(Q_{n,h_{n}}^{*}, 
       \bG_{0})\right)\right]. 
  \end{eqnarray*}
  The  second  term  of  the  sum  between  the  first  pair  of  brackets  is
  $o_{P}(1/\sqrt{n})$ by \eqref{eq:hla3:one} in \textbf{A3}.  As for the first
  term  between the  second pair  of brackets,  \eqref{eq:remainder:intro} and
  \eqref{eq:remainder:bound:intro} entail that it satisfies
  \begin{equation*}
    P_{0}    D^{*}     (Q_{n,h_{n}}^{*},    \bG_{0})    =     \Psi(P_{0})    -
    \Psi(P_{n,h_{n}}^{*}) + \Rem_{20} (Q_{n,h_{n}}^{*}, \bG_{0}) = \Psi(P_{0})    -
    \Psi(P_{n,h_{n}}^{*}),
  \end{equation*}
  where  we  also  use  the   fact  that  $\Psi(P_{n,h_{n}}^{*})$  depends  on
  $P_{n,h_{n}}^{*}$ only through $Q_{n,h_{n}}^{*}$. Thus, it holds that 
  \begin{multline}
    \label{eq:theo:step:one}
    \Psi(P_{n,h_{n}}^{*})  -  \Psi(P_{0})  -  (P_{n} -  P_{0})  D^{*}  (Q_{1},
    \bG_{0})  + o_{P}  (1/\sqrt{n})  \\=  P_{0} \left(D^{*}  (Q_{n,h_{n}}^{*},
      \bG_{n,h_{n}}) - D^{*}(Q_{n,h_{n}}^{*}, \bG_{0})\right) \equiv T_{1,n}.
  \end{multline}

  Let   us  now   study   $T_{1,n}$,  the   right-hand   side  expression   in
  \eqref{eq:theo:step:one}. It rewrites as
  \begin{multline*}
    T_{1,n} = P_{0} \left(D^{*} (Q_{1}, \bG_{n,h_{n}}) - D^{*}(Q_{1},
      \bG_{0})\right) \\
    +   P_{0}    \left(D^{*}   (Q_{n,h_{n}}^{*},   \bG_{n,h_{n}})    -   D^{*}
      (Q_{n,h_{n}}^{*},    \bG_{0})\right)    -    P_{0}    \left(D^{*}(Q_{1},
      \bG_{n,h_{n}}) - D^{*}(Q_{1}, \bG_{0})\right).
  \end{multline*}
  Consider  the three  terms in  the right-hand  side of  the above  equation.
  Combining  \eqref{eq:remainder:intro}, \eqref{eq:remainder:bound:intro}  and
  the  fact that  $\Psi(P) =  \Psi(P')$ whenever  $P$ and  $P'$ have  the same
  $Q$-component   reveals   that  the   first   and   third  terms   equal   both
  $\Phi_{0}(\bG_{n,h_{n}})          -          \Phi_{0}(\bG_{0})$          and
  $\Rem_{20} (Q_{1},  \bG_{n,h_{n}})$.  For  similar reasons, the  second term
  equals  $\Rem_{20} (Q_{n,h_{n}}^{*},  \bG_{n,h_{n}})$.  Therefore,  by using
  successively \eqref{eq:hla3:three} in  \textbf{A3} then \eqref{eq:hla4} from
  \textbf{A4}, we obtain that
  \begin{eqnarray}
    \notag
    T_{1,n}
    &=& \Phi_{0}(\bG_{n,h_{n}})  - \Phi_{0}(\bG_{0})  +  \left(\Rem_{20}
        (Q_{n,h_{n}}^{*},      \bG_{n,h_{n}})     -      \Rem_{20}     (Q_{1},
        \bG_{n,h_{n}})\right)\\ 
    \notag
    &=&    \big[\Phi_{0}(\bG_{n,\th_{n}})     -    \Phi_{0}(\bG_{0})\big]    +
        \big[\Phi_{0}(\bG_{n,h_{n}}) -  \Phi_{0}(\bG_{n,\th_{n}})\big] + o_{P}
        (1/\sqrt{n})\\ 
    \label{eq:theo:step:two} 
    &=& (P_{n} -  P_{0})  \Delta(P_{1})  + \big[\Phi_{0}(\bG_{n,h_{n}})  -
        \Phi_{0}(\bG_{n,\th_{n}})\big] + o_{P} (1/\sqrt{n}).  
  \end{eqnarray}

  \textit{Step two: showing  that the would-be first order  term is complete.}
  The rest of the proof consists in  showing that the term between brackets in
  \eqref{eq:theo:step:two},  say   $T_{2,n}$,  is   $o_{P}(1/\sqrt{n})$.   The
  inequality below follows from the  definition of $\Phi_{0}$ and the triangle
  inequality, and the equality from \eqref{eq:hla5:one} in \textbf{A5}:
  \begin{eqnarray*}
    |T_{2,n}| 
    &\leq&   \left|P_{n}   \left(D^{*}(Q_{1},   \bG_{n,\th_{n}})   -
           D^{*}(Q_{1}, \bG_{n,h_{n}})\right)\right| \\
    &&+ \left| (P_{n} -  P_{0}) \left(D^{*}(Q_{1}, \bG_{n,h_{n}}) - D^{*}(Q_{1},
       \bG_{n,\th_{n}})\right)\right| \\
    &=& \left|P_{n}   \left(D^{*}(Q_{1},   \bG_{n,\th_{n}})   -
        D^{*}(Q_{1}, \bG_{n,h_{n}})\right)\right| + o_{P} (1/\sqrt{n}).
  \end{eqnarray*}
  Therefore,  it suffices  to  prove  that the  absolute  value  in the  above
  right-hand     side     expression    is     $o_{P}(1/\sqrt{n})$.      Under
  \textbf{A1}$(Q_{1},h_{n},c_{5})$  (guaranteed  by  \textbf{A5}),  for  every
  $1 \leq i \leq n$, the Taylor-Lagrange inequality yields
  \begin{equation*}
    \left|\left(D^{*}(Q_{1},   \bG_{n,\th_{n}})   -
        D^{*}(Q_{1},    \bG_{n,h_{n}})   -    (\th_{n}    -   h_{n})    \times
        \partial_{h_{n}}  D^{*}(Q_{1},  \bG_{n,\Cdot})\right)(O_{i})\right|
    \lesssim (\th_{n} - h_{n})^{2}
  \end{equation*}
  hence, by convexity, 
  \begin{equation*}
    \left|P_{n}\left(D^{*}(Q_{1},   \bG_{n,\th_{n}})   -
        D^{*}(Q_{1},    \bG_{n,h_{n}})   -    (\th_{n}    -   h_{n})    \times
        \partial_{h_{n}} D^{*}(Q_{1}, \bG_{n,\Cdot})\right)\right| 
    \lesssim (\th_{n} - h_{n})^{2}.    
  \end{equation*}
  Since  $P_{n}\partial_{h_{n}}   D^{*}(Q_{n,h_{n}}^{*},  \bG_{n,\Cdot})$  and
  $(\th_{n} -  h_{n})$ are  both $o_{P}(1/n^{1/4})$ by  \eqref{eq:hla2:two} in
  \textbf{A2} and \textbf{A5}, we get
  \begin{eqnarray*}
    &&P_{n}     \left(D^{*}(Q_{1},     \bG_{n,\th_{n}})     -     D^{*}(Q_{1},
       \bG_{n,h_{n}})\right) \\ 
    &=&    (\th_{n}    -    h_{n})   \times    P_{n}    \left(\partial_{h_{n}}
        D^{*}(Q_{n,h_{n}}^{*}, \bG_{n,\Cdot}) - \partial_{h_{n}} D^{*}(Q_{1},
        \bG_{n,\Cdot})\right) + o_{P}(1/\sqrt{n})\\
    &=& (\th_{n} - h_{n}) \times (P_{n} - P_{0}) \left(\partial_{h_{n}}
        D^{*}(Q_{n,h_{n}}^{*}, \bG_{n,\Cdot}) - \partial_{h_{n}} D^{*}(Q_{1},
        \bG_{n,\Cdot})\right) \\
    &&   +    (\th_{n}   -   h_{n})   \times    P_{0}   \left(\partial_{h_{n}}
       D^{*}(Q_{n,h_{n}}^{*},  \bG_{n,\Cdot}) -  \partial_{h_{n}} D^{*}(Q_{1},
       \bG_{n,\Cdot})\right) + o_{P}(1/\sqrt{n}). 
  \end{eqnarray*}
  In light  of \eqref{eq:hla5:two}  and \eqref{eq:hla5:three}  in \textbf{A5},
  the righ-hand side  expression is $o_{P} (1/\sqrt{n})$.   This completes the
  proof:  $T_{2,n}  =   o_{P}  (1/\sqrt{n})$,  hence  \eqref{eq:theo:step:two}
  rewrites as
  \begin{equation*}
    T_{1,n} = (P_{n} - P_{0}) \Delta(P_{1}) + o_{P} (1/\sqrt{n}),
  \end{equation*}
  and       \eqref{eq:theo:high:level}        finally       follows       from
  \eqref{eq:theo:step:one}.
\end{proof}

\subsection{Proof of Lemma~\ref{lem:easy}}
\label{subsec:lem:easy}

\begin{proof}
  Set   $\bq_{n}  \equiv   \bQ_{n}   (1,  \cdot)   -   \bQ_{n}  (0,   \cdot)$,
  $\bq_{1}   \equiv  \bQ_{1}   (1,   \cdot)  -   \bQ_{1}   (0,  \cdot)$,   and
  $\psi_{1}      \equiv      P_{0}       \bq_{1}$.       Using      inequality
  $(a+b)^{2} \leq  2 (a^{2} + b^{2})$  (valid for all real  numbers $a,b$), we
  first remark that
  \begin{eqnarray*}
    P_{0} (\bq_{n} - \bq_{1})^{2} 
    &\leq& 2 P_{0} [\bQ_{n}(1,\cdot) - \bQ_{1}(1,\cdot)]^{2} + 2 P_{0}
           [\bQ_{n}(0,\cdot) - \bQ_{1}(0,\cdot)]^{2} \\ 
    &=& 2 P_{0} (\bQ_{n} - \bQ_{0})^{2}/\ell\bG_{0} \lesssim P_{0} (\bQ_{n} -
        \bQ_{0})^{2}. 
  \end{eqnarray*}
  Therefore, it also  holds that $P_{0} (\bq_{n} - \bq_{1})^{2}  = o_{P} (1)$.
  Second, we decompose the difference $\psi_{n} - \psi_{1}$ as
  \begin{equation*}
    \psi_{n} -  \psi_{1} =  P_{n} \bq_{n}  - P_{0} \bq_{1}  = (P_{n}  - P_{0})
    (\bq_{n}  -  \bq_{1})  +  (P_{n}  - P_{0})  \bq_{1}  +  P_{0}  (\bq_{n}  -
    \bq_{1}). 
  \end{equation*}
  Lemma~19.24 in~\citep{VdV98} guarantees that the first term in the above RHS
  expression is $o_{P} (1/\sqrt{n})$.  Because $\bq_{1}$ is uniformly bounded,
  the  standard  central  limit  theorem (for  sequences  of  independent  and
  identically distributed, real-valued random  variables with finite variance)
  implies that  the second  term is $O_{P}  (1/\sqrt{n})$. Finally,  the third
  term  is $o_{P}(1)$  by  the  Cauchy-Schwarz inequality  and  the remark  we
  previously made. In summary, $\psi_{n} - \psi_{1} = o_{P} (1)$, as stated.
\end{proof}

\subsection{Proof of Corollary~\ref{theo:specific}}
\label{subsec:theo:specific}

The proof of Corollary~\ref{theo:specific} uses  repeatedly the fact that some
specific   random    functions   fall   in   $P_{0}$-Donsker    classes   with
$P_{0}$-probability  tending  to  one.   Specifically, the  proof  will  refer
several times to the following lemma (its proof is deferred to the end of this
section).

\begin{lemma}
  \label{lem:Donsker}
  Suppose  that  the  assumptions of  Corollary~\ref{theo:specific}  are  met.
  Then,       with       $P_{0}$-probability       tending       to       one,
  $\bQ_{n,\hnk}^{*}(1,\cdot)           -           \bQ_{n,\hnk}^{*}(0,\cdot)$,
  $D^{*}                                                    (P_{n,\hnk}^{*})$,
  $D^{*}  (Q_{1},   \bG_{n,\hnk})  -   D^{*}  (Q_{1},   \bG_{n,\th_{n}})$  and
  $\partial_{\hnk}  D^{*}  (Q_{n,\hnk}^{*}, \bG_{n,\Cdot})  -  \partial_{\hnk}
  D^{*} (Q_{n,\hnk}^{*}, \bG_{n,\Cdot})$ also fall in $P_{0}$-Donsker classes.
\end{lemma}
We can now present the proof of Corollary~\ref{theo:specific}.
\begin{proof}[Proof of Corollary~\ref{theo:specific}]
  There is no obvious counterpart  in \textbf{C1}, \textbf{C2} and \textbf{C4}
  to   \eqref{eq:hla3:one}   and    \eqref{eq:hla3:three}   appearing   within
  \textbf{A3}.   Yet,  under  \textbf{C2}  and  \textbf{C4},  $\bG_{n,  \hnk}$
  consistently estimates $\bG_{0}$ and $\bQ_{n,\hnk}^{*}$ converges to a limit
  $\bQ_{1}$ that  may differ  from $\bQ_{0}$.  Moreover,  since $\bG_{n,\hnk}$
  and $\bG_{0}$ are bounded away from zero by \textbf{C1}, we have
  \begin{equation}
    \label{eq:cvg:Dstar}
    P_{0} (D^{*} (P_{n,\hnk}^{*}) - D^{*} (P_{1}))^{2} = o_{P} (1),
  \end{equation}
  where $P_{1} \in  \xM$ is any element  of model $\xM$ of which  the $Q$- and
  $\bG$-components  equal $Q_{1}  \equiv (Q_{W,0},  \bQ_{1})$ and  $\bG_{0}$ (see
  proof below).  By Lemma~\ref{lem:Donsker}, $D^{*} (P_{n,\hnk}^{*})$ falls in
  a $P_{0}$-Donsker  class with $P_{0}$-probability  tending to one.   It thus
  holds that
  \begin{equation*}
    (P_{n} - P_{0}) (D^{*} (P_{n,\hnk}^{*}) - D^{*} (P_{1})) = o_{P} (1/\sqrt{n}),
  \end{equation*}
  as requested in \eqref{eq:hla3:one} of \textbf{A3}. In addition, the following
  convergence also occurs,
  \begin{equation}
    \label{eq:hla3:three:spec}
    \Rem_{20}   (\bQ_{n,\hnk}^{*},    \bG_{n,\hnk})   -    \Rem_{20}   (\bQ_{1},
    \bG_{n,\hnk}) = o_{P} (1/\sqrt{n}),
  \end{equation}
  as  requested in  \eqref{eq:hla3:three}  of \textbf{A3}  (see proof  below).
  Consequently, \textbf{C1}, \textbf{C2} and \textbf{C4} imply \textbf{A3}.

  Let  us now  turn to  assumption  \textbf{A5}.  To  alleviate notation,  let
  $\bG_{n,h}''(W)$ be the second order  derivative of $t \mapsto \bG_{n,t}(W)$
  at   $h   \in   \xT$   under   \textbf{C1}.    Given   the   definition   of
  $D^{*}      (Q_{1},      \bG_{n,t})      (O)$,      see      \eqref{eq:EIC},
  $t \mapsto  D^{*} (Q_{1}, \bG_{n,t})  (O)$ is twice differentiable  on $\xT$
  and, for each $h \in \xT$,
  \begin{equation*}
    \partial_{h}^{2}  D^{*} (Q_{1},  \bG_{n,\Cdot}) (O)  = (Y  - \bQ_{1}(A,W))
    \times    \left(\frac{\bG_{n,h}''(W)}{\ell\bG_{n,h}(A,W)^{2}}   -    2(2A-1)
      \frac{\bG_{n,h}'(W)^{2}}{\ell\bG_{n,h}(A,W)^{3}}\right).  
  \end{equation*}
  Obviously, under \textbf{C1}, there exists  a universal constant $C_{3} > 0$
  such      that     the      supremum     in      $h     \in      \xT$     of
  $\partial_{h}^{2} D^{*} (Q_{1}, \bG_{n,\Cdot}) (O)$ is $P_{0}$-almost surely
  smaller        than         $C_{3}$.         Consequently,        assumption
  \textbf{A1}$(\bQ_{1}, \hnk, C_{3})$ is met.  In addition, we show below that
  \eqref{eq:hla5:one}, \eqref{eq:hla5:two} and \eqref{eq:hla5:three} hold true
  whenever \textbf{C1} to \textbf{C4} are met.

  In summary, \textbf{A2} is satisfied by construction of $\psi_{n,\hnk}^{*}$,
  see \eqref{eq:two:eqns:solved:recur};  \textbf{A4} is assumed to  hold true;
  \textbf{A3}  and \textbf{A5}  are met.   Thus, Theorem~\ref{theo:high:level}
  applies and implies the result stated in Corollary~\ref{theo:specific}. This
  completes the proof.
\end{proof}

\begin{proof}[Proof of \eqref{eq:cvg:Dstar}]
  Suppose that  the assumptions  of Corollary~\ref{theo:specific} are  met and
  recall  decomposition~\eqref{eq:EIC}.   To   alleviate  notation,  introduce
  $\bQ_{n}                      \equiv                      \bQ_{n,\hnk}^{*}$,
  $\bq_{n}  \equiv  \bQ_{n,\hnk}^{*}(1,\cdot)- \bQ_{n,\hnk}^{*}(0,\cdot)$  and
  $\bG_{n} \equiv  \bG_{n,\hnk}$. By Lemma~\ref{lem:Donsker},  $\bq_{n}$ falls
  in a  $P_{0}$-Donsker class with  $P_{0}$-probability tending to  one. Using
  repeatedly inequality $(a+b)^{2} \leq 2 (a^{2} + b^{2})$, we obtain
  \begin{eqnarray*}
    P_{0}(D^{*} (P_{n,\hnk}^{*}) - D^{*} (P_{1}))^{2} 
    &\lesssim&      P_{0}      (\ell\bG_{n}     -      \ell\bG_{0})^{2}      +
               P_{0}(\bQ_{n}\ell\bG_{0} - \bQ_{1}\ell\bG_{n})^{2} \\
    && + P_{0}(\bq_{n} - \bq_{1})^{2} + (\psi_{n,\hnk}^{*} - \psi_{0})^{2}\\ 
    &\lesssim& P_{0} (\bG_{n} - \bG_{0})^{2} + P_{0}(\bQ_{n} - \bQ_{1})^{2} \\
    && + P_{0}(\bq_{n} - \bq_{1})^{2} + (\psi_{n,\hnk}^{*} - \psi_{0})^{2}\\ 
    &=& P_{0}(\bq_{n}  - \bq_{1})^{2}  + (\psi_{n,\hnk}^{*} -  \psi_{0})^{2} +
        o_{P} (1). 
  \end{eqnarray*}
  The  assumptions of  Lemma~\ref{lem:easy} are  met too.   Therefore, we  can
  retrieve the bound
  \begin{equation*}
    P_{0}(\bq_{n} - \bq_{1})^{2} \lesssim P_{0}(\bQ_{n} - \bQ_{1})^{2} = o_{P}
    (1) 
  \end{equation*}
  from    its    proof    and    assert    that    its    conclusion    holds:
  $(\psi_{n,\hnk}^{*} -  \psi_{0}) = o_{P}  (1)$. This completes the  proof of
  \eqref{eq:cvg:Dstar}.
\end{proof}

\begin{proof}[Proof of \eqref{eq:hla3:three:spec}]
  Suppose that the assumptions of Corollary~\ref{theo:specific} are met. In view
  of \eqref{eq:remainder}, we have
  \begin{eqnarray*}
    T_{n} 
    &\equiv& |\Rem_{20}   (\bQ_{n,\hnk}^{*},    \bG_{n,\hnk})   -    \Rem_{20}   (\bQ_{1},
             \bG_{n,\hnk})| \\
    &=&           \left|E_{P_{0}}          \left[(2A-1)           \left(1          -
        \frac{\ell\bG_{0}(A,         W)}{\ell\bG_{n,\hnk}(A,        W)}\right)
        (\bQ_{n,\hnk}^{*} (A,W) - \bQ_{1} (A,W))\right]\right|.
  \end{eqnarray*}
  Therefore, the  Cauchy-Schwarz inequality and equality  $(\ell\bG_{n,\hnk} -
  \ell\bG_{0})^{2} = (\bG_{n,\hnk} - \bG_{0})^{2}$ yield
  \begin{equation*}
    T_{n}^{2}   \lesssim  P_{0}   (\bG_{n,\hnk}   -   \bG_{0})^{2}  \times   P_{0}
    (\bQ_{n,\hnk}^{*} - \bQ_{1})^{2} = o_{P} (1/n),
  \end{equation*}
  which completes the proof of \eqref{eq:hla3:three:spec}.
\end{proof}

\begin{proof}[Proof    of    \eqref{eq:hla5:one}    in    the    context    of
  Section~\ref{sec:ctmle_con}] 
  Suppose that the assumptions of Corollary~\ref{theo:specific} are met. Using
  \textbf{C1} and the Taylor-Lagrange inequality yields
  \begin{equation*}
    |\bG_{n,\th_{n}} (W) - \bG_{n,\hnk}(W)| \lesssim |\th_{n} - \hnk| 
  \end{equation*}
  hence
  $P_{0} (\bG_{n,\th_{n}} - \bG_{n,\hnk})^{2}  \lesssim (\th_{n} - \hnk)^{2} =
  o_{P} (1)$ (with much to spare). Now, observe that
  \begin{multline*}
    \left(D^{*} (Q_{1}, \bG_{n,\hnk})  - D^{*} (Q_{1}, \bG_{n,\th_{n}})\right)
    (O)       =      (Y       -       \bQ_{1}      (A,W))       (2A-1)\\\times
    \left(\frac{1}{\ell\bG_{n,\hnk}(A,W)}                                    -
      \frac{1}{\ell\bG_{n,\th_{n}}(A,W)}\right),
  \end{multline*}
  which evidently implies the upper bound
  \begin{eqnarray*}
    \left|\left(D^{*} (Q_{1}, \bG_{n,\hnk})  - D^{*} (Q_{1}, \bG_{n,\th_{n}})\right)
    (O)\right|
    &\lesssim& |\ell\bG_{n,\hnk}(A,W) - \ell\bG_{n,\th_{n}} (A,W)|\\
    &=& |\bG_{n,\hnk}(W)  - \bG_{n,\th_{n}} (W)|.
  \end{eqnarray*}
  Therefore,
  $P_{0} (D^{*} (Q_{1}, \bG_{n,\hnk})  - D^{*} (Q_{1}, \bG_{n,\th_{n}}))^{2} =
  o_{P}   (1)$.    Furthermore,    Lemma~\ref{lem:Donsker}   guarantees   that
  $D^{*} (Q_{1},  \bG_{n,\hnk}) - D^{*}  (Q_{1}, \bG_{n,\th_{n}})$ falls  in a
  $P_{0}$-Donsker  class with  $P_{0}$-probability tending  to one.   The same
  argument as the one that lead to \eqref{eq:asymp:exp:intro}
  in   Section~\ref{subsec:select:uncoop}   thus   completes  the   proof   of
  \eqref{eq:hla5:one}.
\end{proof}

\begin{proof}[Proof    of    \eqref{eq:hla5:two}    in    the    context    of
  Section~\ref{sec:ctmle_con}] 
  Suppose that  the assumptions  of Corollary~\ref{theo:specific} are  met. In
  view of \eqref{eq:EIC}, we have
  \begin{eqnarray}
    \label{eq:partial:star}
    \partial_{\hnk} D^{*} (Q_{n,\hnk}^{*}, \bG_{n,\Cdot}) (O)
    &=& \frac{2A-1}{\ell\bG_{n,\hnk} (A,W)} \bG_{n,\hnk}'(W)
        (Y - \bQ_{n,\hnk}^{*}(A,W)),\\
    \label{eq:partial:one}
    \partial_{\hnk} D^{*} (Q_{1}, \bG_{n,\Cdot}) (O)
    &=& \frac{2A-1}{\ell\bG_{n,\hnk} (A,W)} \bG_{n,\hnk}'(W)
        (Y - \bQ_{1}(A,W)),
  \end{eqnarray}
  hence 
  \begin{equation}
    \label{eq:partial:diff}
    \left|\left(\partial_{\hnk} D^{*} (Q_{n,\hnk}, \bG_{n,\Cdot})
        - \partial_{\hnk} D^{*} (Q_{1}, \bG_{n,\Cdot})\right)\right|
    \lesssim |\bQ_{1} - \bQ_{n,\hnk}^{*}|.
  \end{equation}
  Therefore, the Cauchy-Schwarz inequality implies the bound
  \begin{multline*}
    \left((\hnk   -    \th_{n})   \times    P_{0}\left(\partial_{\hnk}   D^{*}
        (Q_{n,\hnk},   \bG_{n,\Cdot})   -    \partial_{\hnk}   D^{*}   (Q_{1},
        \bG_{n,\Cdot})\right)\right)^{2} \\ 
    \lesssim   (\hnk  -   \th_{n})^{2}   \times   P_{0}  (\bQ_{n,\hnk}^{*}   -
    \bQ_{1})^{2} = o_{P} (1/n), 
  \end{multline*}
  thus completing the proof of \eqref{eq:hla5:two}.
\end{proof}

\begin{proof}[Proof    of    \eqref{eq:hla5:three}    in    the    context    of
  Section~\ref{sec:ctmle_con}] 
  Suppose that  the assumptions  of Corollary~\ref{theo:specific} are  met. By
  Lemma~\ref{lem:Donsker},
  $\partial_{\hnk}  D^{*}  (Q_{n,\hnk}^{*}, \bG_{n,\Cdot})  -  \partial_{\hnk}
  D^{*}  (Q_{1},  \bG_{n,\Cdot})$  falls   in  a  $P_{0}$-Donsker  class  with
  $P_{0}$-probability tending to one.   In view of \eqref{eq:partial:diff}, it
  holds that
  \begin{equation*}
    P_{0}\left(\partial_{\hnk} D^{*} (Q_{n,\hnk}, \bG_{n,\Cdot})
      -  \partial_{\hnk}  D^{*}  (Q_{1}, \bG_{n,\Cdot})\right)^{2}  =  o_{P}
    (1). 
  \end{equation*}
  The same argument as the one that lead to \eqref{eq:asymp:exp:intro}
  in   Section~\ref{subsec:select:uncoop}   thus   completes  the   proof   of
  \eqref{eq:hla5:three}.
\end{proof}

\begin{proof}[Proof of Lemma~\ref{lem:Donsker}]
  We proceed by order of appearance in the statement of the lemma. First, note
  that
  \begin{equation*}
    \bQ_{n,\hnk}^{*}(1,W)      -      \bQ_{n,\hnk}^{*}(0,W)      =      (2A-1)
    \bQ_{n,\hnk}^{*}(A,W). 
  \end{equation*}
  Second,    derive    from    \eqref{eq:EIC}   the    explicit    forms    of
  $D^{*}        (P_{n,\hnk}^{*})$,        $D^{*}(Q_{1},        \bG_{n,\hnk})$,
  $D^{*}(Q_{1}, \bG_{n,\th_{n}})$, and of the difference of the two last ones.
  Third,         recall         the          explicit         forms         of
  $\partial_{\hnk}     D^{*}      (Q_{n,\hnk}^{*},     \bG_{n,\Cdot})$     and
  $D^{*}(Q_{1},   \bG_{n,\hnk})$   given    in   \eqref{eq:partial:star}   and
  \eqref{eq:partial:one},  and  derive from  them  that  of their  difference.
  Thanks  to  \textbf{C4}  and   the  above  explicit  forms,  straightforward
  applications of \citep[Theorem~2.10.6]{vdVW96} yield the result.
\end{proof}

We conclude this article on a  final remark about \As{4}{$P_{n}, k$}.  Suppose
that \As{2}{$P_{n},  k$} is met. If,  in light of \textbf{C1},  we also assume
that $t  \mapsto \bG_{n,t} (W)$ is  twice differentiable in a  neighborhood of
$h_{n,k}$, then the assumptions of the implicit function theorem are satisfied
and $h \mapsto \varepsilon_{n,h,k}$ is differentiable around $h_{n,k}$.

\bibliographystyle{abbrvnat}
\bibliography{references}

\end{document}